\documentclass[12pt,english]{article}
\usepackage[T1]{fontenc}
\usepackage[latin9]{inputenc}
\usepackage{geometry}
\usepackage{array}
\usepackage{float}
\usepackage{url}
\usepackage{bm}
\usepackage{amsmath}
\usepackage{amsthm}
\usepackage{graphicx}
\usepackage{rotfloat}
\usepackage[authoryear]{natbib}

\makeatletter

\providecommand{\tabularnewline}{\\}

\numberwithin{equation}{section}
\newcommand{\lyxaddress}[1]{
	\par {\raggedright #1
	\vspace{1.4em}
	\noindent\par}
}
\theoremstyle{plain}
\newtheorem{thm}{\protect\theoremname}

\@ifundefined{date}{}{\date{}}
\makeatother

\usepackage{babel}
\providecommand{\theoremname}{Theorem}

\begin{document}
\title{A dual-stress Bayesian Weibull accelerated life testing model}
\author{Neill Smit$^{1}$ \thanks{Corresponding author. E-mail: neillsmit1@gmail.com}
and Lizanne Raubenheimer$^{2}$}
\maketitle

\lyxaddress{$^{1}$School of Mathematical and Statistical Sciences, North-West
University, South Africa}

\lyxaddress{$^{2}$Department of Statistics, Rhodes University, South Africa}
\begin{abstract}
In this paper, a Bayesian accelerated life testing model is presented.
The Weibull distribution is used as the life distribution and the
generalised Eyring model as the time transformation function. This
is a model that allows for the use of more than one stressor, whereas
other commonly used acceleration models, such as the Arrhenius and
power law models, allow for only one stressor. The generalised Eyring-Weibull
model is used in an application, where MCMC methods are utilised to
generate samples for posterior inference.

\end{abstract}

\section{Introduction\label{sec:INTRODUCTION}}

Reliability and life testing is vitally important in the manufacturing
of consumer and capital goods, particularly in the engineering fields.
It is crucial to determine whether a product will perform its prescribed
function, without failure, for a desired period of time. Reliability
and life testing provides a theoretical as well as practical framework
by which the life characteristics of a product, component or system
can be quantified. An extensive discussion on the objectives of, and
need for, reliability and life testing can be found in Kececioglu
(2002).

In this day and age, manufacturers are pressured to provide durable
products, timely and at competitive prices. Performing life tests
on high-reliability materials, components and systems can be a tedious
and unproductive task. Many modern products are made to endure for
years or even decades. Obtaining sufficient failure time data for
these products at normal use conditions can prove to be very costly
and time consuming. This longevity obstacle has lead to a rising interest
in reliability engineering and the development of accelerated life
tests (ALTs). In ALTs, products are tested in a more severe than normal
use environment, by applying or fluctuating one or more stressors
(stress variables) at accelerated levels, in order to induce early
failures (Nelson, 1990). This failure data is then extrapolated to
predict the reliability characteristics of the products at normal
use conditions (Pan, 2009).

Stressors can include temperature, pressure, voltage, wattage, humidity,
loads, vibration amplitude, use rate, etc.(Kececioglu, 2002; Escobar
\& Meeker, 2006). Appertaining to the understanding of the physics
of failure, a functional relationship, known as a time transformation
function (TTF), is assumed between the stressors and the parameters
of the life distribution (see, for example, Singpurwalla, 1971a,b,
1973; Singpurwalla et al. , 1975). The most commonly used models for
acceleration include the Arrhenius, inverse power law, Eyring, and
generalised Eyring models. Various other acceleration models are discussed
in Escobar \& Meeker (2006), Kececioglu (2002), and Thiraviam (2010).

In this paper, a Bayesian approach to the generalised Eyring model
with the Weibull distribution as life distribution will be explored.
Section 2 consists of a short review of some prominent Bayesian models
and methods in accelerated life testing. The generalised Eyring-Weibull
(GEW) model and the general likelihood formulation for this model
is given in Section 3. In Section 4, variations of the GEW model is
presented by means of different choices for the prior distributions.
The posterior and full conditional posterior distributions for each
variation, up to at least proportionality, is also provided. Section
5 will provide an extensive application and comparison of the GEW
models to a data set. Due to the mathematically intractable posteriors,
the log-concavity of the models are assessed in the appendix to determine
which Markov Chain Monte Carlo (MCMC) method can be utilised for posterior
sampling. The paper will conclude in Section 6 with some final remarks
on the GEW model.

\section{Bayesian Accelerated Life Testing\label{sec:BAYESIAN_ALT}}

The development of Bayesian methods to draw inferences from ALTs has
experienced a surge in recent years. Ahmad (1990) provides a general
overview of early Bayesian methods in ALTs. These methods include
reducing high dimensional integration problems via semi-sufficient
statistics, the use of semi-parametric inference, a Kalman filter
approach, and failure rate strategies. Achcar \& Louzada-Neto (1991)
considers an exponential distribution with the Eyring model. Type-II
censored data is used and Jeffreys\textquoteright{} priors are employed
for the model parameters, where Laplace approximations are utilised
for integrals that are difficult to solve analytically. Chaloner \&
Larntz (1992) examines experimental Bayesian design for ALTs, where
lifetimes follow either a Weibull or a log-normal distribution. A
dynamic general linear model setup for ALTs, that uses linear Bayesian
methods for inference, is presented in Mazzuchi \& Soyer (1992). Lifetimes
are assumed to be exponential with the power law as the TTF. Achcar
(1993) explores Laplace approximation methods for deriving posterior
densities for various commonly used life distributions and TTFs.

Van Dorp et al. (1996) develops a Bayesian model for step-stress ALTs
where the lifetimes are exponentially distributed, and provides Bayesian
point estimates and credibility intervals for parameters at normal
use conditions. Dietrich \& Mazzuchi (1996) discusses the design of
experiments in ALTs, pointing out some problems that emerge from typical
regression and ANOVA techniques, and proposes an alternative analysis
procedure based on Bayesian considerations. Mazzuchi et al. (1997)
presents a Bayesian approach, based on general linear models, for
inference from ALTs. The lifetimes are assumed to be Weibull distributed
and the power law is used as TTF on the scale parameter. Erkanli \&
Soyer (2000) provides simulation-based designs for ALTs, using a Bayesian
desicion theory approach, where lifetimes follow an exponential distribution
with the power law as TTF. Perdona \& Louzada-Neto (2005) proposes
an ALT model where lifetimes are exponential and a general log-non-linear
TTF is used. A uniform prior is used and Laplace approximations are
utilised to find marginal posterior distributions.

Van Dorp \& Mazzuchi (2004) develops a general Bayes inference model
for ALTs, which is compatible with various types of stress loading,
where the lifetimes are exponentially distributed and strict conformity
to a specific TTF is not compulsory. MCMC methods are used for inference
and to derive posterior quantities. This model is also used in Van
Dorp et al. (2006) to compare constant stress, step-stress and profile
stress ALTs within a single Bayesian inferential framework, and is
extended to the Weibull distribution in Van Dorp \& Mazzuchi (2005).
Le�n et al. (2007) presents Bayesian inferences from ALTs, making
use of MCMC methods, where items originate from different groups with
random effects. Although the procedure is applied to an elementary
example, it can also be used with multiple random effects and acceleration
factors. Barriga et al. (2008) considers Bayesian methods for ALTs
where the exponentiated Weibull is used as the life distribution and
the Arrhenius model is the TTF. MCMC methods are used to sample from
the posterior for inferences, where a mixture of normal and gamma
priors are imposed on the parameters.

Soyer (2008) reviews Bayesian designs for ALTs by comparing Bayesian
decision theory, linear Bayesian, and Bayesian simulation-based designs.
An exponential life distribution with the power law as TTF is used
throughout the review. A basic parametric Bayesian ALT model, where
lifetimes follow a Weibull distribution with the power law as TTF,
is discussed in Soyer et al. (2008), and then extended to allow for
a more dynamic TTF. The two extensions are a hierarchical Bayesian
model and a Markov model where inferences are made via MCMC methods.
Pan (2009) provides a Bayesian approach, and introduces a calibration
factor, which allows for a combination of field data and ALT data
to be used for reliability prediction. A log-linear TTF is used with
the exponential and Weibull life distributions, and MCMC methods are
employed for intricate posterior distributions. Upadhyay \& Mukherjee
(2010) presents a Bayesian comparison between accelerated Weibull
and Birnbaum-Saunders models, where the TTF is the inverse power law.
A combination of independent vague priors for the model parameters
are chosen and MCMC techniques are used to generate posterior samples.

Yuan et al. (2014) proposes a semi-parametric Bayesian approach for
ALTs where the TTF is log-linear and no assumption is made on the
form of the life distribution. The authors employ a Dirichlet process
mixture model, using a Weibull kernel, to model failure times. Model
fitting is performed via a simulation-based algorithm that incorporates
Gibbs sampling. Mukhopadhyay \& Roy (2016) considers Bayesian analysis
for ALTs where log-lifetimes follow a distribution from the log-concave
location-scale family and a linear TTF is used. After a general discussion,
the article focuses on using the log-normal life distribution with
the Arrhenius model as TTF. MCMC techniques are used to obtain posterior
samples for inference. Classical and Bayesian inference on a progressive
stress ALT with a generalised Birnbaum-Saunders life distribution
is presented in Sha (2018). MCMC methods are used for Bayesian analysis
and the author states that the proposed Bayesian method is not only
efficient and accurate, but outperforms the classical likelihood-based
approach to ALTs.

To our knowledge, the generalised Eyring model has primarily been
implemented only in frequentistic ALT setups, and is still less commonly
used than other TTFs. The generalised Eyring model allows for an ALT
model with more than one stressor, which makes it a more applicable
model to use in industry. It is however an intricate model due to
the number of unknown parameters, which complicates both classical
and Bayesian inferences. The aim of this article is to present a Bayesian
ALT model that incorporates the generalised Eyring relationship, where
lifetimes follow a Weibull distribution. The Weibull distribution
is a more flexible and appropriate life distribution to use in practice,
compared to, for example, the exponential distribution. In the event
of mathematically intractable posterior distributions, MCMC methods
are used to draw posterior samples for inference.

\section{The Generalised Eyring-Weibull Model Specification\label{sec:GEW_MODEL}}

Let $X$ be a continuous random variable that follows a Weibull distribution
with scale parameter $\alpha$ and shape parameter $\beta$ $\left(\alpha>0,\beta>0\right)$.
The probability density function (PDF) is given by
\begin{eqnarray}
f\left(x\left|\alpha,\beta\right.\right) & = & \alpha\beta x^{\beta-1}\exp\left(-\alpha x^{\beta}\right),\quad x\geq0,\label{eq:PDF_WEIBULL}
\end{eqnarray}

\noindent and the reliability function by
\begin{equation}
R\left(x\right)=1-F\left(x\right)=\exp\left(-\alpha x^{\beta}\right).\label{eq:RF_WEIBULL}
\end{equation}

\noindent Consider two stressors, one thermal and one non-thermal.
Indicate the $k$ distinct accelerated levels of the stressors by
$\left\{ T_{i},S_{i}\right\} ,i=1,\ldots,k$, where $T_{i},i=1,\ldots,k$,
is the accelerated levels of the thermal stressor and $S_{i},i=1,\ldots,k$
is the accelerated levels of the non-thermal stressor. An item is
exposed to the constant application of a specific stress level combination
$\left\{ T_{i},S_{i}\right\} $. A common assumption in the literature
is that the Weibull scale parameter $\alpha$ is then dependent on
the stress levels, whereas the shape parameter $\beta$ is not (see,
for example, Mazzuchi et al. , 1997; Soyer et al. , 2008; Upadhyay
\& Mukherjee, 2010). The reparameterisation of $\alpha$ given by
the generalised Eyring model is
\begin{eqnarray}
\alpha_{i} & = & T_{i}\exp\left(-\theta_{1}-\frac{\theta_{2}}{T_{i}}-\theta_{3}V_{i}-\frac{\theta_{4}V_{i}}{T_{i}}\right),\label{eq:GEW_ALPHA}
\end{eqnarray}

\noindent where $\theta_{1},$ $\theta_{2},$ $\theta_{3},$ and $\theta_{4}$
are unknown model parameters (ReliaSoft Corporation, 2015), and $V_{i}$
is a function of the non-thermal stressor $S_{i}$(Escobar \& Meeker,
2006). For a lifetime subjected to the $i^{th}$ level of the stressors,
it follows from (\ref{eq:PDF_WEIBULL}) and (\ref{eq:GEW_ALPHA})
that the Weibull PDF can be written as
\begin{eqnarray}
f\left(x_{i}\left|\theta_{1},\theta_{2},\theta_{3},\theta_{4},\beta\right.\right) & = & T_{i}\exp\left(-\theta_{1}-\frac{\theta_{2}}{T_{i}}-\theta_{3}V_{i}-\frac{\theta_{4}V_{i}}{T_{i}}\right)\beta x_{i}^{\beta-1}\nonumber \\
 &  & \times\exp\left[-T_{i}\exp\left(-\theta_{1}-\frac{\theta_{2}}{T_{i}}-\theta_{3}V_{i}-\frac{\theta_{4}V_{i}}{T_{i}}\right)x_{i}^{\beta}\right].\label{eq:GEW_PDF}
\end{eqnarray}
From (\ref{eq:RF_WEIBULL}) and (\ref{eq:GEW_ALPHA}), the Weibull
reliability function at some time $\tau$ can be written as
\begin{equation}
R\left(\tau\right)=\exp\left[-T_{i}\exp\left(-\theta_{1}-\frac{\theta_{2}}{T_{i}}-\theta_{3}V_{i}-\frac{\theta_{4}V_{i}}{T_{i}}\right)\tau^{\beta}\right].\label{eq:GEW_RF}
\end{equation}
Suppose that $n_{i}$ items are tested at each of the $k$ different
stress levels and denote the failure times by $x_{ij},j=1,\ldots,n_{i},i=1,\ldots,k.$
The likelihood function, in general, is then given by
\[
L\left(\underline{x}\left|\theta_{1},\theta_{2},\theta_{3},\theta_{4},\beta\right.\right)=\prod_{i=1}^{k}\left[\prod_{j=1}^{r_{i}}f\left(x_{ij}\left|\theta_{1},\theta_{2},\theta_{3},\theta_{4},\beta\right.\right)\right]\left[R\left(\tau_{i}\right)\right]^{n_{i}-r_{i}}.
\]
Note that for complete samples $r_{i}=n_{i}$. For type-I censoring
$r_{i}$ is the number of failures that occur before censoring time
$\tau_{i}$, where $\tau_{i}<\infty,i=1,\ldots,k$ is predetermined
censoring times for the $k$ different stress levels. For type-II
censoring $\tau_{i}=x_{i(r_{i})}$, where $x_{i(r_{i})}$ is the failure
time of the $r_{i}^{th}$ failure, where $r_{i},i=1,\ldots,k$ is
the pre-chosen number of failures after which censoring occurs for
the $k$ different stress levels. From (\ref{eq:GEW_PDF}) and (\ref{eq:GEW_RF})
it follows that

\noindent 
\begin{eqnarray}
L\left(\underline{x}\left|\theta_{1},\theta_{2},\theta_{3},\theta_{4},\beta\right.\right) & = & \prod_{i=1}^{k}\left[\prod_{j=1}^{r_{i}}f\left(x_{ij}\left|\theta_{1},\theta_{2},\theta_{3},\theta_{4},\beta\right.\right)\right]\left[R\left(\tau_{i}\right)\right]^{n_{i}-r_{i}}\nonumber \\
 & = & \prod_{i=1}^{k}\exp\left[-\left(n_{i}-r_{i}\right)T_{i}\exp\left(-\theta_{1}-\frac{\theta_{2}}{T_{i}}-\theta_{3}V_{i}-\frac{\theta_{4}V_{i}}{T_{i}}\right)\tau_{i}^{\beta}\right]\nonumber \\
 &  & \times\beta^{r_{i}}T_{i}^{r_{i}}\exp\left(-\theta_{1}r_{i}-\frac{\theta_{2}r_{i}}{T_{i}}-\theta_{3}r_{i}V_{i}-\frac{\theta_{4}r_{i}V_{i}}{T_{i}}\right)\nonumber \\
 &  & \times\prod_{j=1}^{r_{i}}x_{ij}^{\beta-1}\exp\left[-T_{i}\exp\left(-\theta_{1}-\frac{\theta_{2}}{T_{i}}-\theta_{3}V_{i}-\frac{\theta_{4}V_{i}}{T_{i}}\right)x_{ij}^{\beta}\right]\nonumber \\
 & = & \beta^{\sum_{i=1}^{k}r_{i}}\exp\left(-\theta_{1}\sum_{i=1}^{k}r_{i}-\theta_{2}\sum_{i=1}^{k}\frac{r_{i}}{T_{i}}-\theta_{3}\sum_{i=1}^{k}r_{i}V_{i}-\theta_{4}\sum_{i=1}^{k}\frac{r_{i}V_{i}}{T_{i}}\right)\nonumber \\
 &  & \times\exp\left[-\sum_{i=1}^{k}\left(n_{i}-r_{i}\right)T_{i}\exp\left(-\theta_{1}-\frac{\theta_{2}}{T_{i}}-\theta_{3}V_{i}-\frac{\theta_{4}V_{i}}{T_{i}}\right)\tau_{i}^{\beta}\right]\label{eq:GEW_LIKELIHOOD}\\
 &  & \times\exp\left[-\sum_{i=1}^{k}\sum_{j=1}^{r_{i}}T_{i}\exp\left(-\theta_{1}-\frac{\theta_{2}}{T_{i}}-\theta_{3}V_{i}-\frac{\theta_{4}V_{i}}{T_{i}}\right)x_{ij}^{\beta}\right]\left[\prod_{i=1}^{k}\prod_{j=1}^{r_{i}}T_{i}x_{ij}^{\beta-1}\right].\nonumber 
\end{eqnarray}

\section{Priors and Posterior\label{sec:PRIORS_POSTERIORS}}

In this section, a number of GEW models are defined, resulting from
the selection of different prior distributions. Soland (1969) affirms
that there is no conjugate family of continuous joint prior distributions
for the two-parameter Weibull distribution, and proposes a gamma prior
for the scale parameter and a discrete prior for the shape parameter.
Tsokos (1972) suggests using an inverse gamma prior for the scale
parameter and a uniform prior for the shape parameter. A compound
inverse gamma prior on the scale parameter, with either a discrete,
inverse gamma, or uniform prior on the shape parameter is proposed
in Papadopoulos \& Tsokos (1976). Kundu (2008) assumes a gamma prior
for the scale parameter and a non-specific log-concave prior density
with support $(0,\infty)$ for the shape parameter. Banerjee \& Kundu
(2008) considers gamma priors on both the scale and shape parameters.

Taking the above into consideration, we formulate four GEW models.
Assume that the priors on the unknown parameters $\theta_{1},$ $\theta_{2},$
$\theta_{3},$ $\theta_{4}$ and $\beta$ are independent, and given
by
\[
\pi\left(\theta_{1},\theta_{2},\theta_{3},\theta_{4},\beta\right)=\pi\left(\theta_{1}\right)\pi\left(\theta_{2}\right)\pi\left(\theta_{3}\right)\pi\left(\theta_{4}\right)\pi\left(\beta\right).
\]
The joint posterior distribution is then given by
\begin{eqnarray*}
\pi\left(\theta_{1},\theta_{2},\theta_{3},\theta_{4},\beta\left|\underline{x}\right.\right) & \propto & L\left(\underline{x}\left|\theta_{1},\theta_{2},\theta_{3},\theta_{4},\beta\right.\right)\pi\left(\theta_{1},\theta_{2},\theta_{3},\theta_{4},\beta\right).
\end{eqnarray*}

\subsection{$GEW_{1}$ Model\label{subsec:GEW1(UUUUU)}}

The GEW model where uniform priors are imposed on all the parameters,
thus
\begin{eqnarray*}
\theta_{1}\sim U(c_{0},c_{1}) & ,c_{1}>c_{0} & ,\pi_{1}\left(\theta_{1}\right)\propto\textrm{constant}\\
\theta_{2}\sim U(c_{2},c_{3}) & ,c_{3}>c_{2} & ,\pi_{1}\left(\theta_{2}\right)\propto\textrm{constant}\\
\theta_{3}\sim U(c_{4},c_{5}) & ,c_{5}>c_{4} & ,\pi_{1}\left(\theta_{3}\right)\propto\textrm{constant}\\
\theta_{4}\sim U(c_{6},c_{7}) & ,c_{7}>c_{6} & ,\pi_{1}\left(\theta_{4}\right)\propto\textrm{constant}\\
\beta\sim U(c_{8},c_{9}) & ,c_{9}>c_{8} & ,\pi_{1}\left(\beta\right)\propto\textrm{constant},
\end{eqnarray*}
is denoted by $GEW_{1}$. The joint prior for $GEW_{1}$ is then given
by
\begin{equation}
\pi_{1}\left(\theta_{1},\theta_{2},\theta_{3},\theta_{4},\beta\right)\propto\textrm{constant},\label{eq:GEW1_PRIORS}
\end{equation}
and using (\ref{eq:GEW_LIKELIHOOD}) and (\ref{eq:GEW1_PRIORS}) the
joint posterior is given by
\begin{eqnarray*}
\pi_{1}\left(\theta_{1},\theta_{2},\theta_{3},\theta_{4},\beta\left|\underline{x}\right.\right) & \propto & \beta^{\sum_{i=1}^{k}r_{i}}\exp\left(-\theta_{1}\sum_{i=1}^{k}r_{i}-\theta_{2}\sum_{i=1}^{k}\frac{r_{i}}{T_{i}}-\theta_{3}\sum_{i=1}^{k}r_{i}V_{i}-\theta_{4}\sum_{i=1}^{k}\frac{r_{i}V_{i}}{T_{i}}\right)\\
 &  & \times\exp\left[-\sum_{i=1}^{k}\left(n_{i}-r_{i}\right)T_{i}\exp\left(-\theta_{1}-\frac{\theta_{2}}{T_{i}}-\theta_{3}V_{i}-\frac{\theta_{4}V_{i}}{T_{i}}\right)\tau_{i}^{\beta}\right]\\
 &  & \times\exp\left[-\sum_{i=1}^{k}\sum_{j=1}^{r_{i}}T_{i}\exp\left(-\theta_{1}-\frac{\theta_{2}}{T_{i}}-\theta_{3}V_{i}-\frac{\theta_{4}V_{i}}{T_{i}}\right)x_{ij}^{\beta}\right]\left[\prod_{i=1}^{k}\prod_{j=1}^{r_{i}}T_{i}x_{ij}^{\beta-1}\right].
\end{eqnarray*}
The posterior is of an intractable form, therefore MCMC methods are
employed in order to draw posterior samples for inferences. The full
conditional posteriors for the $GEW_{1}$ model are given by
\begin{eqnarray*}
\pi_{1}\left(\theta_{1}\left|\underline{x},\theta_{2},\theta_{3},\theta_{4},\beta\right.\right) & \propto & \exp(-\theta_{1}\sum_{i=1}^{k}r_{i})\exp\left[-\sum_{i=1}^{k}\left(n_{i}-r_{i}\right)T_{i}\exp\left(-\theta_{1}-\frac{\theta_{2}}{T_{i}}-\theta_{3}V_{i}-\frac{\theta_{4}V_{i}}{T_{i}}\right)\tau_{i}^{\beta}\right]\\
 &  & \times\exp\left[-\sum_{i=1}^{k}\sum_{j=1}^{r_{i}}T_{i}\exp\left(-\theta_{1}-\frac{\theta_{2}}{T_{i}}-\theta_{3}V_{i}-\frac{\theta_{4}V_{i}}{T_{i}}\right)x_{ij}^{\beta}\right]
\end{eqnarray*}
\begin{eqnarray*}
\pi_{1}\left(\theta_{2}\left|\underline{x},\theta_{1},\theta_{3},\theta_{4},\beta\right.\right) & \propto & \exp\left(-\theta_{2}\sum_{i=1}^{k}\frac{r_{i}}{T_{i}}\right)\exp\left[-\sum_{i=1}^{k}\left(n_{i}-r_{i}\right)T_{i}\exp\left(-\theta_{1}-\frac{\theta_{2}}{T_{i}}-\theta_{3}V_{i}-\frac{\theta_{4}V_{i}}{T_{i}}\right)\tau_{i}^{\beta}\right]\\
 &  & \times\exp\left[-\sum_{i=1}^{k}\sum_{j=1}^{r_{i}}T_{i}\exp\left(-\theta_{1}-\frac{\theta_{2}}{T_{i}}-\theta_{3}V_{i}-\frac{\theta_{4}V_{i}}{T_{i}}\right)x_{ij}^{\beta}\right]
\end{eqnarray*}
\begin{eqnarray*}
\pi_{1}\left(\theta_{3}\left|\underline{x},\theta_{1},\theta_{2},\theta_{4},\beta\right.\right) & \propto & \exp\left(-\theta_{3}\sum_{i=1}^{k}r_{i}V_{i}\right)\exp\left[-\sum_{i=1}^{k}\left(n_{i}-r_{i}\right)T_{i}\exp\left(-\theta_{1}-\frac{\theta_{2}}{T_{i}}-\theta_{3}V_{i}-\frac{\theta_{4}V_{i}}{T_{i}}\right)\tau_{i}^{\beta}\right]\\
 &  & \times\exp\left[-\sum_{i=1}^{k}\sum_{j=1}^{r_{i}}T_{i}\exp\left(-\theta_{1}-\frac{\theta_{2}}{T_{i}}-\theta_{3}V_{i}-\frac{\theta_{4}V_{i}}{T_{i}}\right)x_{ij}^{\beta}\right]
\end{eqnarray*}
\begin{eqnarray*}
\pi_{1}\left(\theta_{4}\left|\underline{x},\theta_{1},\theta_{2},\theta_{3},\beta\right.\right) & \propto & \exp\left(-\theta_{4}\sum_{i=1}^{k}\frac{r_{i}V_{i}}{T_{i}}\right)\exp\left[-\sum_{i=1}^{k}\left(n_{i}-r_{i}\right)T_{i}\exp\left(-\theta_{1}-\frac{\theta_{2}}{T_{i}}-\theta_{3}V_{i}-\frac{\theta_{4}V_{i}}{T_{i}}\right)\tau_{i}^{\beta}\right]\\
 &  & \times\exp\left[-\sum_{i=1}^{k}\sum_{j=1}^{r_{i}}T_{i}\exp\left(-\theta_{1}-\frac{\theta_{2}}{T_{i}}-\theta_{3}V_{i}-\frac{\theta_{4}V_{i}}{T_{i}}\right)x_{ij}^{\beta}\right]
\end{eqnarray*}
\begin{eqnarray*}
\pi_{1}\left(\beta\left|\underline{x},\theta_{1},\theta_{2},\theta_{3},\theta_{4}\right.\right) & \propto & \beta^{\sum_{i=1}^{k}r_{i}}\exp\left[-\sum_{i=1}^{k}\left(n_{i}-r_{i}\right)T_{i}\exp\left(-\theta_{1}-\frac{\theta_{2}}{T_{i}}-\theta_{3}V_{i}-\frac{\theta_{4}V_{i}}{T_{i}}\right)\tau_{i}^{\beta}\right]\\
 &  & \times\exp\left[-\sum_{i=1}^{k}\sum_{j=1}^{r_{i}}T_{i}\exp\left(-\theta_{1}-\frac{\theta_{2}}{T_{i}}-\theta_{3}V_{i}-\frac{\theta_{4}V_{i}}{T_{i}}\right)x_{ij}^{\beta}\right]\left[\prod_{i=1}^{k}\prod_{j=1}^{r_{i}}x_{ij}^{\beta-1}\right].
\end{eqnarray*}

\subsection{$GEW_{2}$ Model\label{subsec:GEW2(GGGGG)}}

Let $GEW_{2}$ denote the GEW model where gamma priors on all the
parameters are assumed, with
\begin{eqnarray*}
\theta_{1}\sim\Gamma(c_{10},c_{11}) & ,c_{10},c_{11}>0 & ,\pi_{2}\left(\theta_{1}\right)\propto\theta_{1}^{c_{10}-1}\exp\left(-c_{11}\theta_{1}\right)\\
\theta_{2}\sim\Gamma(c_{12},c_{13}) & ,c_{12},c_{13}>0 & ,\pi_{2}\left(\theta_{2}\right)\propto\theta_{2}^{c_{12}-1}\exp\left(-c_{13}\theta_{2}\right)\\
\theta_{3}\sim\Gamma(c_{14},c_{15}) & ,c_{14},c_{15}>0 & ,\pi_{2}\left(\theta_{3}\right)\propto\theta_{3}^{c_{14}-1}\exp\left(-c_{15}\theta_{3}\right)\\
\theta_{4}\sim\Gamma(c_{16},c_{17}) & ,c_{16},c_{17}>0 & ,\pi_{2}\left(\theta_{4}\right)\propto\theta_{4}^{c_{16}-1}\exp\left(-c_{17}\theta_{4}\right)\\
\beta\sim\Gamma(c_{18},c_{19}) & ,c_{18},c_{19}>0 & ,\pi_{2}\left(\beta\right)\propto\beta^{c_{18}-1}\exp\left(-c_{19}\beta\right).
\end{eqnarray*}
The joint prior for $GEW_{2}$ is then given by
\begin{equation}
\pi_{2}\left(\theta_{1},\theta_{2},\theta_{3},\theta_{4},\beta\right)\propto\theta_{1}^{c_{10}-1}\theta_{2}^{c_{12}-1}\theta_{3}^{c_{14}-1}\theta_{4}^{c_{16}-1}\beta^{c_{18}-1}\exp\left(-c_{11}\theta_{1}-c_{13}\theta_{2}-c_{15}\theta_{3}-c_{17}\theta_{4}-c_{19}\beta\right),\label{eq:GEW2_PRIORS}
\end{equation}
and using (\ref{eq:GEW_LIKELIHOOD}) and (\ref{eq:GEW2_PRIORS}) the
joint posterior is given by
\begin{eqnarray*}
\pi_{2}\left(\theta_{1},\theta_{2},\theta_{3},\theta_{4},\beta\left|\underline{x}\right.\right) & \propto & \theta_{1}^{c_{10}-1}\theta_{2}^{c_{12}-1}\theta_{3}^{c_{14}-1}\theta_{4}^{c_{16}-1}\beta^{c_{18}-1}\exp\left(-c_{11}\theta_{1}-c_{13}\theta_{2}-c_{15}\theta_{3}-c_{17}\theta_{4}-c_{19}\beta\right)\\
 &  & \times\beta^{\sum_{i=1}^{k}r_{i}}\exp\left(-\theta_{1}\sum_{i=1}^{k}r_{i}-\theta_{2}\sum_{i=1}^{k}\frac{r_{i}}{T_{i}}-\theta_{3}\sum_{i=1}^{k}r_{i}V_{i}-\theta_{4}\sum_{i=1}^{k}\frac{r_{i}V_{i}}{T_{i}}\right)\\
 &  & \times\exp\left[-\sum_{i=1}^{k}\left(n_{i}-r_{i}\right)T_{i}\exp\left(-\theta_{1}-\frac{\theta_{2}}{T_{i}}-\theta_{3}V_{i}-\frac{\theta_{4}V_{i}}{T_{i}}\right)\tau_{i}^{\beta}\right]\\
 &  & \times\exp\left[-\sum_{i=1}^{k}\sum_{j=1}^{r_{i}}T_{i}\exp\left(-\theta_{1}-\frac{\theta_{2}}{T_{i}}-\theta_{3}V_{i}-\frac{\theta_{4}V_{i}}{T_{i}}\right)x_{ij}^{\beta}\right]\left[\prod_{i=1}^{k}\prod_{j=1}^{r_{i}}T_{i}x_{ij}^{\beta-1}\right].
\end{eqnarray*}
Due to the complexity of the posterior, MCMC methods are used to draw
posterior samples to base inferences on. For $GEW_{2}$, the full
conditional posteriors are given by
\begin{eqnarray*}
\pi_{2}\left(\theta_{1}\left|\underline{x},\theta_{2},\theta_{3},\theta_{4},\beta\right.\right) & \propto & \theta_{1}^{c_{10}-1}\exp\left(-c_{11}\theta_{1}\right)\exp(-\theta_{1}\sum_{i=1}^{k}r_{i})\\
 &  & \times\exp\left[-\sum_{i=1}^{k}\left(n_{i}-r_{i}\right)T_{i}\exp\left(-\theta_{1}-\frac{\theta_{2}}{T_{i}}-\theta_{3}V_{i}-\frac{\theta_{4}V_{i}}{T_{i}}\right)\tau_{i}^{\beta}\right]\\
 &  & \times\exp\left[-\sum_{i=1}^{k}\sum_{j=1}^{r_{i}}T_{i}\exp\left(-\theta_{1}-\frac{\theta_{2}}{T_{i}}-\theta_{3}V_{i}-\frac{\theta_{4}V_{i}}{T_{i}}\right)x_{ij}^{\beta}\right]
\end{eqnarray*}
\begin{eqnarray*}
\pi_{2}\left(\theta_{2}\left|\underline{x},\theta_{1},\theta_{3},\theta_{4},\beta\right.\right) & \propto & \theta_{2}^{c_{12}-1}\exp\left(-c_{13}\theta_{2}\right)\exp\left(-\theta_{2}\sum_{i=1}^{k}\frac{r_{i}}{T_{i}}\right)\\
 &  & \times\exp\left[-\sum_{i=1}^{k}\left(n_{i}-r_{i}\right)T_{i}\exp\left(-\theta_{1}-\frac{\theta_{2}}{T_{i}}-\theta_{3}V_{i}-\frac{\theta_{4}V_{i}}{T_{i}}\right)\tau_{i}^{\beta}\right]\\
 &  & \times\exp\left[-\sum_{i=1}^{k}\sum_{j=1}^{r_{i}}T_{i}\exp\left(-\theta_{1}-\frac{\theta_{2}}{T_{i}}-\theta_{3}V_{i}-\frac{\theta_{4}V_{i}}{T_{i}}\right)x_{ij}^{\beta}\right]
\end{eqnarray*}
\begin{eqnarray*}
\pi_{2}\left(\theta_{3}\left|\underline{x},\theta_{1},\theta_{2},\theta_{4},\beta\right.\right) & \propto & \theta_{3}^{c_{14}-1}\exp\left(-c_{15}\theta_{3}\right)\exp\left(-\theta_{3}\sum_{i=1}^{k}r_{i}V_{i}\right)\\
 &  & \times\exp\left[-\sum_{i=1}^{k}\left(n_{i}-r_{i}\right)T_{i}\exp\left(-\theta_{1}-\frac{\theta_{2}}{T_{i}}-\theta_{3}V_{i}-\frac{\theta_{4}V_{i}}{T_{i}}\right)\tau_{i}^{\beta}\right]\\
 &  & \times\exp\left[-\sum_{i=1}^{k}\sum_{j=1}^{r_{i}}T_{i}\exp\left(-\theta_{1}-\frac{\theta_{2}}{T_{i}}-\theta_{3}V_{i}-\frac{\theta_{4}V_{i}}{T_{i}}\right)x_{ij}^{\beta}\right]
\end{eqnarray*}
\begin{eqnarray*}
\pi_{2}\left(\theta_{4}\left|\underline{x},\theta_{1},\theta_{2},\theta_{3},\beta\right.\right) & \propto & \theta_{4}^{c_{16}-1}\exp\left(-c_{17}\theta_{4}\right)\exp\left(-\theta_{4}\sum_{i=1}^{k}\frac{r_{i}V_{i}}{T_{i}}\right)\\
 &  & \times\exp\left[-\sum_{i=1}^{k}\left(n_{i}-r_{i}\right)T_{i}\exp\left(-\theta_{1}-\frac{\theta_{2}}{T_{i}}-\theta_{3}V_{i}-\frac{\theta_{4}V_{i}}{T_{i}}\right)\tau_{i}^{\beta}\right]\\
 &  & \times\exp\left[-\sum_{i=1}^{k}\sum_{j=1}^{r_{i}}T_{i}\exp\left(-\theta_{1}-\frac{\theta_{2}}{T_{i}}-\theta_{3}V_{i}-\frac{\theta_{4}V_{i}}{T_{i}}\right)x_{ij}^{\beta}\right]
\end{eqnarray*}
\begin{eqnarray*}
\pi_{2}\left(\beta\left|\underline{x},\theta_{1},\theta_{2},\theta_{3},\theta_{4}\right.\right) & \propto & \beta^{c_{18}-1}\exp\left(-c_{19}\beta\right)\beta^{\sum_{i=1}^{k}r_{i}}\\
 &  & \times\exp\left[-\sum_{i=1}^{k}\left(n_{i}-r_{i}\right)T_{i}\exp\left(-\theta_{1}-\frac{\theta_{2}}{T_{i}}-\theta_{3}V_{i}-\frac{\theta_{4}V_{i}}{T_{i}}\right)\tau_{i}^{\beta}\right]\\
 &  & \times\exp\left[-\sum_{i=1}^{k}\sum_{j=1}^{r_{i}}T_{i}\exp\left(-\theta_{1}-\frac{\theta_{2}}{T_{i}}-\theta_{3}V_{i}-\frac{\theta_{4}V_{i}}{T_{i}}\right)x_{ij}^{\beta}\right]\left[\prod_{i=1}^{k}\prod_{j=1}^{r_{i}}x_{ij}^{\beta-1}\right].
\end{eqnarray*}

\subsection{$GEW_{3}$ Model\label{subsec:GEW3(UUUUG)}}

Consider a mixture of uniform and gamma priors for the model parameters,
and denote this model by $GEW_{3}$. Let $\beta$ have a gamma prior,
and the other parameters all have uniform priors,
\begin{eqnarray*}
\theta_{1}\sim U(c_{20},c_{21}) & ,c_{21}>c_{20} & ,\pi_{3}\left(\theta_{1}\right)\propto\textrm{constant}\\
\theta_{2}\sim U(c_{22},c_{23}) & ,c_{23}>c_{22} & ,\pi_{3}\left(\theta_{2}\right)\propto\textrm{constant}\\
\theta_{3}\sim U(c_{24},c_{25}) & ,c_{25}>c_{24} & ,\pi_{3}\left(\theta_{3}\right)\propto\textrm{constant}\\
\theta_{4}\sim U(c_{26},c_{27}) & ,c_{27}>c_{26} & ,\pi_{3}\left(\theta_{4}\right)\propto\textrm{constant}\\
\beta\sim\Gamma(c_{28},c_{29}) & ,c_{28},c_{29}>0 & ,\pi_{3}\left(\beta\right)\propto\beta^{c_{28}-1}\exp(-c_{29}\beta).
\end{eqnarray*}
The joint prior for $GEW_{3}$ is then given by
\begin{equation}
\pi_{3}\left(\theta_{1},\theta_{2},\theta_{3},\theta_{4},\beta\right)\propto\beta^{c_{28}-1}\exp\left(-c_{29}\beta\right),\label{eq:GEW3_PRIORS}
\end{equation}
and using (\ref{eq:GEW_LIKELIHOOD}) and (\ref{eq:GEW3_PRIORS}) the
joint posterior is given by
\begin{eqnarray*}
\pi_{3}\left(\theta_{1},\theta_{2},\theta_{3},\theta_{4},\beta\left|\underline{x}\right.\right) & \propto & \beta^{c_{28}-1}\exp\left(-c_{29}\beta\right)\beta^{\sum_{i=1}^{k}r_{i}}\\
 &  & \times\exp\left(-\theta_{1}\sum_{i=1}^{k}r_{i}-\theta_{2}\sum_{i=1}^{k}\frac{r_{i}}{T_{i}}-\theta_{3}\sum_{i=1}^{k}r_{i}V_{i}-\theta_{4}\sum_{i=1}^{k}\frac{r_{i}V_{i}}{T_{i}}\right)\\
 &  & \times\exp\left[-\sum_{i=1}^{k}\left(n_{i}-r_{i}\right)T_{i}\exp\left(-\theta_{1}-\frac{\theta_{2}}{T_{i}}-\theta_{3}V_{i}-\frac{\theta_{4}V_{i}}{T_{i}}\right)\tau_{i}^{\beta}\right]\\
 &  & \times\exp\left[-\sum_{i=1}^{k}\sum_{j=1}^{r_{i}}T_{i}\exp\left(-\theta_{1}-\frac{\theta_{2}}{T_{i}}-\theta_{3}V_{i}-\frac{\theta_{4}V_{i}}{T_{i}}\right)x_{ij}^{\beta}\right]\left[\prod_{i=1}^{k}\prod_{j=1}^{r_{i}}T_{i}x_{ij}^{\beta-1}\right].
\end{eqnarray*}
MCMC methods is used for posterior inference since the posterior is
of an unmanageable form. The full conditional posteriors for this
model can be written as
\begin{eqnarray*}
\pi_{3}\left(\theta_{1}\left|\underline{x},\theta_{2},\theta_{3},\theta_{4},\beta\right.\right) & \propto & \exp(-\theta_{1}\sum_{i=1}^{k}r_{i})\exp\left[-\sum_{i=1}^{k}\left(n_{i}-r_{i}\right)T_{i}\exp\left(-\theta_{1}-\frac{\theta_{2}}{T_{i}}-\theta_{3}V_{i}-\frac{\theta_{4}V_{i}}{T_{i}}\right)\tau_{i}^{\beta}\right]\\
 &  & \times\exp\left[-\sum_{i=1}^{k}\sum_{j=1}^{r_{i}}T_{i}\exp\left(-\theta_{1}-\frac{\theta_{2}}{T_{i}}-\theta_{3}V_{i}-\frac{\theta_{4}V_{i}}{T_{i}}\right)x_{ij}^{\beta}\right]
\end{eqnarray*}
\begin{eqnarray*}
\pi_{3}\left(\theta_{2}\left|\underline{x},\theta_{1},\theta_{3},\theta_{4},\beta\right.\right) & \propto & \exp\left(-\theta_{2}\sum_{i=1}^{k}\frac{r_{i}}{T_{i}}\right)\exp\left[-\sum_{i=1}^{k}\left(n_{i}-r_{i}\right)T_{i}\exp\left(-\theta_{1}-\frac{\theta_{2}}{T_{i}}-\theta_{3}V_{i}-\frac{\theta_{4}V_{i}}{T_{i}}\right)\tau_{i}^{\beta}\right]\\
 &  & \times\exp\left[-\sum_{i=1}^{k}\sum_{j=1}^{r_{i}}T_{i}\exp\left(-\theta_{1}-\frac{\theta_{2}}{T_{i}}-\theta_{3}V_{i}-\frac{\theta_{4}V_{i}}{T_{i}}\right)x_{ij}^{\beta}\right]
\end{eqnarray*}
\begin{eqnarray*}
\pi_{3}\left(\theta_{3}\left|\underline{x},\theta_{1},\theta_{2},\theta_{4},\beta\right.\right) & \propto & \exp\left(-\theta_{3}\sum_{i=1}^{k}r_{i}V_{i}\right)\exp\left[-\sum_{i=1}^{k}\left(n_{i}-r_{i}\right)T_{i}\exp\left(-\theta_{1}-\frac{\theta_{2}}{T_{i}}-\theta_{3}V_{i}-\frac{\theta_{4}V_{i}}{T_{i}}\right)\tau_{i}^{\beta}\right]\\
 &  & \times\exp\left[-\sum_{i=1}^{k}\sum_{j=1}^{r_{i}}T_{i}\exp\left(-\theta_{1}-\frac{\theta_{2}}{T_{i}}-\theta_{3}V_{i}-\frac{\theta_{4}V_{i}}{T_{i}}\right)x_{ij}^{\beta}\right]
\end{eqnarray*}
\begin{eqnarray*}
\pi_{3}\left(\theta_{4}\left|\underline{x},\theta_{1},\theta_{2},\theta_{3},\beta\right.\right) & \propto & \exp\left(-\theta_{4}\sum_{i=1}^{k}\frac{r_{i}V_{i}}{T_{i}}\right)\exp\left[-\sum_{i=1}^{k}\left(n_{i}-r_{i}\right)T_{i}\exp\left(-\theta_{1}-\frac{\theta_{2}}{T_{i}}-\theta_{3}V_{i}-\frac{\theta_{4}V_{i}}{T_{i}}\right)\tau_{i}^{\beta}\right]\\
 &  & \times\exp\left[-\sum_{i=1}^{k}\sum_{j=1}^{r_{i}}T_{i}\exp\left(-\theta_{1}-\frac{\theta_{2}}{T_{i}}-\theta_{3}V_{i}-\frac{\theta_{4}V_{i}}{T_{i}}\right)x_{ij}^{\beta}\right]
\end{eqnarray*}
\begin{eqnarray*}
\pi_{3}\left(\beta\left|\underline{x},\theta_{1},\theta_{2},\theta_{3},\theta_{4}\right.\right) & \propto & \beta^{c_{28}-1}\exp\left(-c_{29}\beta\right)\beta^{\sum_{i=1}^{k}r_{i}}\\
 &  & \times\exp\left[-\sum_{i=1}^{k}\left(n_{i}-r_{i}\right)T_{i}\exp\left(-\theta_{1}-\frac{\theta_{2}}{T_{i}}-\theta_{3}V_{i}-\frac{\theta_{4}V_{i}}{T_{i}}\right)\tau_{i}^{\beta}\right]\\
 &  & \times\exp\left[-\sum_{i=1}^{k}\sum_{j=1}^{r_{i}}T_{i}\exp\left(-\theta_{1}-\frac{\theta_{2}}{T_{i}}-\theta_{3}V_{i}-\frac{\theta_{4}V_{i}}{T_{i}}\right)x_{ij}^{\beta}\right]\left[\prod_{i=1}^{k}\prod_{j=1}^{r_{i}}x_{ij}^{\beta-1}\right].
\end{eqnarray*}

\subsection{$GEW_{4}$ Model\label{subsec:GEW4(GUUUG)}}

Denote by $GEW_{4}$ another model where a different mixture of uniform
and gamma priors are used. Impose gamma priors on the parameters $\theta_{1},\theta_{2},\theta_{3}$
and $\theta_{4}$,and let $\beta$ have a uniform prior,
\begin{eqnarray*}
\theta_{1}\sim\Gamma(c_{30},c_{31}) & ,c_{30},c_{31}>0 & ,\pi_{4}\left(\theta_{1}\right)\propto\theta_{1}^{c_{30}-1}\exp(-c_{31}\theta_{1})\\
\theta_{2}\sim\Gamma(c_{32},c_{33}) & ,c_{32},c_{33}>0 & ,\pi_{4}\left(\theta_{2}\right)\propto\theta_{2}^{c_{32}-1}\exp(-c_{33}\theta_{2})\\
\theta_{3}\sim\Gamma(c_{34},c_{35}) & ,c_{34},c_{35}>0 & ,\pi_{4}\left(\theta_{3}\right)\propto\theta_{3}^{c_{34}-1}\exp(-c_{35}\theta_{3})\\
\theta_{4}\sim\Gamma(c_{36},c_{37}) & ,c_{36},c_{37}>0 & ,\pi_{4}\left(\theta_{4}\right)\propto\theta_{4}^{c_{36}-1}\exp(-c_{37}\theta_{4})\\
\beta\sim U(c_{38},c_{39}) & ,c_{39}>c_{38} & ,\pi_{4}\left(\beta\right)\propto\textrm{constant}.
\end{eqnarray*}
The joint prior for $GEW_{4}$ is then given by
\begin{equation}
\pi_{4}\left(\theta_{1},\theta_{2},\theta_{3},\theta_{4},\beta\right)\propto\theta_{1}^{c_{30}-1}\theta_{2}^{c_{32}-1}\theta_{3}^{c_{34}-1}\theta_{4}^{c_{36}-1}\exp(-c_{31}\theta_{1}-c_{33}\theta_{2}-c_{35}\theta_{3}-c_{37}\theta_{4}),\label{eq:GEW4_PRIORS}
\end{equation}
and using (\ref{eq:GEW_LIKELIHOOD}) and (\ref{eq:GEW4_PRIORS}) the
joint posterior is given by
\begin{eqnarray*}
\pi_{4}\left(\theta_{1},\theta_{2},\theta_{3},\theta_{4},\beta\left|\underline{x}\right.\right) & \propto & \theta_{1}^{c_{30}-1}\theta_{2}^{c_{32}-1}\theta_{3}^{c_{34}-1}\theta_{4}^{c_{36}-1}\exp\left(-c_{31}\theta_{1}-c_{33}\theta_{2}-c_{35}\theta_{3}-c_{37}\theta_{4}\right)\beta^{\sum_{i=1}^{k}r_{i}}\\
 &  & \times\exp\left(-\theta_{1}\sum_{i=1}^{k}r_{i}-\theta_{2}\sum_{i=1}^{k}\frac{r_{i}}{T_{i}}-\theta_{3}\sum_{i=1}^{k}r_{i}V_{i}-\theta_{4}\sum_{i=1}^{k}\frac{r_{i}V_{i}}{T_{i}}\right)\\
 &  & \times\exp\left[-\sum_{i=1}^{k}\left(n_{i}-r_{i}\right)T_{i}\exp\left(-\theta_{1}-\frac{\theta_{2}}{T_{i}}-\theta_{3}V_{i}-\frac{\theta_{4}V_{i}}{T_{i}}\right)\tau_{i}^{\beta}\right]\\
 &  & \times\exp\left[-\sum_{i=1}^{k}\sum_{j=1}^{r_{i}}T_{i}\exp\left(-\theta_{1}-\frac{\theta_{2}}{T_{i}}-\theta_{3}V_{i}-\frac{\theta_{4}V_{i}}{T_{i}}\right)x_{ij}^{\beta}\right]\left[\prod_{i=1}^{k}\prod_{j=1}^{r_{i}}T_{i}x_{ij}^{\beta-1}\right].
\end{eqnarray*}
The posterior is difficult to work with, consequently MCMC methods
are employed to draw posterior samples for inferences. The full conditional
posteriors for the $GEW_{4}$ model is given by
\begin{eqnarray*}
\pi_{4}\left(\theta_{1}\left|\underline{x},\theta_{2},\theta_{3},\theta_{4},\beta\right.\right) & \propto & \theta_{1}^{c_{30}-1}\exp\left(-c_{31}\theta_{1}\right)\exp(-\theta_{1}\sum_{i=1}^{k}r_{i})\\
 &  & \times\exp\left[-\sum_{i=1}^{k}\left(n_{i}-r_{i}\right)T_{i}\exp\left(-\theta_{1}-\frac{\theta_{2}}{T_{i}}-\theta_{3}V_{i}-\frac{\theta_{4}V_{i}}{T_{i}}\right)\tau_{i}^{\beta}\right]\\
 &  & \times\exp\left[-\sum_{i=1}^{k}\sum_{j=1}^{r_{i}}T_{i}\exp\left(-\theta_{1}-\frac{\theta_{2}}{T_{i}}-\theta_{3}V_{i}-\frac{\theta_{4}V_{i}}{T_{i}}\right)x_{ij}^{\beta}\right]
\end{eqnarray*}
\begin{eqnarray*}
\pi_{4}\left(\theta_{2}\left|\underline{x},\theta_{1},\theta_{3},\theta_{4},\beta\right.\right) & \propto & \theta_{2}^{c_{32}-1}\exp(-c_{33}\theta_{2})\exp\left(-\theta_{2}\sum_{i=1}^{k}\frac{r_{i}}{T_{i}}\right)\\
 &  & \times\exp\left[-\sum_{i=1}^{k}\left(n_{i}-r_{i}\right)T_{i}\exp\left(-\theta_{1}-\frac{\theta_{2}}{T_{i}}-\theta_{3}V_{i}-\frac{\theta_{4}V_{i}}{T_{i}}\right)\tau_{i}^{\beta}\right]\\
 &  & \times\exp\left[-\sum_{i=1}^{k}\sum_{j=1}^{r_{i}}T_{i}\exp\left(-\theta_{1}-\frac{\theta_{2}}{T_{i}}-\theta_{3}V_{i}-\frac{\theta_{4}V_{i}}{T_{i}}\right)x_{ij}^{\beta}\right]
\end{eqnarray*}
\begin{eqnarray*}
\pi_{4}\left(\theta_{3}\left|\underline{x},\theta_{1},\theta_{2},\theta_{4},\beta\right.\right) & \propto & \theta_{3}^{c_{34}-1}\exp(-c_{35}\theta_{3})\exp\left(-\theta_{3}\sum_{i=1}^{k}r_{i}V_{i}\right)\\
 &  & \times\exp\left[-\sum_{i=1}^{k}\left(n_{i}-r_{i}\right)T_{i}\exp\left(-\theta_{1}-\frac{\theta_{2}}{T_{i}}-\theta_{3}V_{i}-\frac{\theta_{4}V_{i}}{T_{i}}\right)\tau_{i}^{\beta}\right]\\
 &  & \times\exp\left[-\sum_{i=1}^{k}\sum_{j=1}^{r_{i}}T_{i}\exp\left(-\theta_{1}-\frac{\theta_{2}}{T_{i}}-\theta_{3}V_{i}-\frac{\theta_{4}V_{i}}{T_{i}}\right)x_{ij}^{\beta}\right]
\end{eqnarray*}
\begin{eqnarray*}
\pi_{4}\left(\theta_{4}\left|\underline{x},\theta_{1},\theta_{2},\theta_{3},\beta\right.\right) & \propto & \theta_{4}^{c_{36}-1}\exp(-c_{37}\theta_{4})\exp\left(-\theta_{4}\sum_{i=1}^{k}\frac{r_{i}V_{i}}{T_{i}}\right)\\
 &  & \times\exp\left[-\sum_{i=1}^{k}\left(n_{i}-r_{i}\right)T_{i}\exp\left(-\theta_{1}-\frac{\theta_{2}}{T_{i}}-\theta_{3}V_{i}-\frac{\theta_{4}V_{i}}{T_{i}}\right)\tau_{i}^{\beta}\right]\\
 &  & \times\exp\left[-\sum_{i=1}^{k}\sum_{j=1}^{r_{i}}T_{i}\exp\left(-\theta_{1}-\frac{\theta_{2}}{T_{i}}-\theta_{3}V_{i}-\frac{\theta_{4}V_{i}}{T_{i}}\right)x_{ij}^{\beta}\right]
\end{eqnarray*}
\begin{eqnarray*}
\pi_{4}\left(\beta\left|\underline{x},\theta_{1},\theta_{2},\theta_{3},\theta_{4}\right.\right) & \propto & \beta^{\sum_{i=1}^{k}r_{i}}\exp\left[-\sum_{i=1}^{k}\left(n_{i}-r_{i}\right)T_{i}\exp\left(-\theta_{1}-\frac{\theta_{2}}{T_{i}}-\theta_{3}V_{i}-\frac{\theta_{4}V_{i}}{T_{i}}\right)\tau_{i}^{\beta}\right]\\
 &  & \times\exp\left[-\sum_{i=1}^{k}\sum_{j=1}^{r_{i}}T_{i}\exp\left(-\theta_{1}-\frac{\theta_{2}}{T_{i}}-\theta_{3}V_{i}-\frac{\theta_{4}V_{i}}{T_{i}}\right)x_{ij}^{\beta}\right]\left[\prod_{i=1}^{k}\prod_{j=1}^{r_{i}}x_{ij}^{\beta-1}\right].
\end{eqnarray*}

\section{Application\label{sec:GEW_APPLICATION}}

An ALT data set in ReliaSoft Corporation (2015) is used for the application
of the GEW model. The data relates to failure times (in hours) obtained
from an electronics epoxy packaging ALT, where temperature and relative
humidity are used as the accelerated stressors. The normal use conditions
are $T_{u}=350\text{K}$ and $S_{u}=0.3$. Table \ref{Flo:GEW_PARMS}
contains the specifications for the priors used in the application.
Flat uniform and gamma priors are imposed on all the parameters for
the $GEW_{1}$, $GEW_{2,1}$, $GEW_{2,2}$, $GEW_{3}$ and $GEW_{4}$
models. Subjective gamma priors, all with mean $5$ but different
variances, are chosen for the parameters of the $GEW_{2,3}$, $GEW_{2,4}$
and $GEW_{2,5}$ models.
\begin{table}[H]
\caption{Prior specifications.}
\textbf{\label{Flo:GEW_PARMS}}
\centering{}\medskip{}
\begin{tabular}{|l|>{\centering}m{2.2cm}|>{\centering}m{2.2cm}|>{\centering}m{2.2cm}|>{\centering}m{2.2cm}|>{\centering}m{2.2cm}|}
\hline 
Model & $\theta_{1}$ & $\theta_{2}$ & $\theta_{3}$ & \textbf{$\theta_{4}$} & \textbf{$\beta$}\tabularnewline
\hline 
$GEW_{1}$ & $U(0,100)$ & $U(0,100)$ & $U(0,100)$ & $U(0,100)$ & $U(0,100)$\tabularnewline
$GEW_{2,1}$ & $\varGamma(1,0.00001)$ & $\varGamma(1,0.00001)$ & $\varGamma(1,0.00001)$ & $\varGamma(1,0.00001)$ & $\varGamma(1,0.00001)$\tabularnewline
$GEW_{2,2}$ & $\varGamma(1,0.001)$ & $\varGamma(1,0.001)$ & $\varGamma(1,0.001)$ & $\varGamma(1,0.001)$ & $\varGamma(1,0.001)$\tabularnewline
$GEW_{2,3}$ & $\varGamma(2.5,0.5)$ & $\varGamma(2.5,0.5)$ & $\varGamma(2.5,0.5)$ & $\varGamma(2.5,0.5)$ & $\varGamma(2.5,0.5)$\tabularnewline
$GEW_{2,4}$ & $\varGamma(5,1)$ & $\varGamma(5,1)$ & $\varGamma(5,1)$ & $\varGamma(5,1)$ & $\varGamma(5,1)$\tabularnewline
$GEW_{2,5}$ & $\varGamma(25,5)$ & $\varGamma(25,5)$ & $\varGamma(25,5)$ & $\varGamma(25,5)$ & $\varGamma(25,5)$\tabularnewline
$GEW_{3}$ & $U(0,100)$ & $U(0,100)$ & $U(0,100)$ & $U(0,100)$ & $\varGamma(1,0.001)$\tabularnewline
$GEW_{4}$ & $\varGamma(1,0.001)$ & $\varGamma(1,0.001)$ & $\varGamma(1,0.001)$ & $\varGamma(1,0.001)$ & $U(0,100)$\tabularnewline
\hline 
\end{tabular}
\end{table}
The log-concavity of the full conditional posteriors of the GEW models
is evaluated and given in the appendix. It is shown that $\pi_{z}\left(\theta_{1}\left|\underline{x},\theta_{2},\theta_{3},\theta_{4},\beta\right.\right)$,
$\pi_{z}\left(\theta_{2}\left|\underline{x},\theta_{1},\theta_{3},\theta_{4},\beta\right.\right)$,
$\pi_{z}\left(\theta_{3}\left|\underline{x},\theta_{1},\theta_{2},\theta_{4},\beta\right.\right)$,
$\pi_{z}\left(\theta_{4}\left|\underline{x},\theta_{1},\theta_{2},\theta_{3},\beta\right.\right)$,
and $\pi_{z}\left(\beta\left|\underline{x},\theta_{1},\theta_{2},\theta_{3},\theta_{4}\right.\right)$
for $z=1,2,3,4$ are all log-concave, subject to the conditions $c_{10},c_{12},c_{14},c_{16},c_{30},c_{32},c_{34},c_{36}\geq1$,
and at least one failure occurring.Under these conditions, it is possible
to use the adaptive rejection sampling (ARS) method of Gilks \& Wild
(1992) to sample from the full conditional posteriors at each iteration
of the Gibbs sampler. Alternative conditions for log-concavity can
also be formulated, but these conditions are difficult to implement.
Slice sampling, introduced by Neal (2003), can also be utilised.

Posterior samples are generated for the models using the Bayesian
data analysis software OpenBUGS. A single Markov chain is initiated
for each model with a burn-in of 50000 iterations. The modified Gelman-Rubin
statistic, proposed by Brooks \& Gelman (1998), and trace plots are
used to confirm that the Markov chains for the above models all converge
well in advance of 50000 iterations. Each chain is then run for another
200000 iterations to obtain posterior samples to base inference on.

The deviance information criterion (DIC), proposed by Spiegelhalter
et al. (2002), is the most popular measure that is used to compare
various models in Bayesian ALTs. The DIC not only takes into consideration
the goodness-of-fit for the model, but also penalizes the model for
complexity in terms of overparameterisation. The models with the smaller
DIC values will typically be favoured above models with larger values
for the DIC, but there are other considerations that also need to
be taken into account. For a parameter vector $\bm{\theta}$, with
likelihood function $L\left(\text{\ensuremath{\underline{x}}}\left|\bm{\theta}\right.\right)$,
the deviance can be defined as $D\left(\bm{\theta}\right)=-2\ln\left[L\left(\text{\ensuremath{\underline{x}}}\left|\bm{\theta}\right.\right)\right]$.
Let $\overline{D}$ be the posterior mean of the deviance, and $\hat{D}\left(\bm{\overline{\theta}}\right)=-2\ln\left[L\left(\text{\ensuremath{\underline{x}}}\left|\bm{\overline{\theta}}\right.\right)\right]$,
with $\bm{\overline{\theta}}$ being the posterior mean of $\bm{\theta}$,
be a point estimate for the deviance. The DIC can then be calculated
as $\textrm{DIC}=\overline{D}+p_{D}$, where $p_{D}$ is the effective
number of parameters given by $p_{D}=\overline{D}-\hat{D}\left(\bm{\overline{\theta}}\right)$.

The DIC values and the effective number of parameters for the various
GEW models are given in Table \ref{Flo:GEW_DIC}. The $GEW_{1}$,
$GEW_{2,1}$, $GEW_{2,2}$, $GEW_{3}$ and $GEW_{4}$ models have
very similar DIC values, with the $GEW_{2,1}$ and $GEW_{4}$ models
exhibiting the lowest DIC. The $GEW_{2,3}$, $GEW_{2,4}$ and $GEW_{2,5}$
models show an increasingly worse fit to the data as subjective priors
with smaller variances are implemented. This may be due to the posterior
being dominated by the prior, when the prior variance is very small.
\begin{table}[H]
\caption{Deviance information criterion.}
\textbf{\label{Flo:GEW_DIC}}
\begin{centering}
\medskip{}
\par\end{centering}
\centering{}%
\begin{tabular}{|>{\centering}p{2cm}|>{\centering}m{2cm}|>{\centering}m{2cm}|}
\hline 
Model & DIC & $p_{D}$\tabularnewline
\hline 
$GEW_{1}$ & 228.4 & 2.038\tabularnewline
$GEW_{2,1}$ & 228.2 & 2.033\tabularnewline
$GEW_{2,2}$ & 228.4 & 2.061\tabularnewline
$GEW_{2,3}$ & 231.8 & 1.890\tabularnewline
$GEW_{2,4}$ & 236.0 & 1.705\tabularnewline
$GEW_{2,5}$ & 251.8 & 1.237\tabularnewline
$GEW_{3}$ & 228.3 & 2.048\tabularnewline
$GEW_{4}$ & 228.2 & 2.029\tabularnewline
\hline 
\end{tabular}
\end{table}

The summary statistics for the marginal posterior distributions of
the GEW models are provided in Table \ref{Flo:GEW_STATS}. For the
models with flat priors, that is $GEW_{1}$, $GEW_{2,1}$, $GEW_{2,2}$,
$GEW_{3}$ and $GEW_{4}$, very similar summary statistics are produced.
When subjective priors are employed, one can note that the central
location measures of the marginal posteriors are progressively, as
more certainty is given by the prior, pulled towards the central location
of the prior. As the variances are reduced between the priors imposed
on the $GEW_{2,3}$, $GEW_{2,4}$ and $GEW_{2,5}$ models, the posterior
is dominated to a greater extent by the prior.
\begin{table}[H]
\caption{Summary statistics for the GEW models.}
\textbf{\label{Flo:GEW_STATS}}
\begin{centering}
\medskip{}
\par\end{centering}
\centering{}%
\begin{tabular}{|c|c|>{\centering}m{1.6cm}|>{\centering}m{1.6cm}|>{\centering}m{1.6cm}|>{\centering}m{1.6cm}|>{\centering}m{1.6cm}|}
\hline 
{\footnotesize{}Model} & {\footnotesize{}Parameter} & {\footnotesize{}Mean} & {\footnotesize{}Standard Deviation} & {\footnotesize{}2.5th Percentile} & {\footnotesize{}Median} & {\footnotesize{}97.5th Percentile}\tabularnewline
\hline 
 & {\footnotesize{}$\theta_{1}$} & {\footnotesize{}3.6597} & {\footnotesize{}2.8837} & {\footnotesize{}0.1155} & {\footnotesize{}2.9560} & {\footnotesize{}10.3400}\tabularnewline
 & {\footnotesize{}$\theta_{2}$} & {\footnotesize{}7.2650} & {\footnotesize{}3.4488} & {\footnotesize{}0.6997} & {\footnotesize{}7.4680} & {\footnotesize{}13.6698}\tabularnewline
{\footnotesize{}$GEW_{1}$} & {\footnotesize{}$\theta_{3}$} & {\footnotesize{}0.5989} & {\footnotesize{}0.5781} & {\footnotesize{}0.0157} & {\footnotesize{}0.4247} & {\footnotesize{}2.1460}\tabularnewline
 & {\footnotesize{}$\theta_{4}$} & {\footnotesize{}0.6053} & {\footnotesize{}0.5802} & {\footnotesize{}0.0157} & {\footnotesize{}0.4329} & {\footnotesize{}2.1390}\tabularnewline
 & {\footnotesize{}$\beta$} & {\footnotesize{}1.9474} & {\footnotesize{}0.3621} & {\footnotesize{}1.2880} & {\footnotesize{}1.9300} & {\footnotesize{}2.7100}\tabularnewline
\hline 
 & {\footnotesize{}$\theta_{1}$} & {\footnotesize{}3.1156} & {\footnotesize{}2.5264} & {\footnotesize{}0.0989} & {\footnotesize{}2.5000} & {\footnotesize{}9.2220}\tabularnewline
 & {\footnotesize{}$\theta_{2}$} & {\footnotesize{}7.9024} & {\footnotesize{}3.1888} & {\footnotesize{}1.2710} & {\footnotesize{}8.0900} & {\footnotesize{}13.7800}\tabularnewline
{\footnotesize{}$GEW_{2,1}$} & {\footnotesize{}$\theta_{3}$} & {\footnotesize{}0.5940} & {\footnotesize{}0.5727} & {\footnotesize{}0.0156} & {\footnotesize{}0.4235} & {\footnotesize{}2.1210}\tabularnewline
 & {\footnotesize{}$\theta_{4}$} & {\footnotesize{}0.6062} & {\footnotesize{}0.5804} & {\footnotesize{}0.0160} & {\footnotesize{}0.4326} & {\footnotesize{}2.1520}\tabularnewline
 & {\footnotesize{}$\beta$} & {\footnotesize{}1.9659} & {\footnotesize{}0.3584} & {\footnotesize{}1.3250} & {\footnotesize{}1.9440} & {\footnotesize{}2.7200}\tabularnewline
\hline 
 & {\footnotesize{}$\theta_{1}$} & {\footnotesize{}3.3681} & {\footnotesize{}2.7709} & {\footnotesize{}0.1046} & {\footnotesize{}2.6440} & {\footnotesize{}10.0200}\tabularnewline
 & {\footnotesize{}$\theta_{2}$} & {\footnotesize{}7.7020} & {\footnotesize{}3.3645} & {\footnotesize{}0.8950} & {\footnotesize{}7.8670} & {\footnotesize{}13.8800}\tabularnewline
{\footnotesize{}$GEW_{2,2}$} & {\footnotesize{}$\theta_{3}$} & {\footnotesize{}0.5870} & {\footnotesize{}0.5664} & {\footnotesize{}0.0156} & {\footnotesize{}0.4165} & {\footnotesize{}2.1010}\tabularnewline
 & {\footnotesize{}$\theta_{4}$} & {\footnotesize{}0.5973} & {\footnotesize{}0.5698} & {\footnotesize{}0.0156} & {\footnotesize{}0.4270} & {\footnotesize{}2.1090}\tabularnewline
 & {\footnotesize{}$\beta$} & {\footnotesize{}1.9722} & {\footnotesize{}0.3571} & {\footnotesize{}1.3280} & {\footnotesize{}1.9570} & {\footnotesize{}2.7190}\tabularnewline
\hline 
 & {\footnotesize{}$\theta_{1}$} & {\footnotesize{}3.8210} & {\footnotesize{}1.9465} & {\footnotesize{}0.7810} & {\footnotesize{}3.5780} & {\footnotesize{}8.1060}\tabularnewline
 & {\footnotesize{}$\theta_{2}$} & {\footnotesize{}5.7987} & {\footnotesize{}2.3221} & {\footnotesize{}1.5860} & {\footnotesize{}5.7200} & {\footnotesize{}10.4900}\tabularnewline
{\footnotesize{}$GEW_{2,3}$} & {\footnotesize{}$\theta_{3}$} & {\footnotesize{}1.1275} & {\footnotesize{}0.6856} & {\footnotesize{}0.1944} & {\footnotesize{}0.9976} & {\footnotesize{}2.8000}\tabularnewline
 & {\footnotesize{}$\theta_{4}$} & {\footnotesize{}1.1457} & {\footnotesize{}0.6974} & {\footnotesize{}0.2005} & {\footnotesize{}1.0110} & {\footnotesize{}2.8460}\tabularnewline
 & {\footnotesize{}$\beta$} & {\footnotesize{}1.7896} & {\footnotesize{}0.3055} & {\footnotesize{}1.2260} & {\footnotesize{}1.7780} & {\footnotesize{}2.4260}\tabularnewline
\hline 
 & {\footnotesize{}$\theta_{1}$} & {\footnotesize{}4.2772} & {\footnotesize{}1.5957} & {\footnotesize{}1.5950} & {\footnotesize{}4.1380} & {\footnotesize{}7.7670}\tabularnewline
 & {\footnotesize{}$\theta_{2}$} & {\footnotesize{}5.3606} & {\footnotesize{}1.8390} & {\footnotesize{}2.1520} & {\footnotesize{}5.2430} & {\footnotesize{}9.2860}\tabularnewline
{\footnotesize{}$GEW_{2,4}$} & {\footnotesize{}$\theta_{3}$} & {\footnotesize{}1.7371} & {\footnotesize{}0.7514} & {\footnotesize{}0.5791} & {\footnotesize{}1.6330} & {\footnotesize{}3.4850}\tabularnewline
 & {\footnotesize{}$\theta_{4}$} & {\footnotesize{}1.7620} & {\footnotesize{}0.7624} & {\footnotesize{}0.5882} & {\footnotesize{}1.6570} & {\footnotesize{}3.5270}\tabularnewline
 & {\footnotesize{}$\beta$} & {\footnotesize{}1.8549} & {\footnotesize{}0.2821} & {\footnotesize{}1.3450} & {\footnotesize{}1.8410} & {\footnotesize{}2.4490}\tabularnewline
\hline 
 & {\footnotesize{}$\theta_{1}$} & {\footnotesize{}5.1327} & {\footnotesize{}0.9165} & {\footnotesize{}3.4680} & {\footnotesize{}5.0930} & {\footnotesize{}7.0610}\tabularnewline
 & {\footnotesize{}$\theta_{2}$} & {\footnotesize{}5.4545} & {\footnotesize{}0.9607} & {\footnotesize{}3.7000} & {\footnotesize{}5.4120} & {\footnotesize{}7.4530}\tabularnewline
{\footnotesize{}$GEW_{2,5}$} & {\footnotesize{}$\theta_{3}$} & {\footnotesize{}3.5383} & {\footnotesize{}0.6988} & {\footnotesize{}2.3040} & {\footnotesize{}3.4920} & {\footnotesize{}5.0300}\tabularnewline
 & {\footnotesize{}$\theta_{4}$} & {\footnotesize{}3.5577} & {\footnotesize{}0.7038} & {\footnotesize{}2.3180} & {\footnotesize{}3.5100} & {\footnotesize{}5.0630}\tabularnewline
 & {\footnotesize{}$\beta$} & {\footnotesize{}2.1894} & {\footnotesize{}0.1884} & {\footnotesize{}1.8310} & {\footnotesize{}2.1860} & {\footnotesize{}2.5710}\tabularnewline
\hline 
 & {\footnotesize{}$\theta_{1}$} & {\footnotesize{}3.2365} & {\footnotesize{}2.5657} & {\footnotesize{}0.1039} & {\footnotesize{}2.6390} & {\footnotesize{}9.3070}\tabularnewline
 & {\footnotesize{}$\theta_{2}$} & {\footnotesize{}7.8247} & {\footnotesize{}3.3255} & {\footnotesize{}1.1760} & {\footnotesize{}7.9700} & {\footnotesize{}14.1600}\tabularnewline
{\footnotesize{}$GEW_{3}$} & {\footnotesize{}$\theta_{3}$} & {\footnotesize{}0.5883} & {\footnotesize{}0.5689} & {\footnotesize{}0.0157} & {\footnotesize{}0.4169} & {\footnotesize{}2.1120}\tabularnewline
 & {\footnotesize{}$\theta_{4}$} & {\footnotesize{}0.5999} & {\footnotesize{}0.5744} & {\footnotesize{}0.0158} & {\footnotesize{}0.4287} & {\footnotesize{}2.1280}\tabularnewline
 & {\footnotesize{}$\beta$} & {\footnotesize{}1.9717} & {\footnotesize{}0.3619} & {\footnotesize{}1.3290} & {\footnotesize{}1.9500} & {\footnotesize{}2.7430}\tabularnewline
\hline 
 & {\footnotesize{}$\theta_{1}$} & {\footnotesize{}3.3114} & {\footnotesize{}2.6794} & {\footnotesize{}0.1083} & {\footnotesize{}2.6300} & {\footnotesize{}9.7870}\tabularnewline
 & {\footnotesize{}$\theta_{2}$} & {\footnotesize{}7.6783} & {\footnotesize{}3.3530} & {\footnotesize{}0.9724} & {\footnotesize{}7.8380} & {\footnotesize{}13.9900}\tabularnewline
{\footnotesize{}$GEW_{4}$} & {\footnotesize{}$\theta_{3}$} & {\footnotesize{}0.5929} & {\footnotesize{}0.5674} & {\footnotesize{}0.0155} & {\footnotesize{}0.4246} & {\footnotesize{}2.0970}\tabularnewline
 & {\footnotesize{}$\theta_{4}$} & {\footnotesize{}0.6064} & {\footnotesize{}0.5846} & {\footnotesize{}0.0161} & {\footnotesize{}0.4325} & {\footnotesize{}2.1690}\tabularnewline
 & {\footnotesize{}$\beta$} & {\footnotesize{}1.9600} & {\footnotesize{}0.3654} & {\footnotesize{}1.3000} & {\footnotesize{}1.9370} & {\footnotesize{}2.7410}\tabularnewline
\hline 
\end{tabular}
\end{table}

The marginal posterior distributions for the GEW models are shown
in Figure \ref{Flo:GEW_MARGINAL}. Again, it can be observed that
the $GEW_{1}$, $GEW_{2,1}$, $GEW_{2,2}$, $GEW_{3}$ and $GEW_{4}$
models produce very similar marginal posteriors. The marginal posteriors
of the $GEW_{2,3}$, $GEW_{2,4}$ and $GEW_{2,5}$ models show how
the density is increasingly concentrated towards the central location
of the priors as more prior certainty is conveyed by means of smaller
prior variances.

\begin{sidewaysfigure}[H]
\includegraphics[width=4.8cm,height=2.2cm]{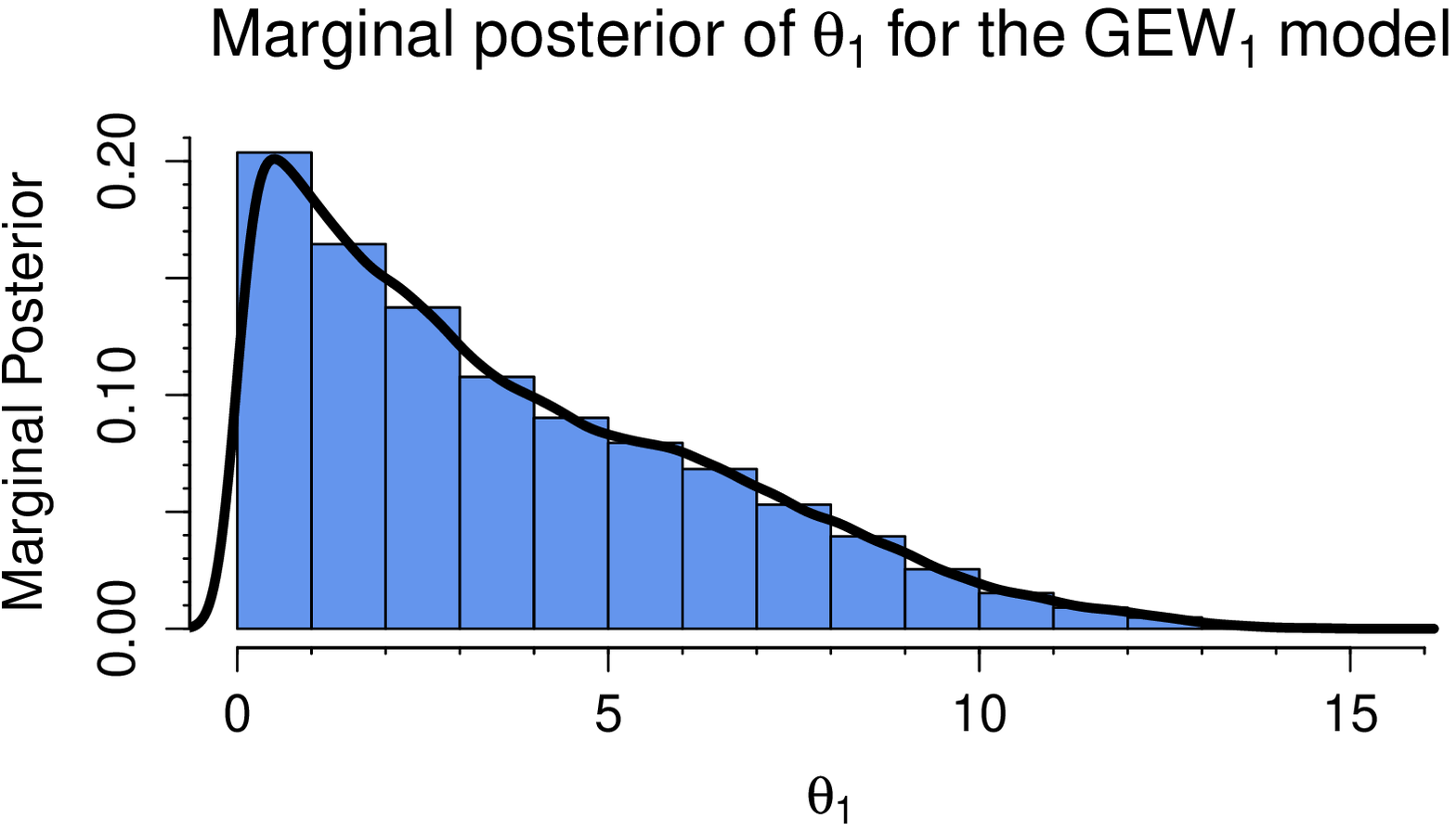}\includegraphics[width=4.8cm,height=2.2cm]{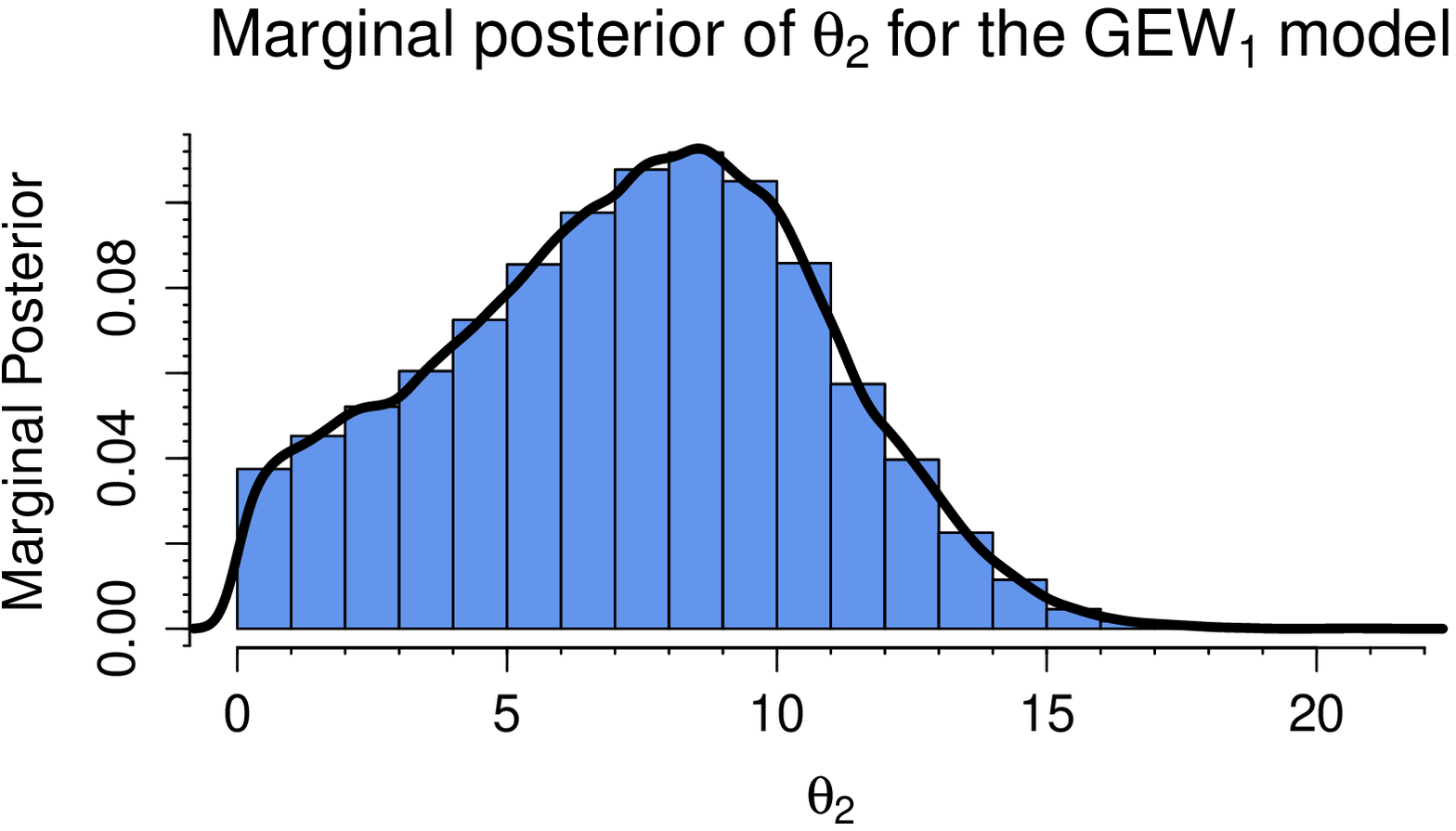}\includegraphics[width=4.8cm,height=2.2cm]{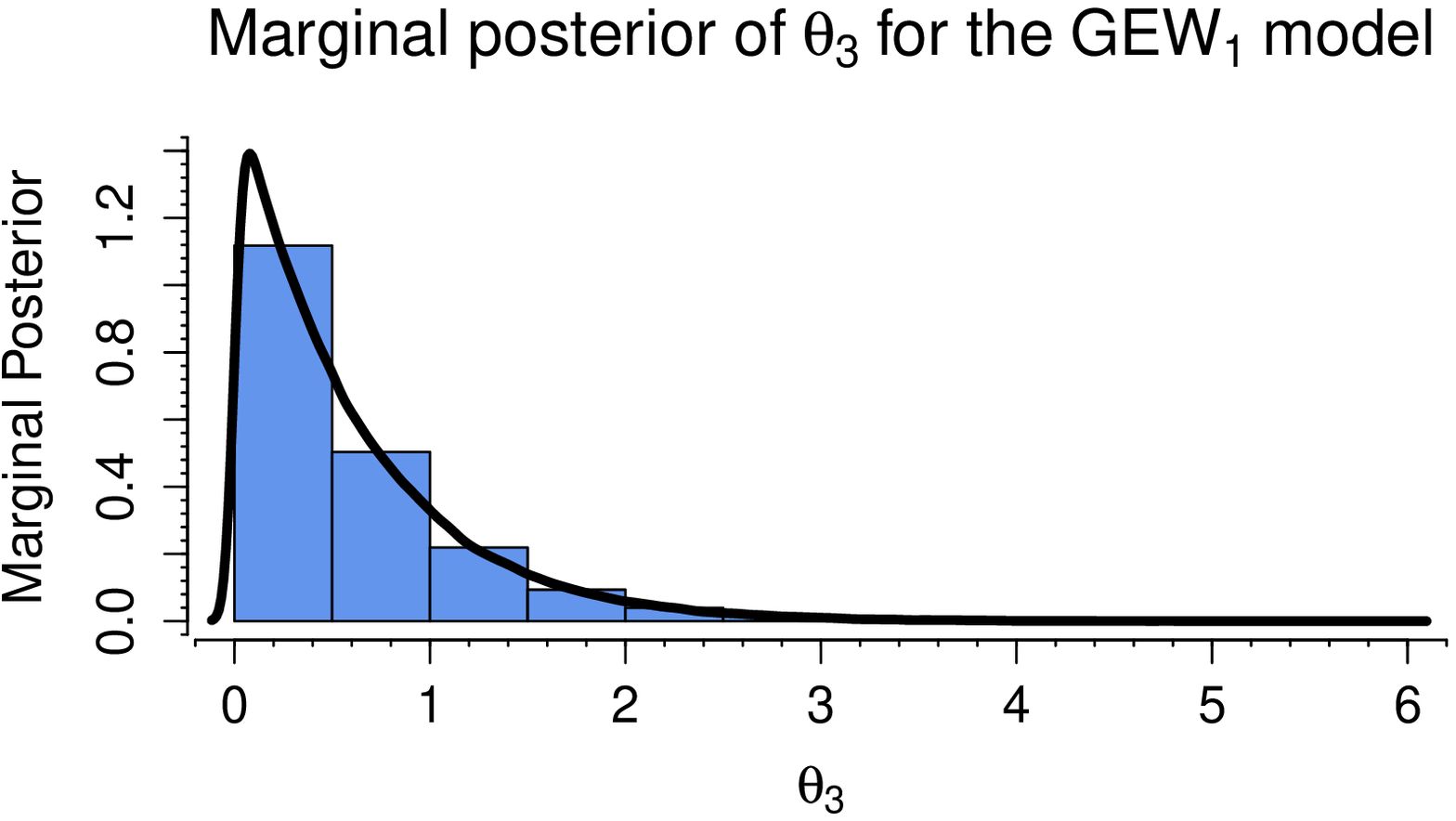}\includegraphics[width=4.8cm,height=2.2cm]{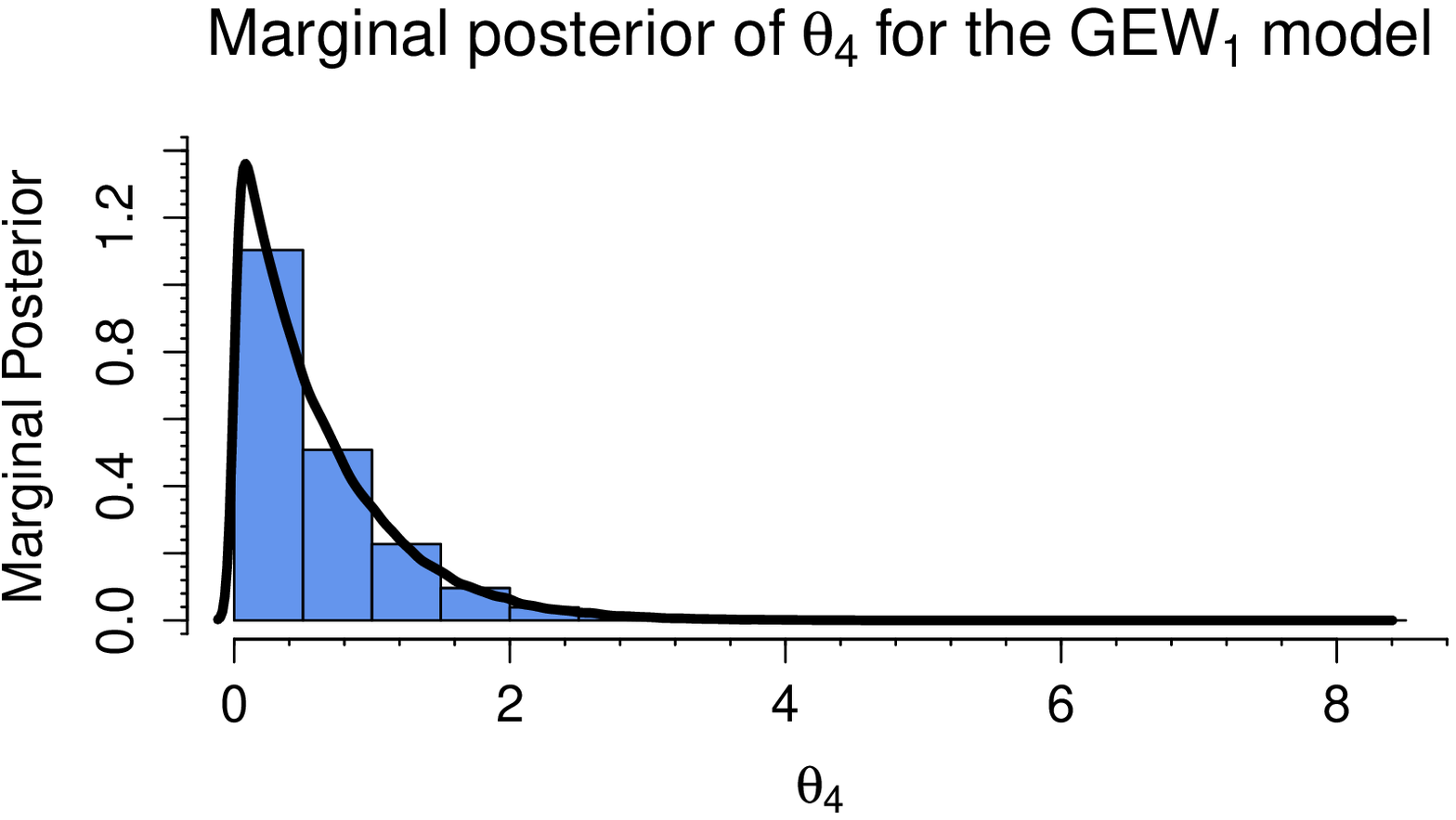}\includegraphics[width=4.8cm,height=2.2cm]{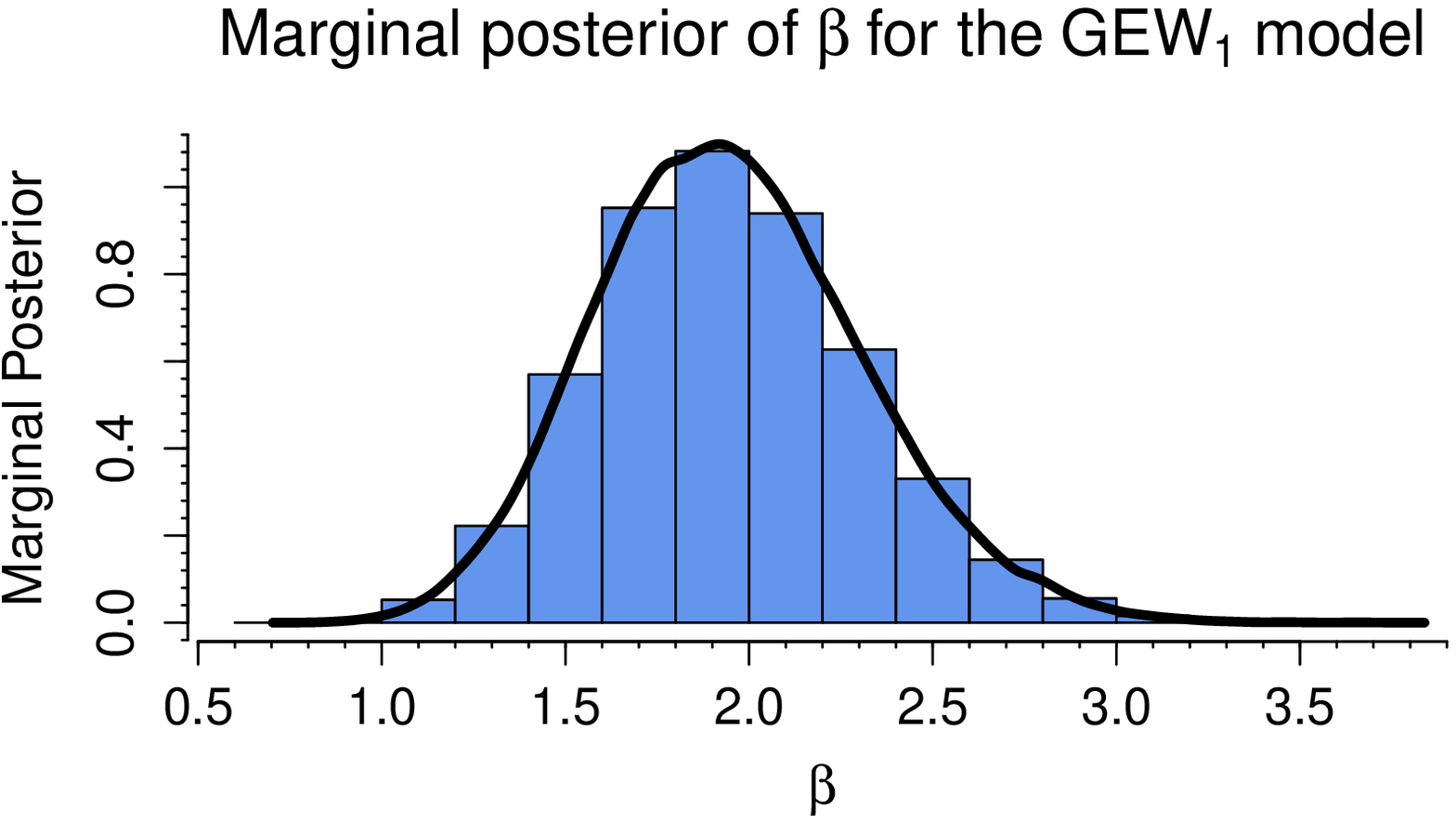}

\includegraphics[width=4.8cm,height=2.2cm]{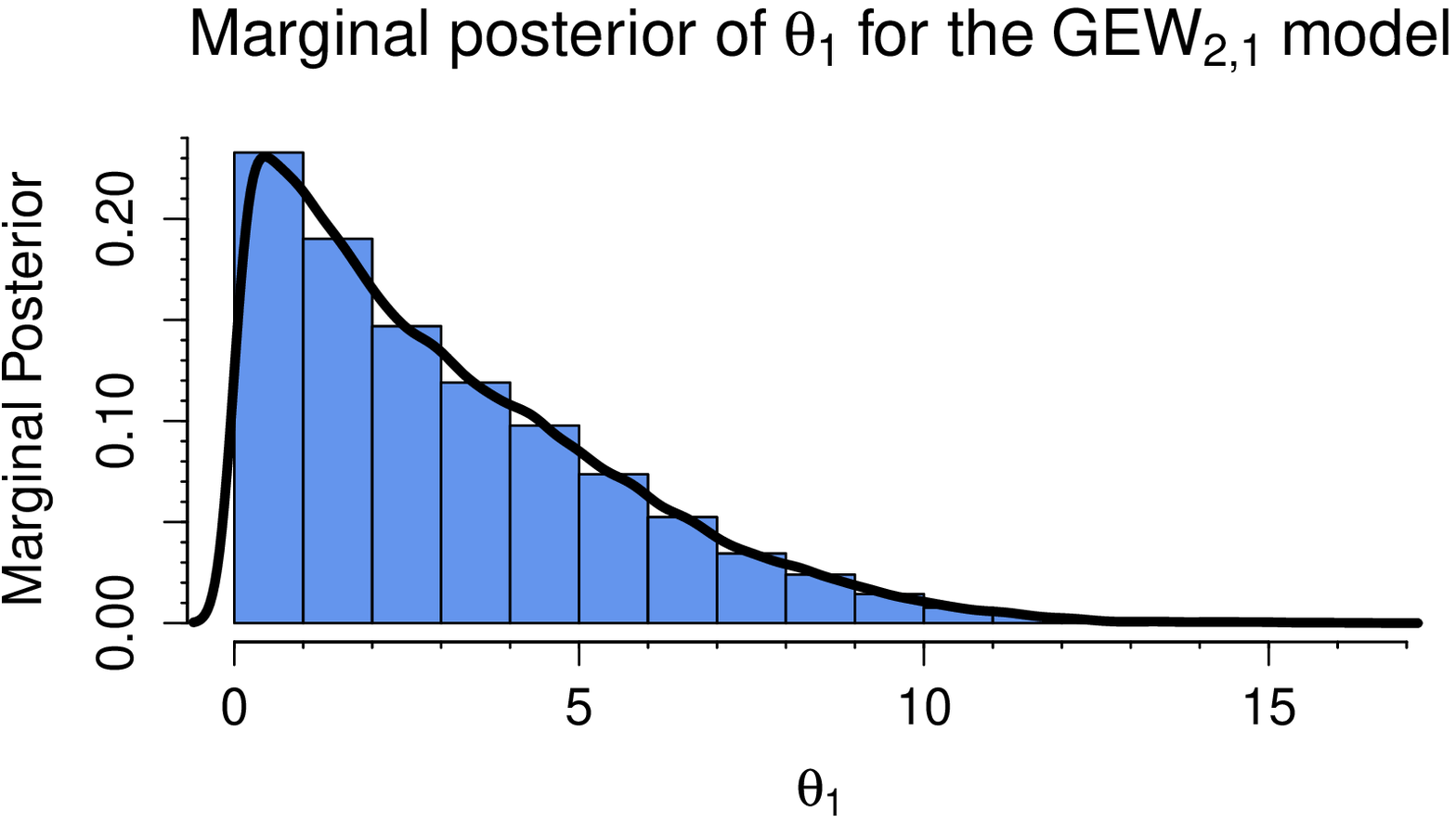}\includegraphics[width=4.8cm,height=2.2cm]{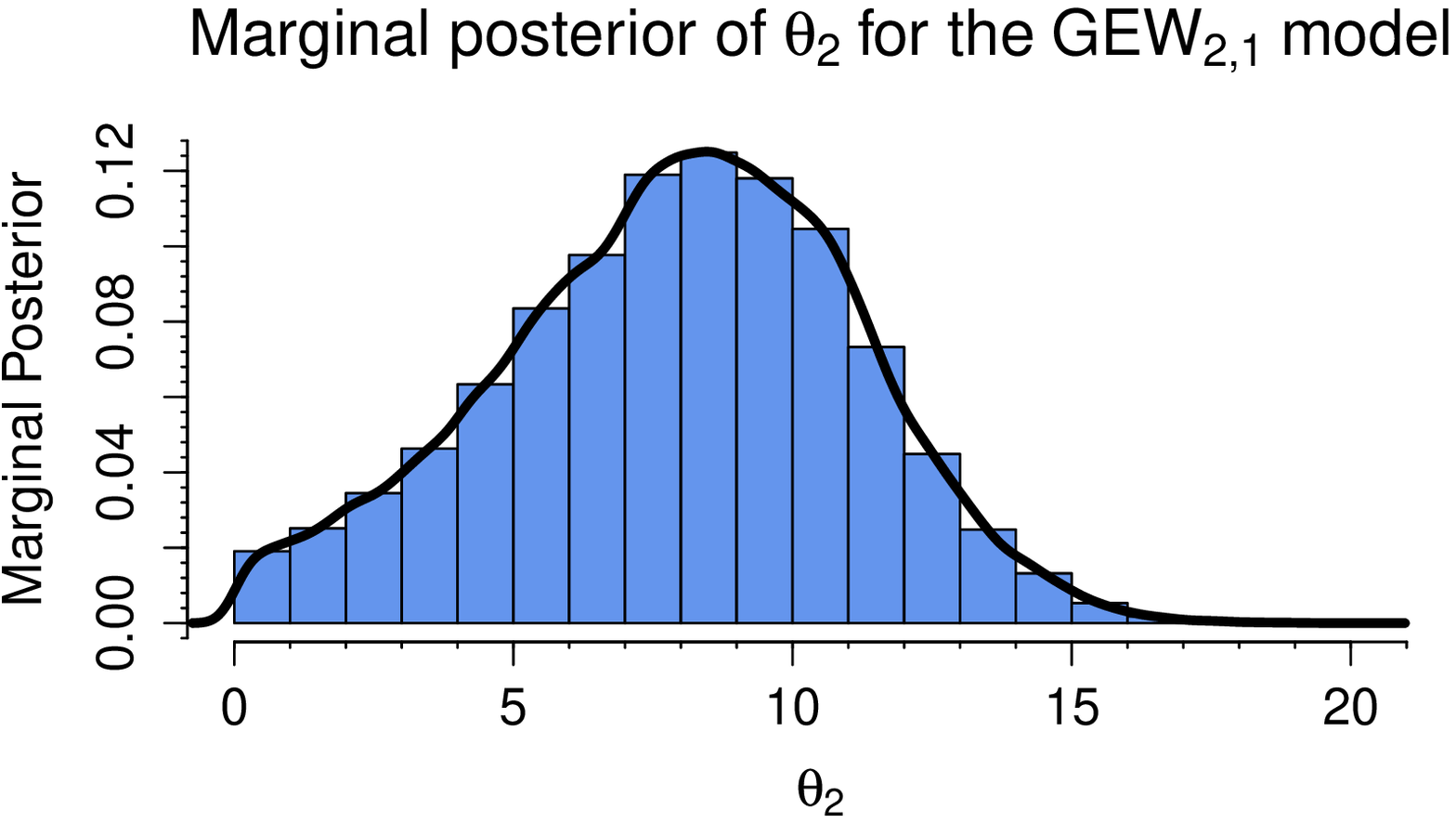}\includegraphics[width=4.8cm,height=2.2cm]{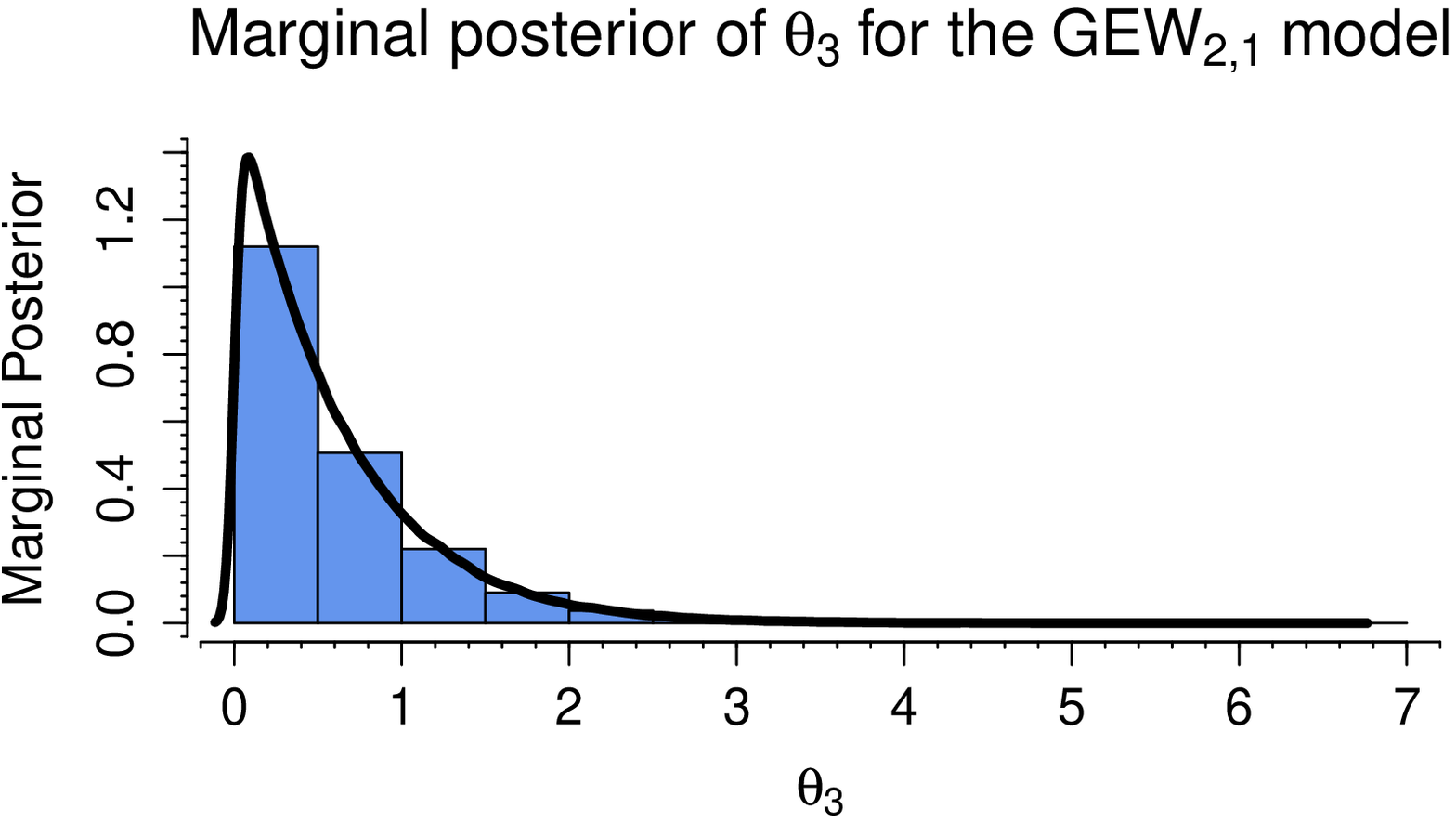}\includegraphics[width=4.8cm,height=2.2cm]{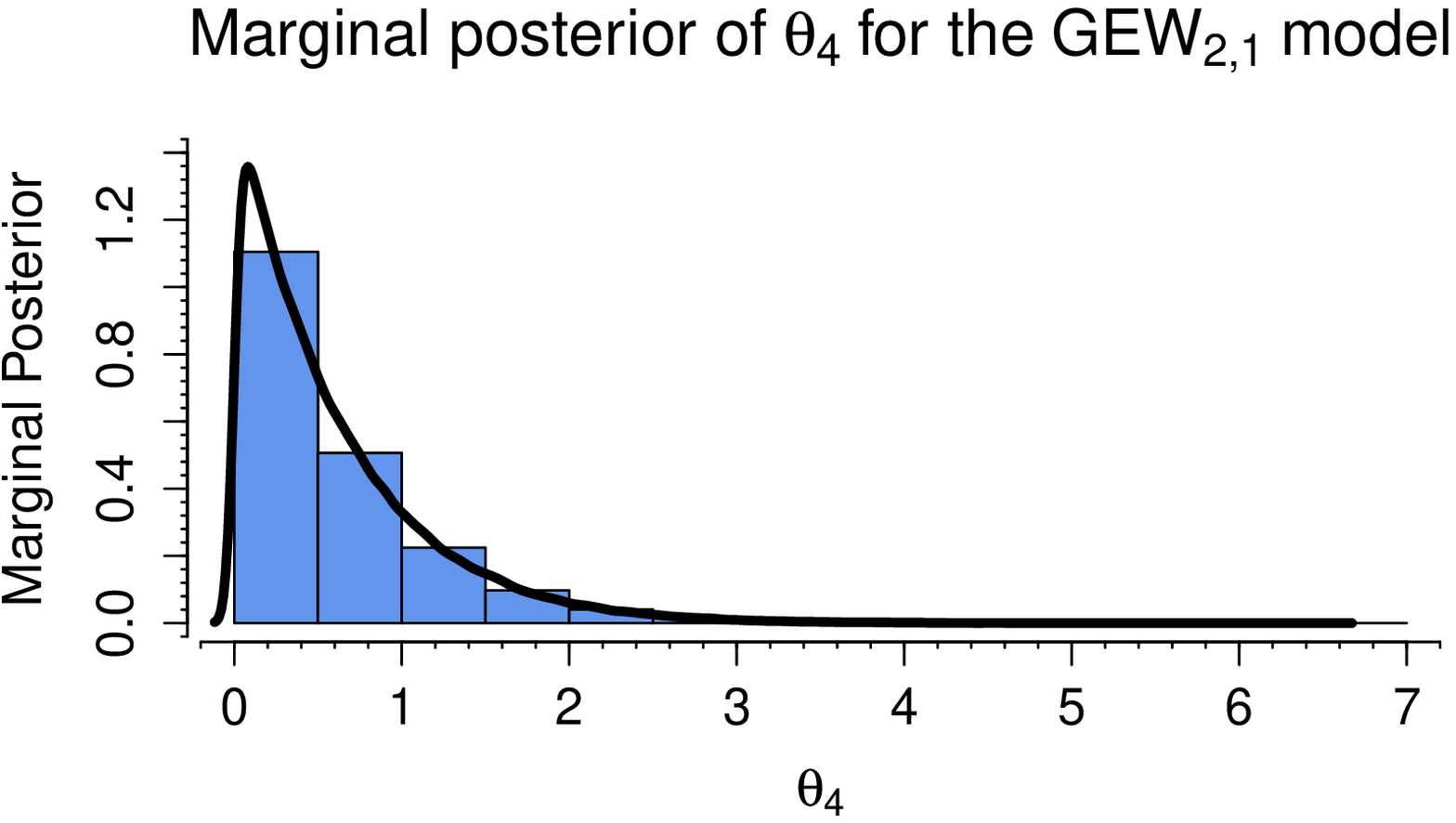}\includegraphics[width=4.8cm,height=2.2cm]{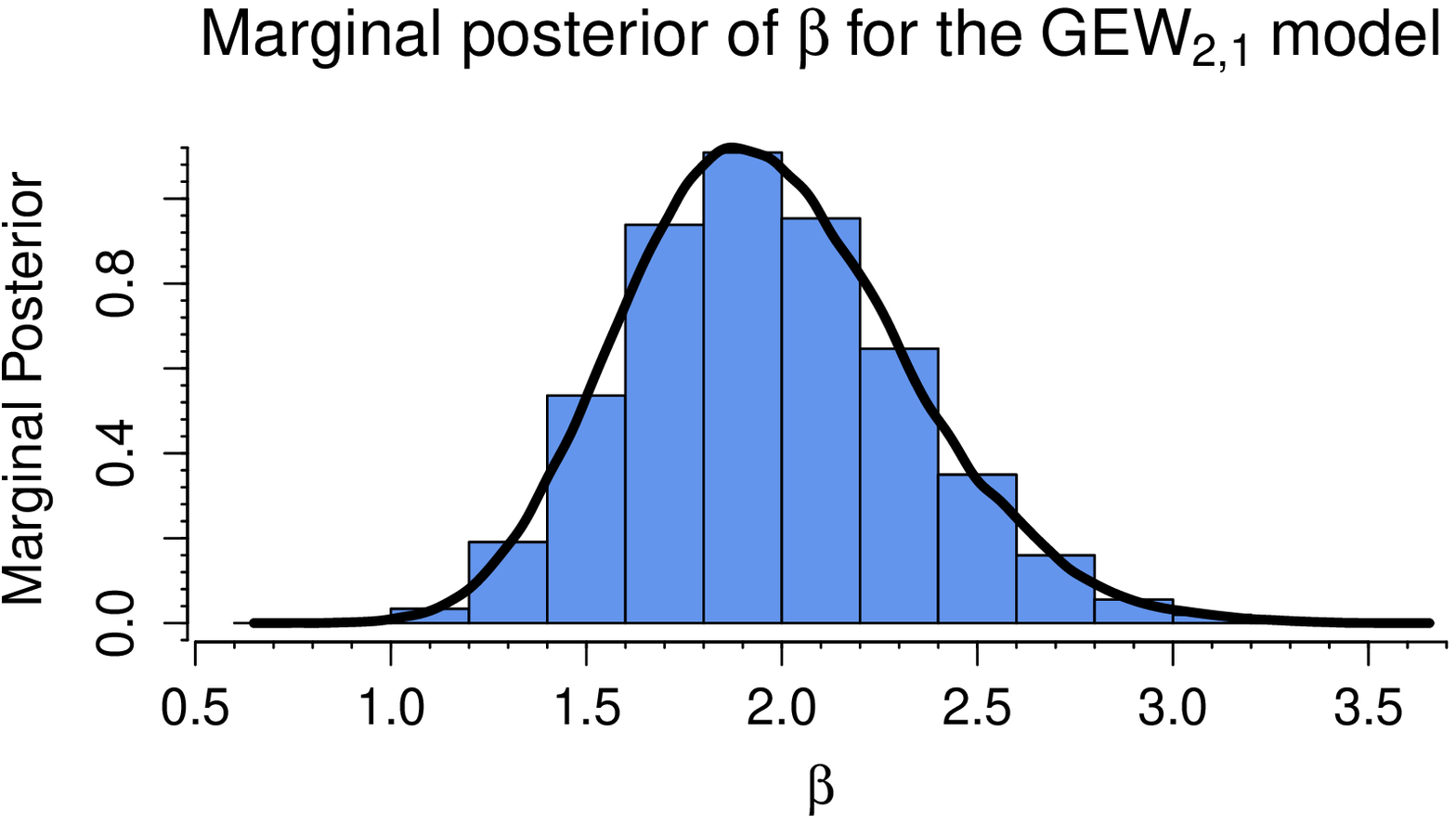}

\includegraphics[width=4.8cm,height=2.2cm]{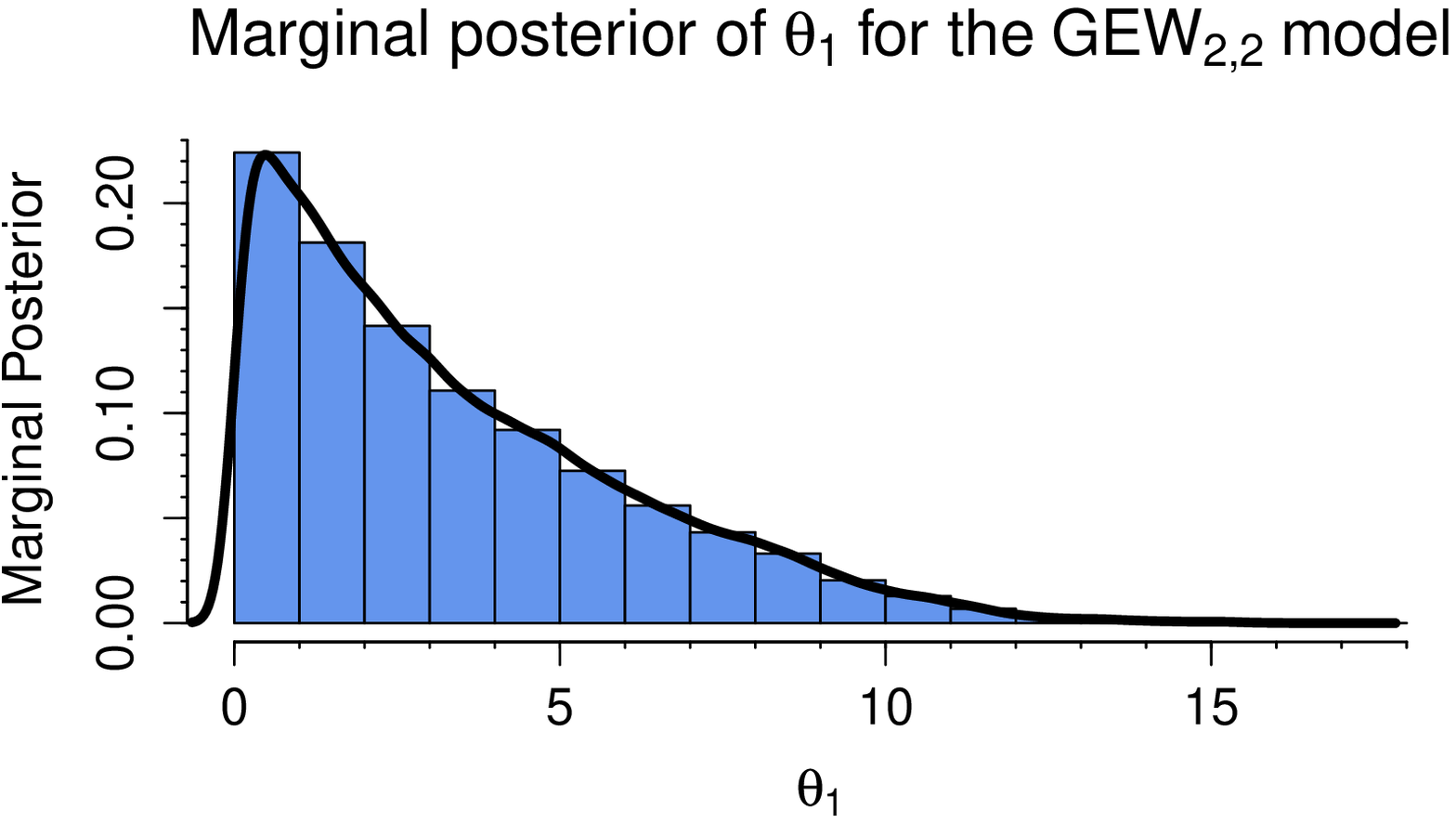}\includegraphics[width=4.8cm,height=2.2cm]{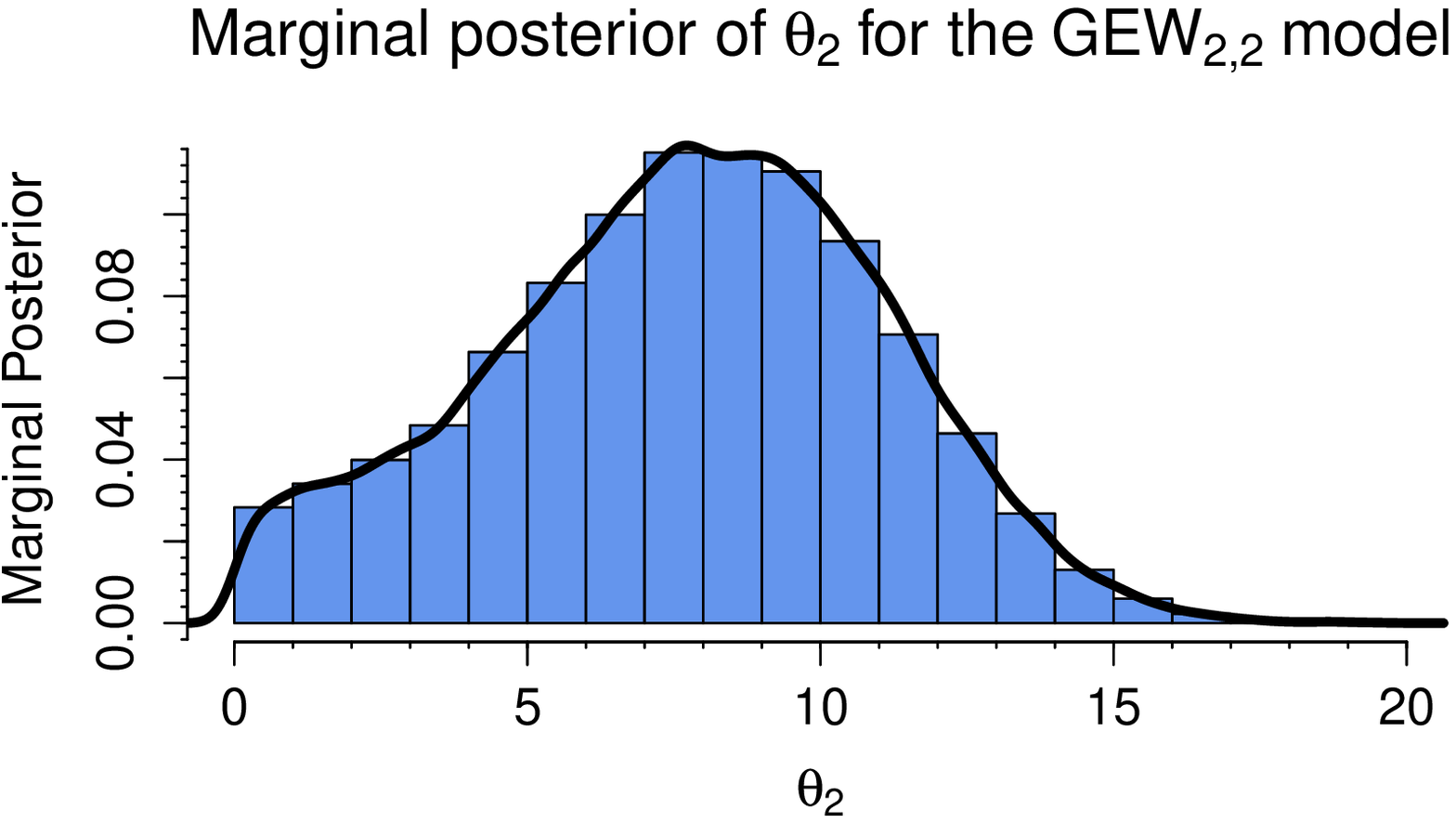}\includegraphics[width=4.8cm,height=2.2cm]{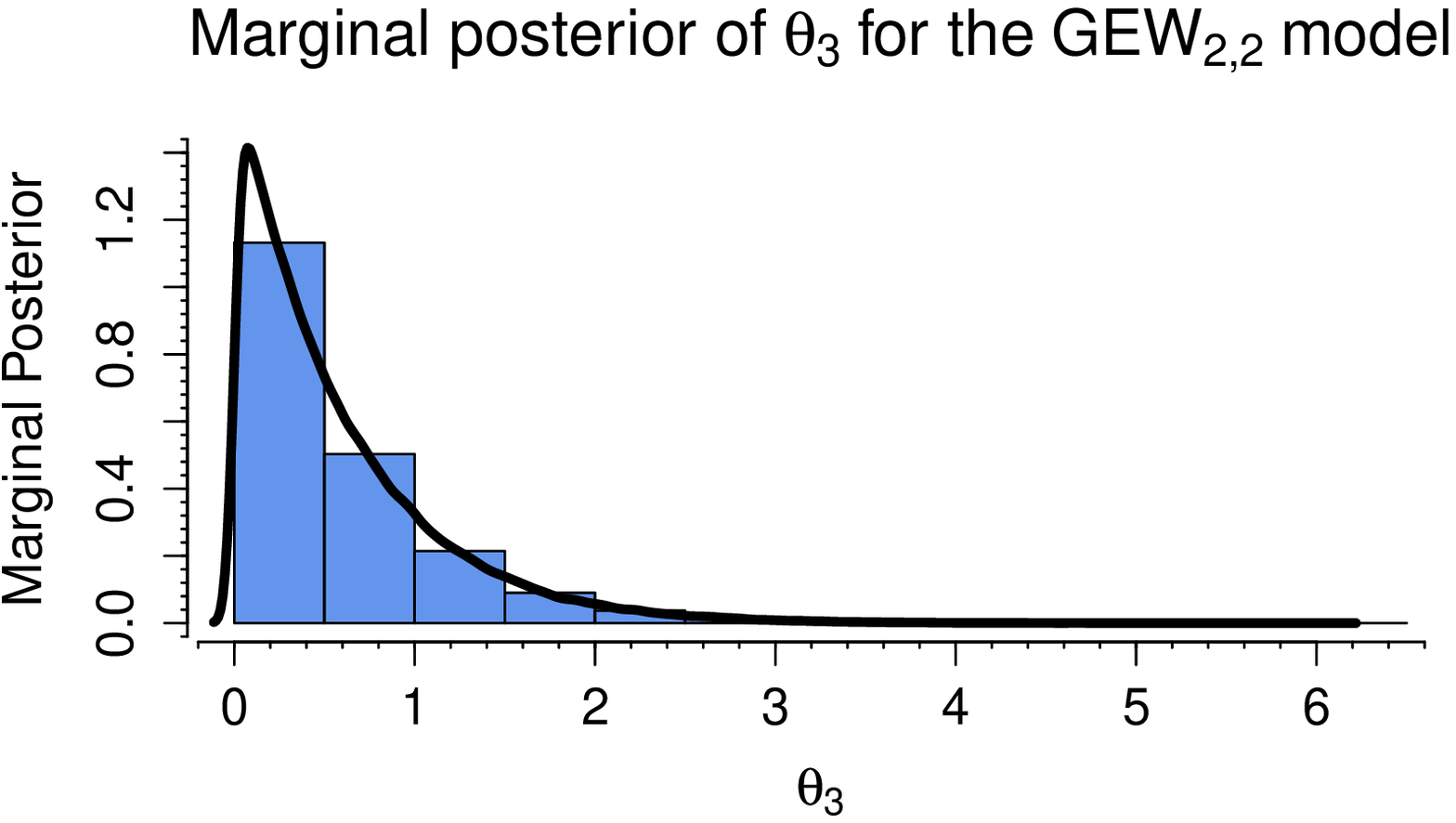}\includegraphics[width=4.8cm,height=2.2cm]{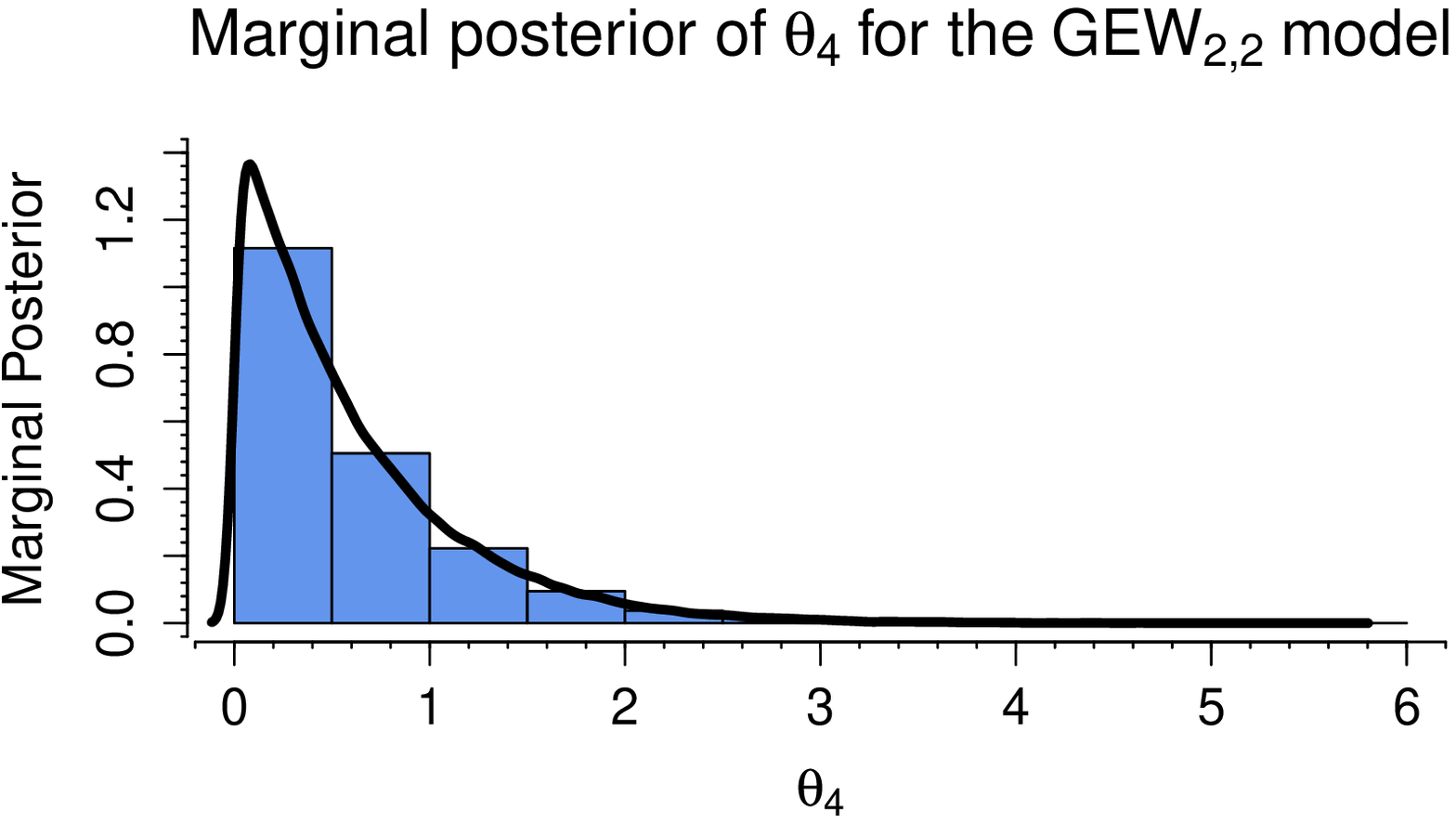}\includegraphics[width=4.8cm,height=2.2cm]{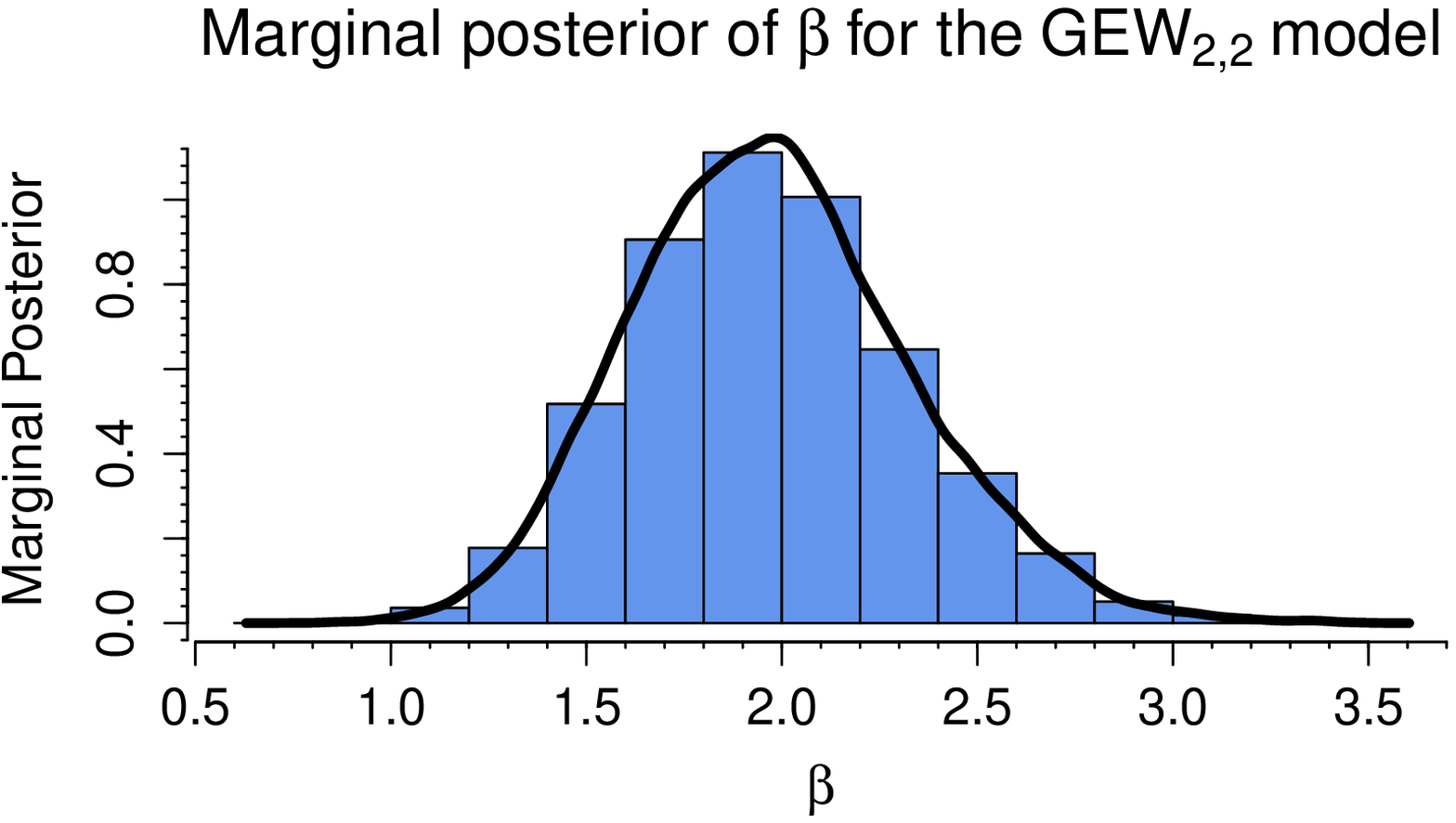}

\includegraphics[width=4.8cm,height=2.2cm]{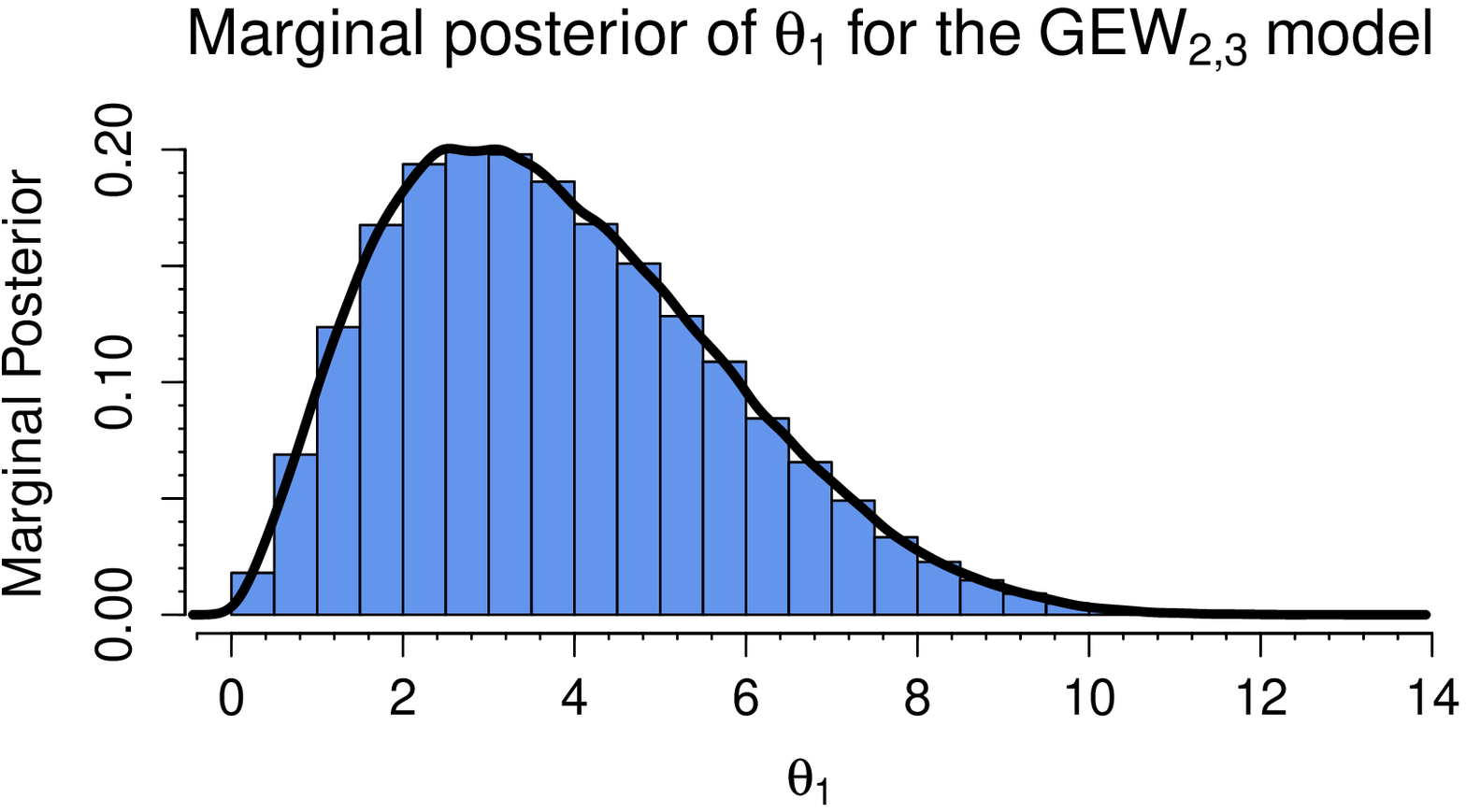}\includegraphics[width=4.8cm,height=2.2cm]{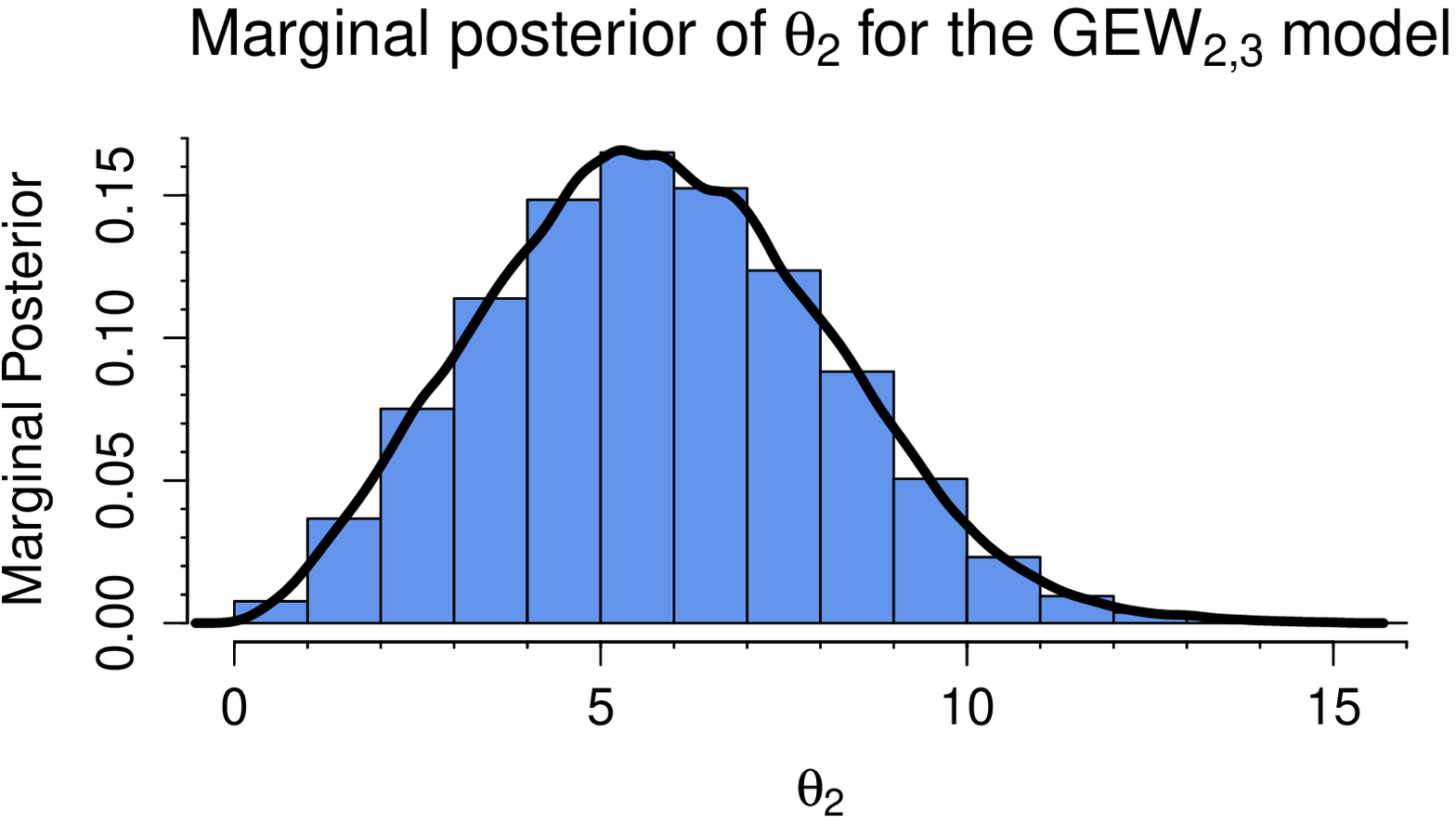}\includegraphics[width=4.8cm,height=2.2cm]{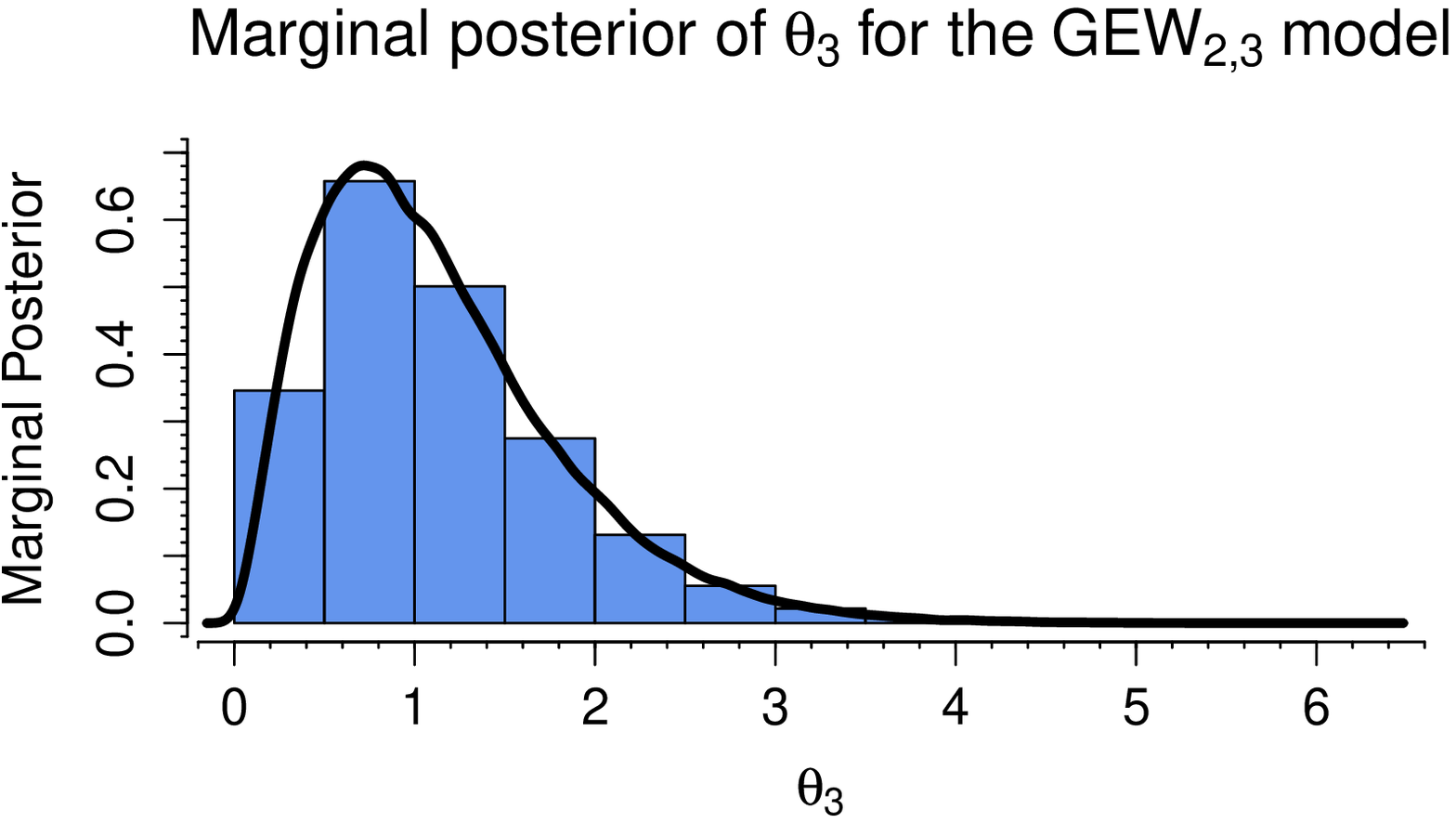}\includegraphics[width=4.8cm,height=2.2cm]{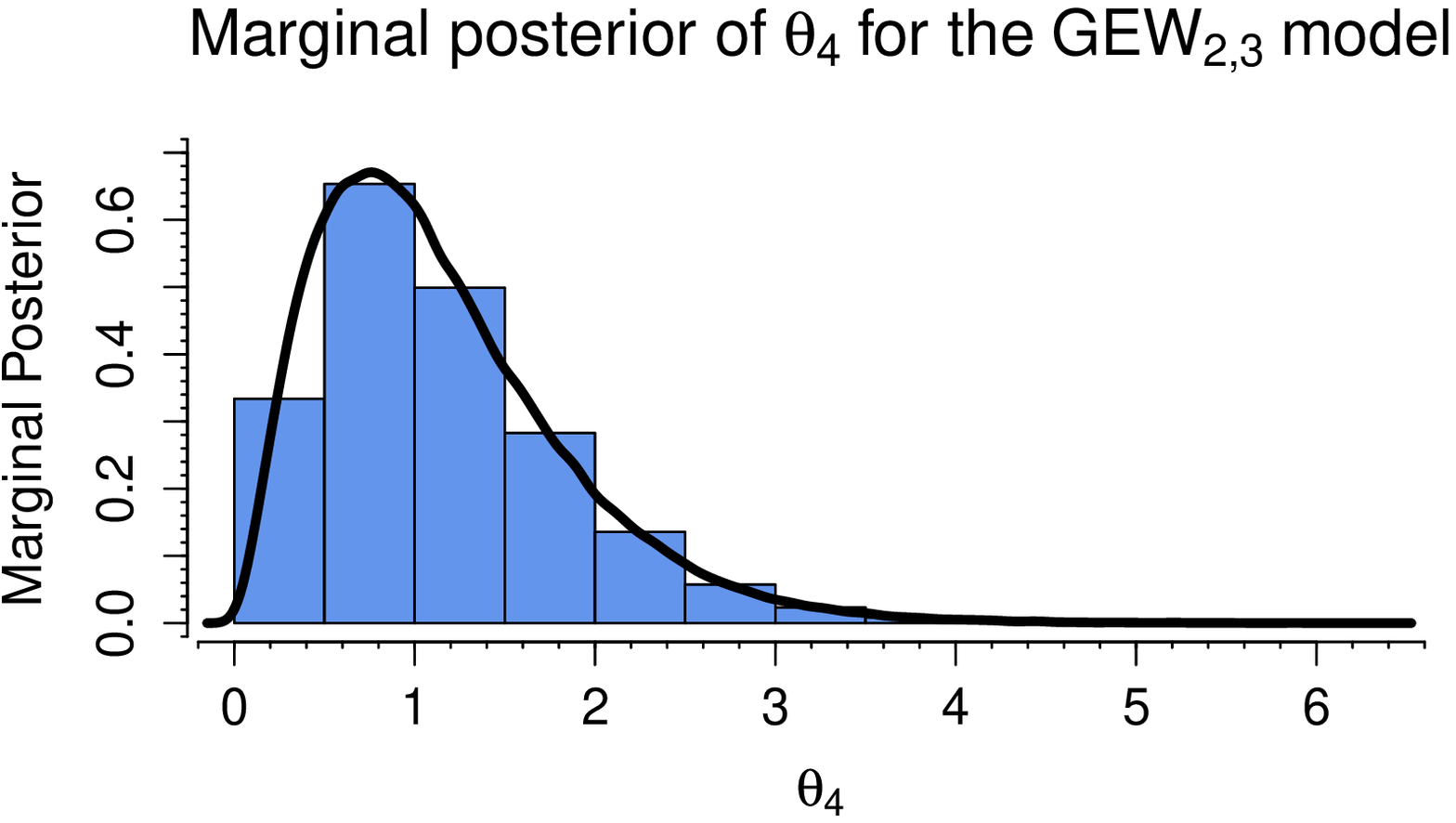}\includegraphics[width=4.8cm,height=2.2cm]{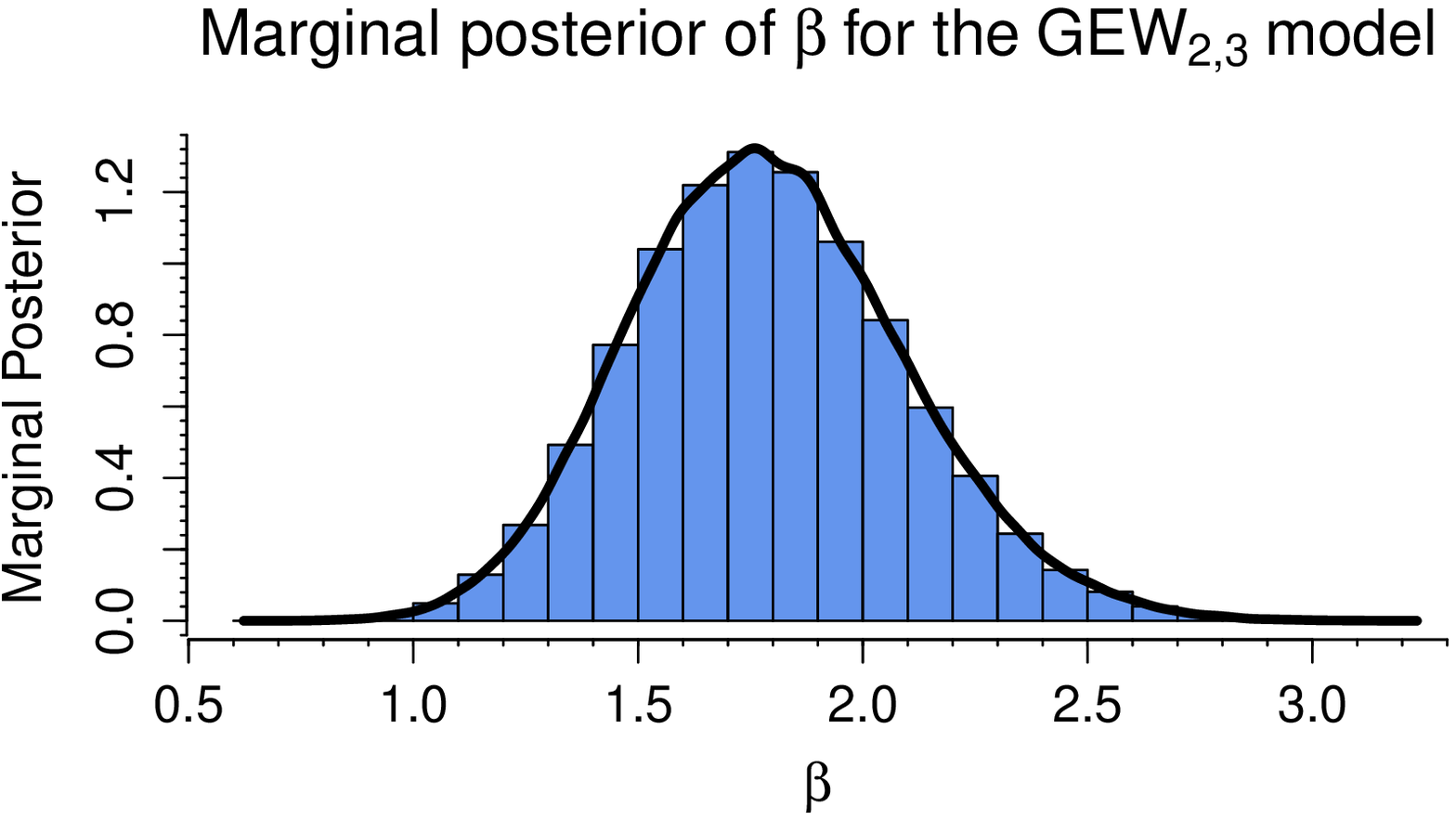}

\includegraphics[width=4.8cm,height=2.2cm]{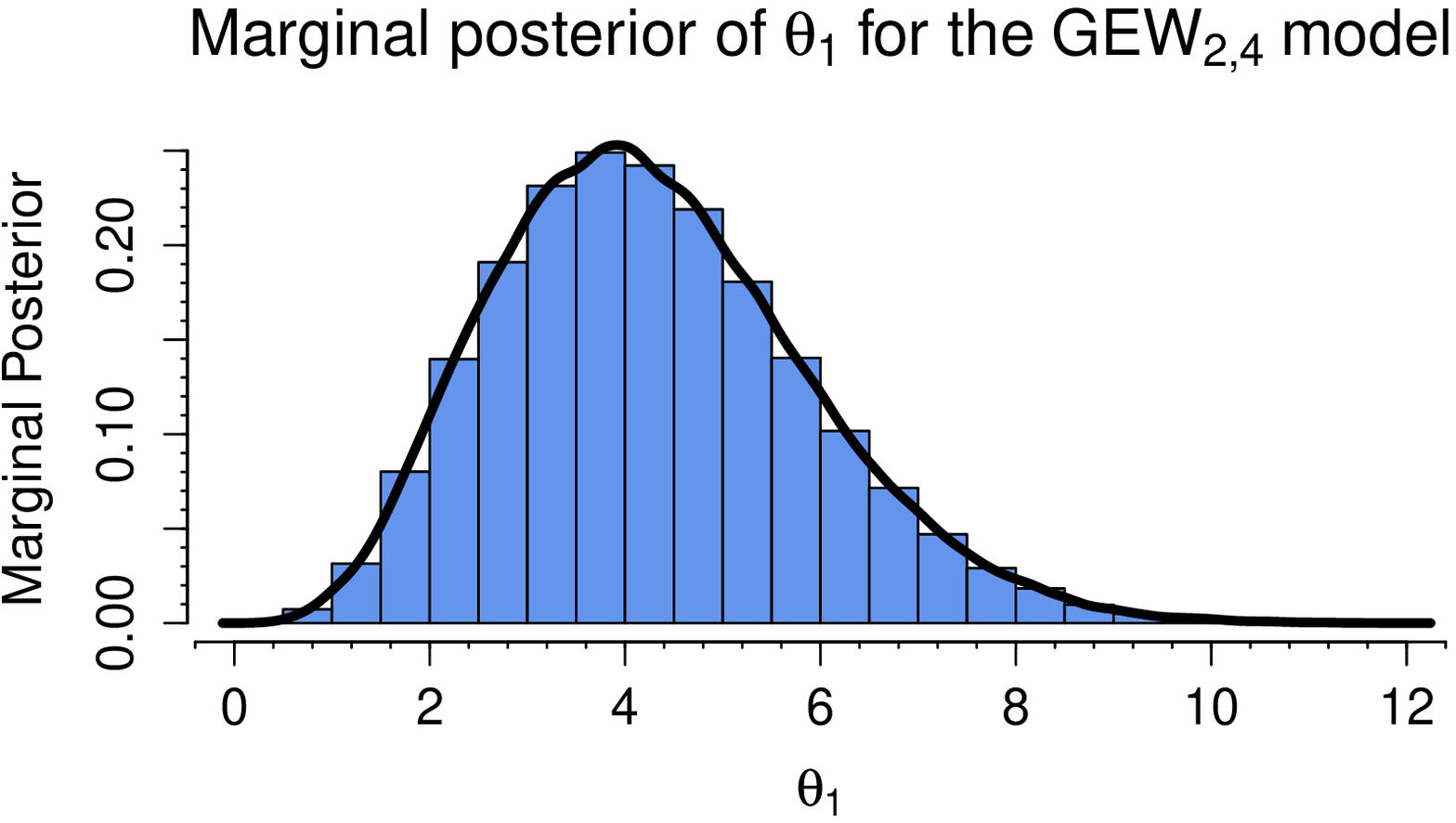}\includegraphics[width=4.8cm,height=2.2cm]{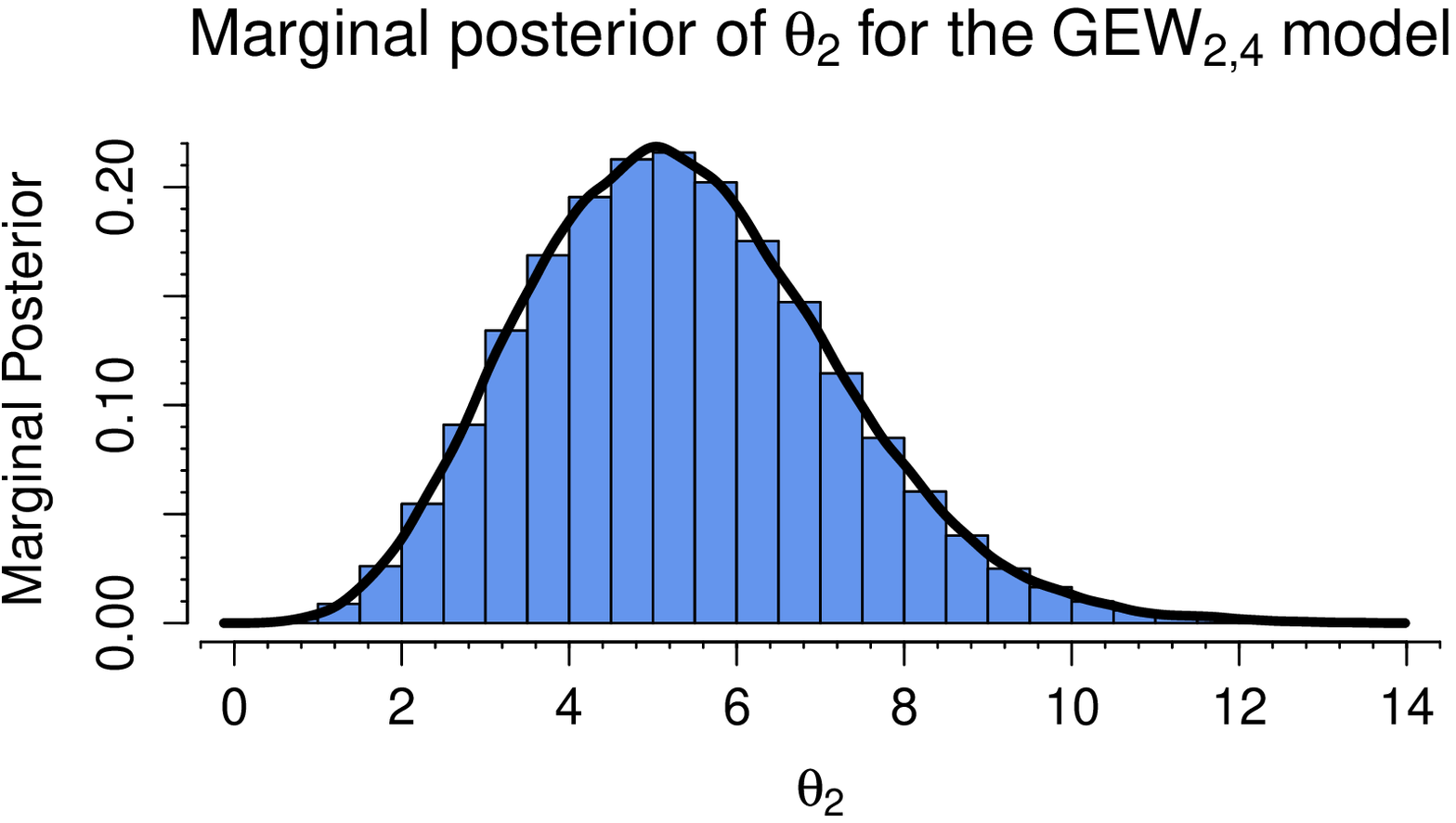}\includegraphics[width=4.8cm,height=2.2cm]{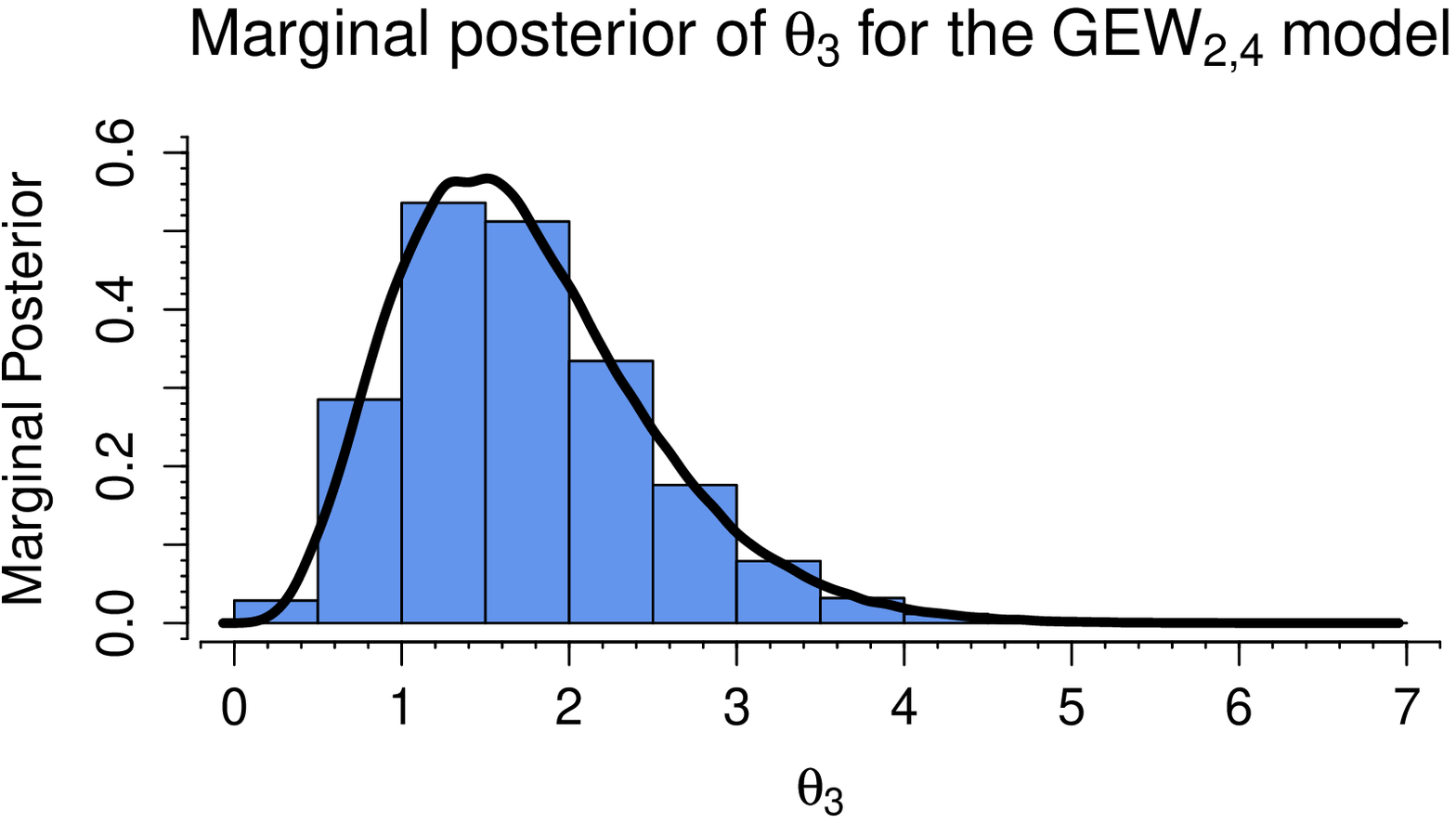}\includegraphics[width=4.8cm,height=2.2cm]{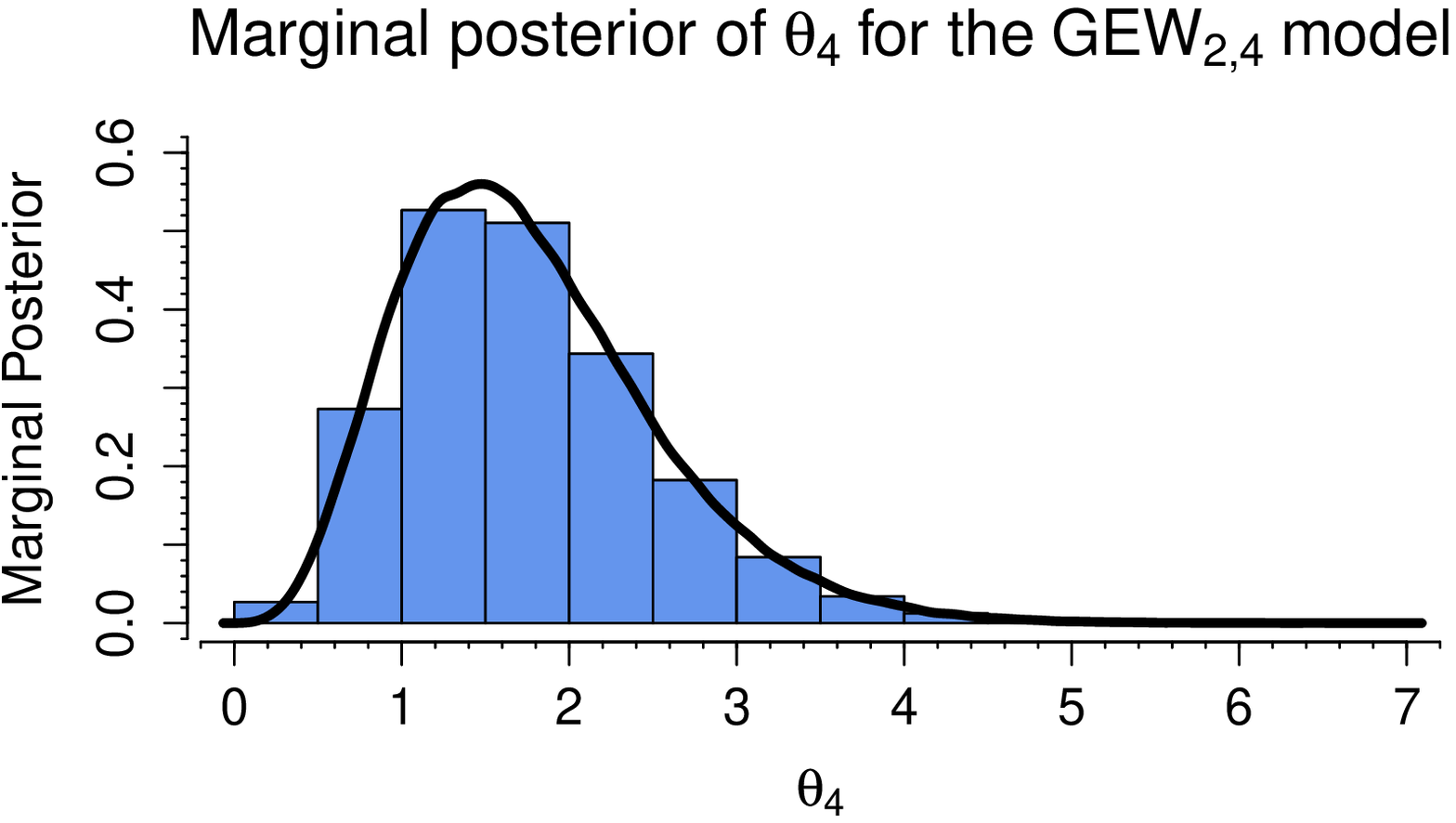}\includegraphics[width=4.8cm,height=2.2cm]{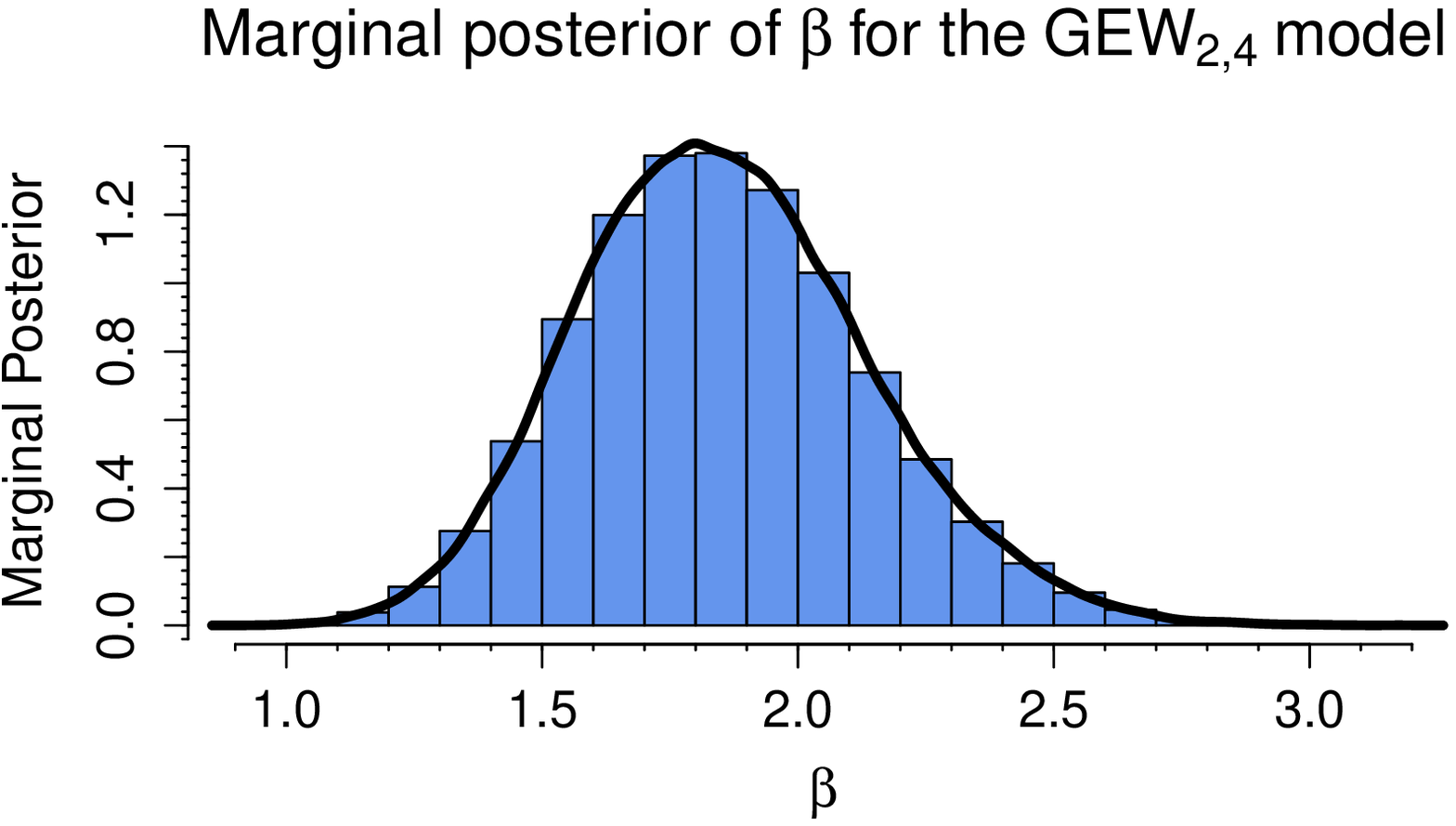}

\includegraphics[width=4.8cm,height=2.2cm]{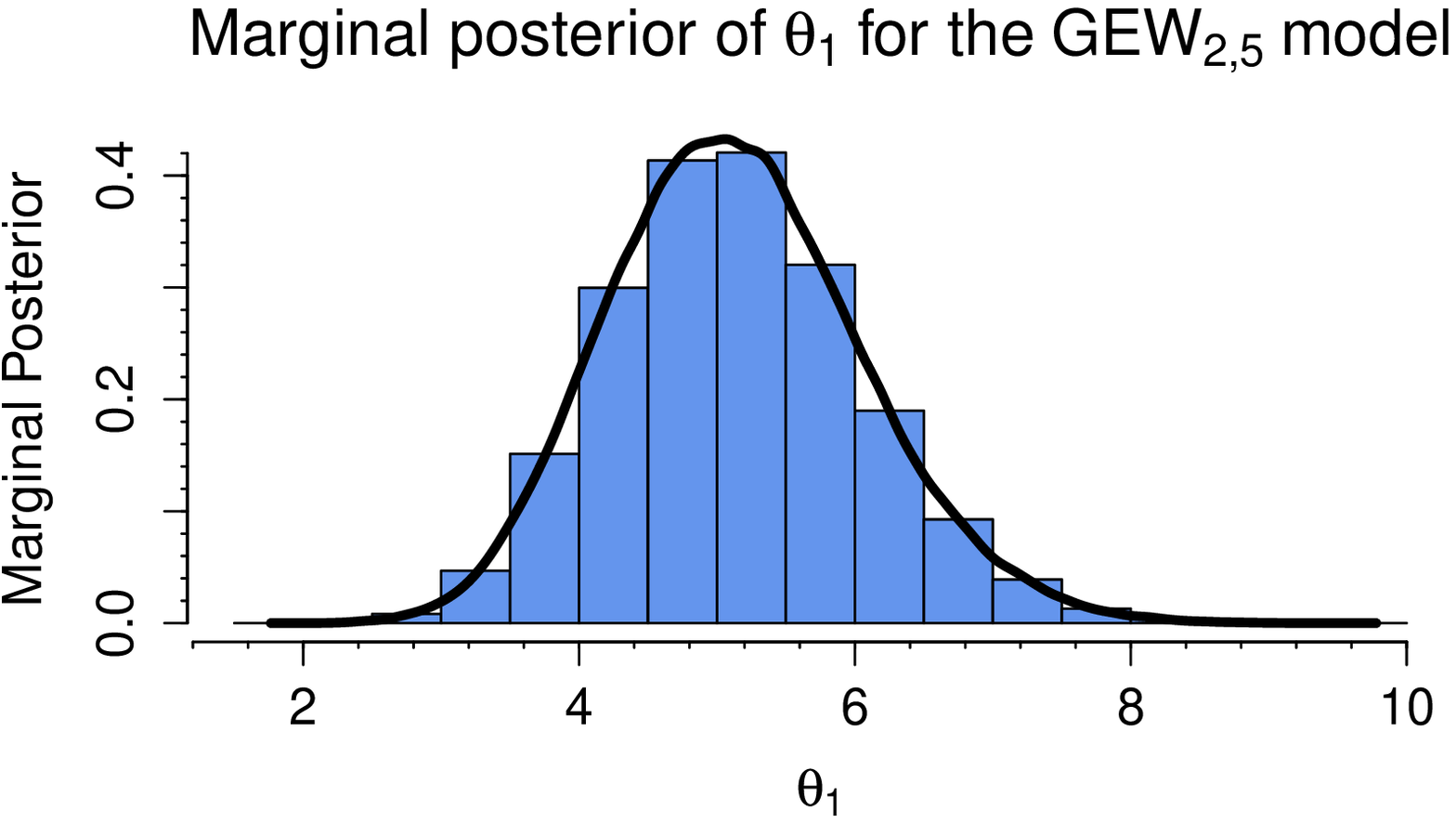}\includegraphics[width=4.8cm,height=2.2cm]{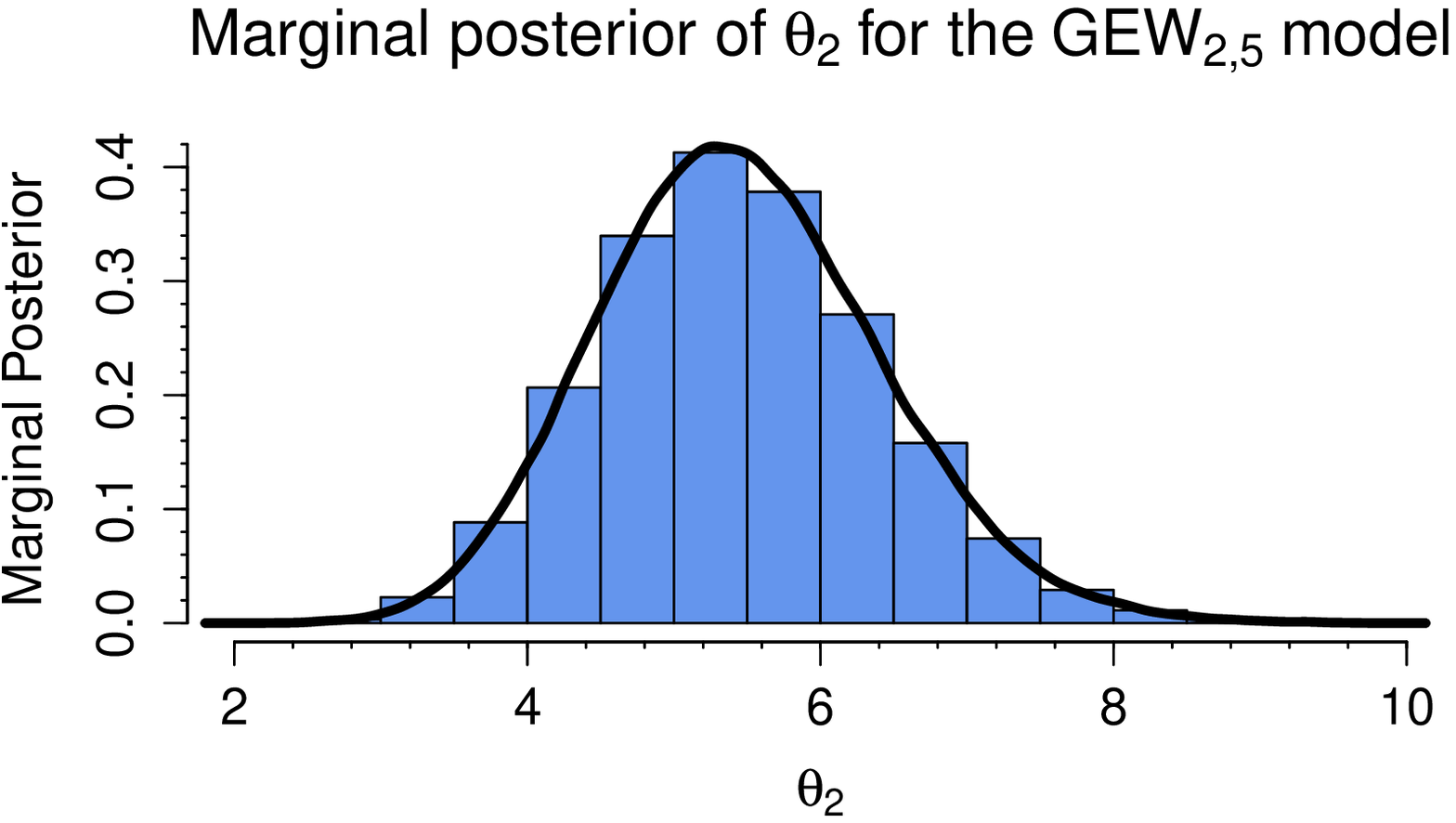}\includegraphics[width=4.8cm,height=2.2cm]{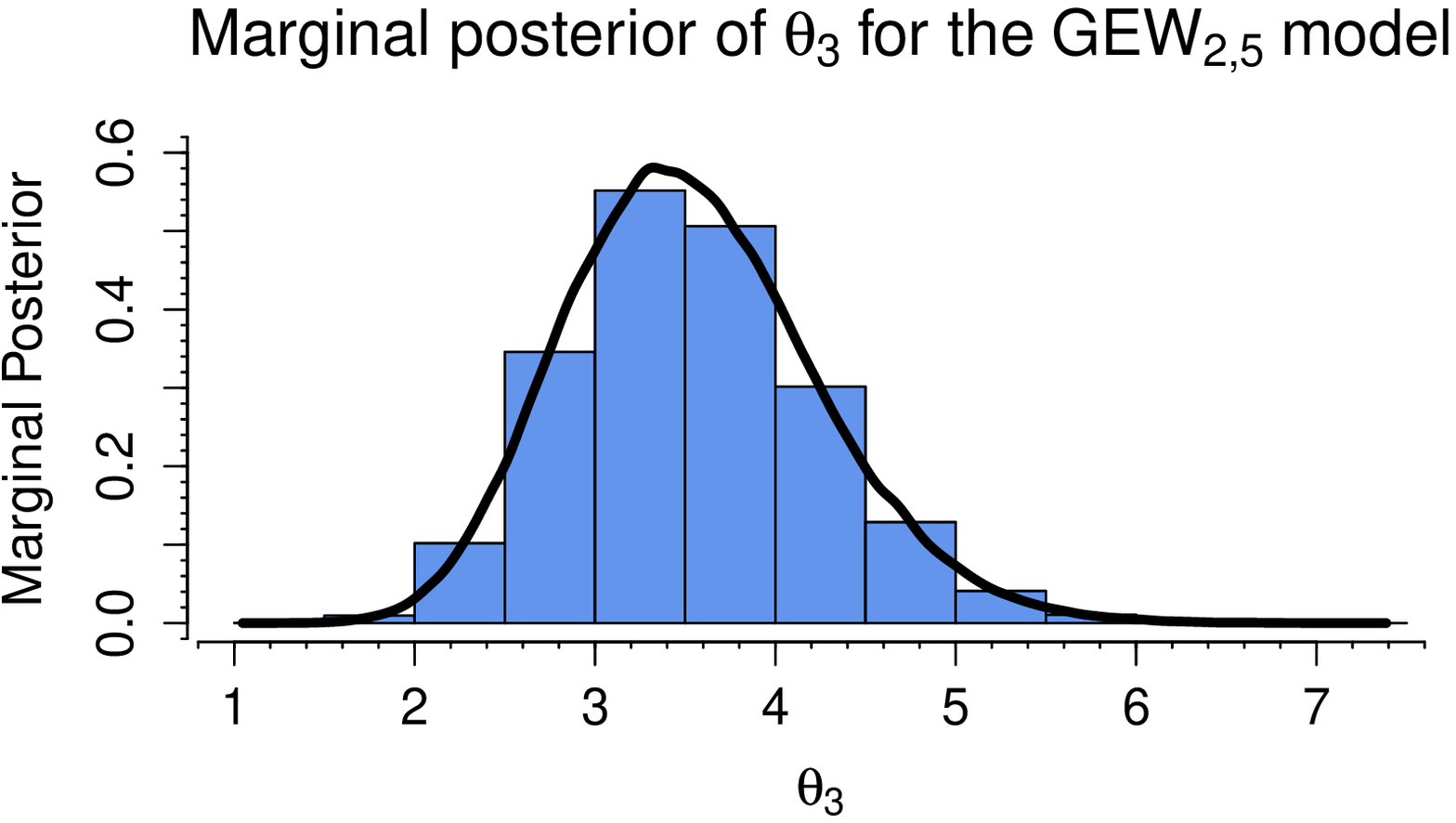}\includegraphics[width=4.8cm,height=2.2cm]{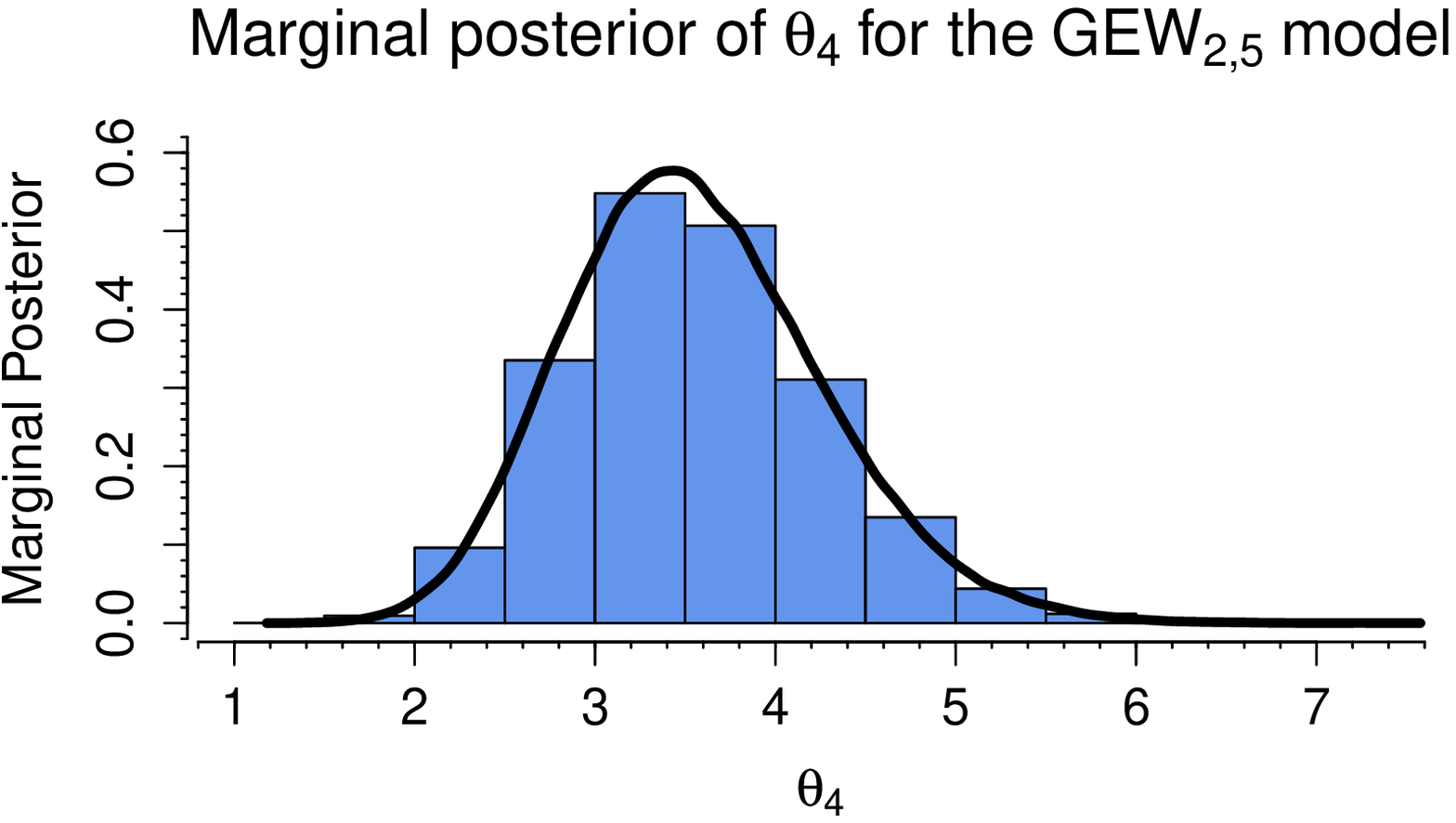}\includegraphics[width=4.8cm,height=2.2cm]{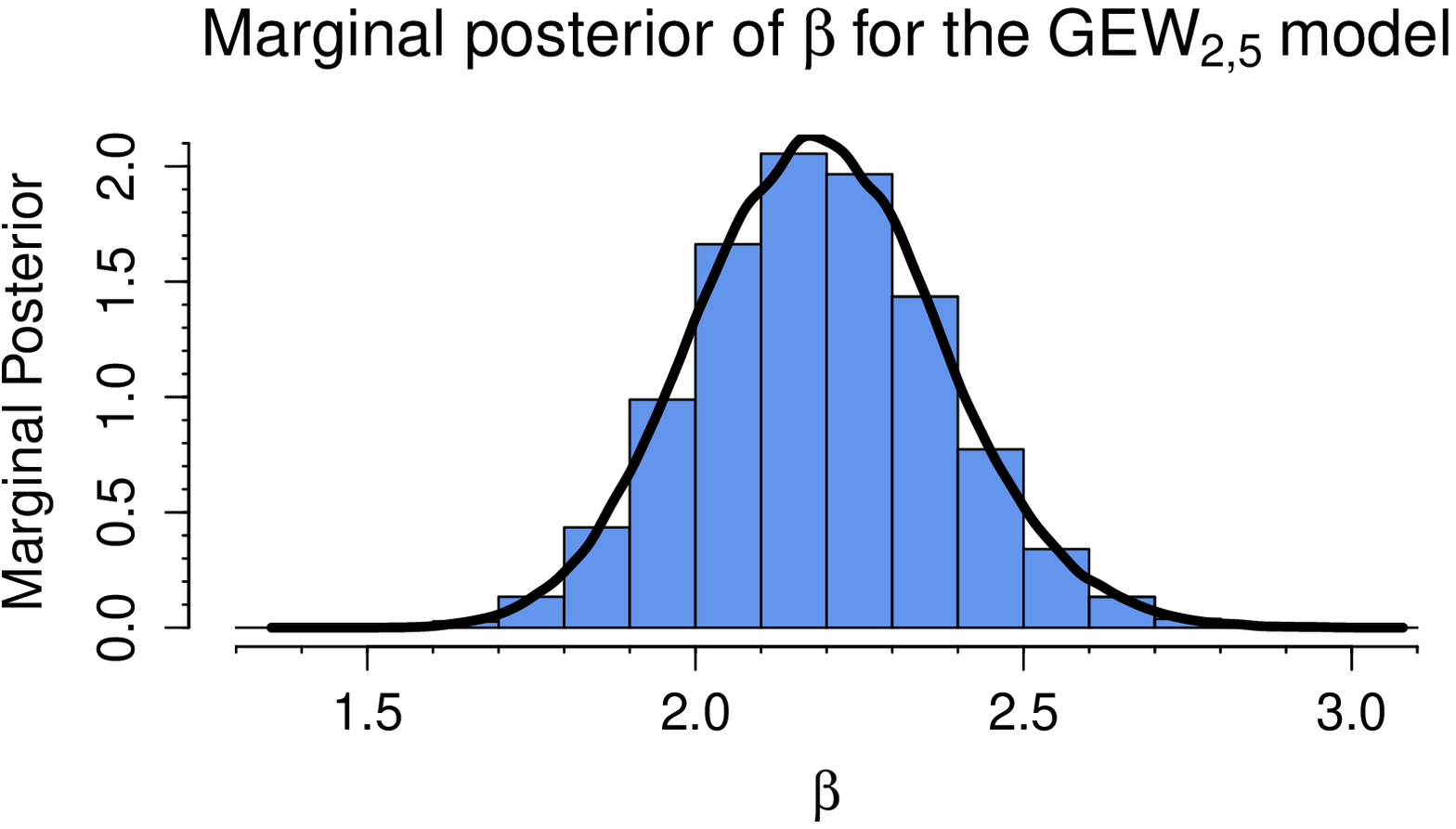}

\includegraphics[width=4.8cm,height=2.2cm]{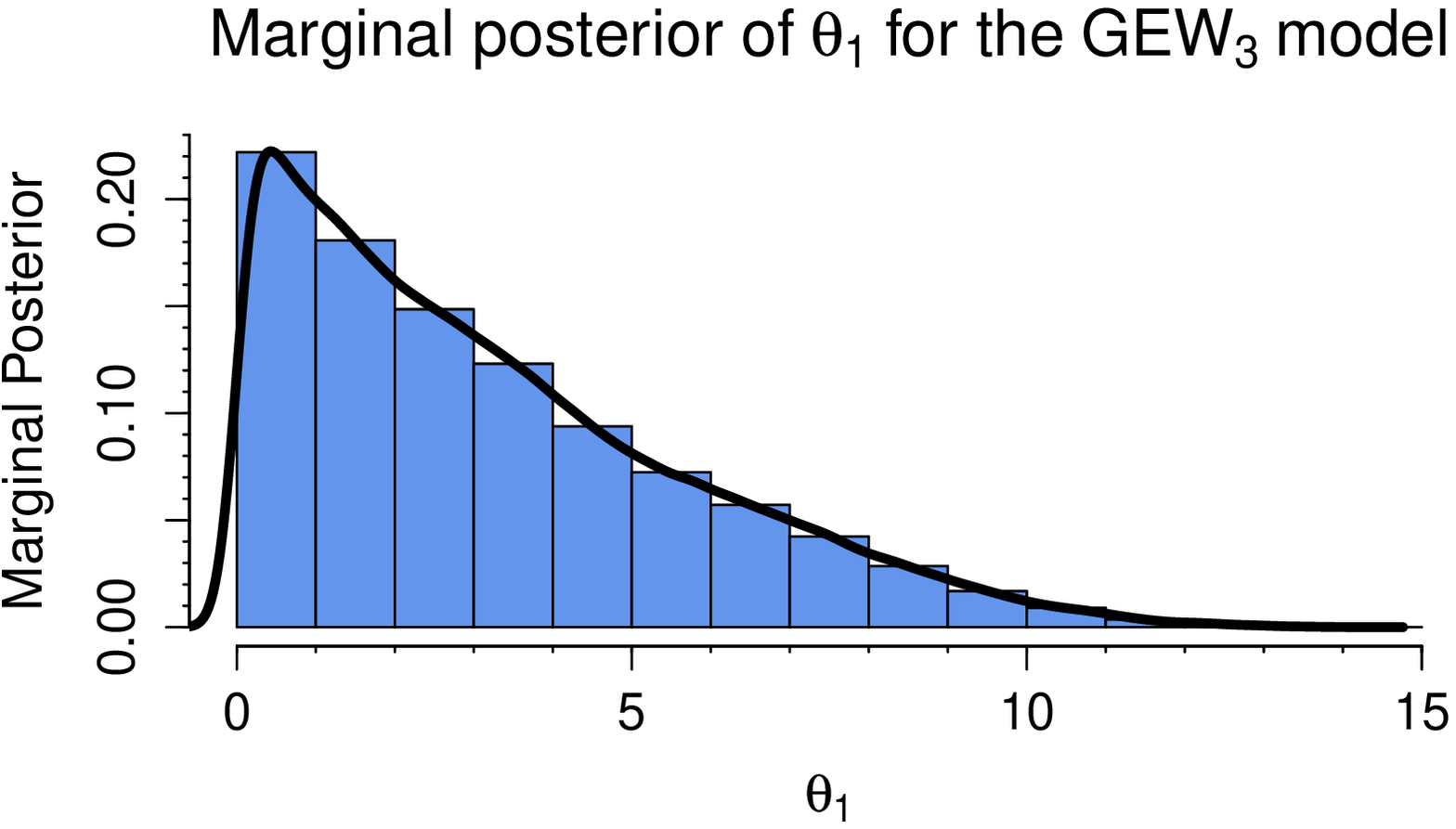}\includegraphics[width=4.8cm,height=2.2cm]{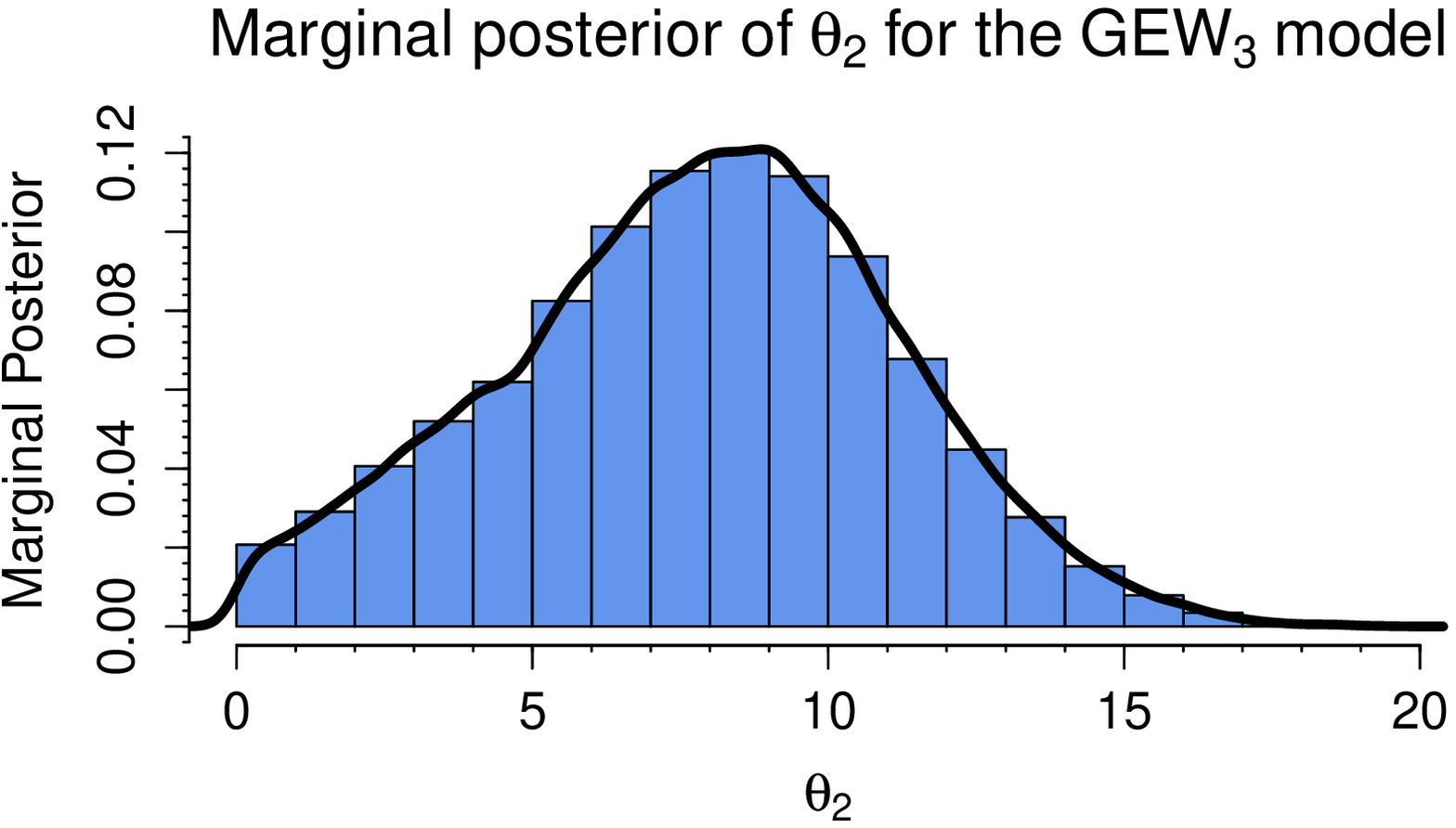}\includegraphics[width=4.8cm,height=2.2cm]{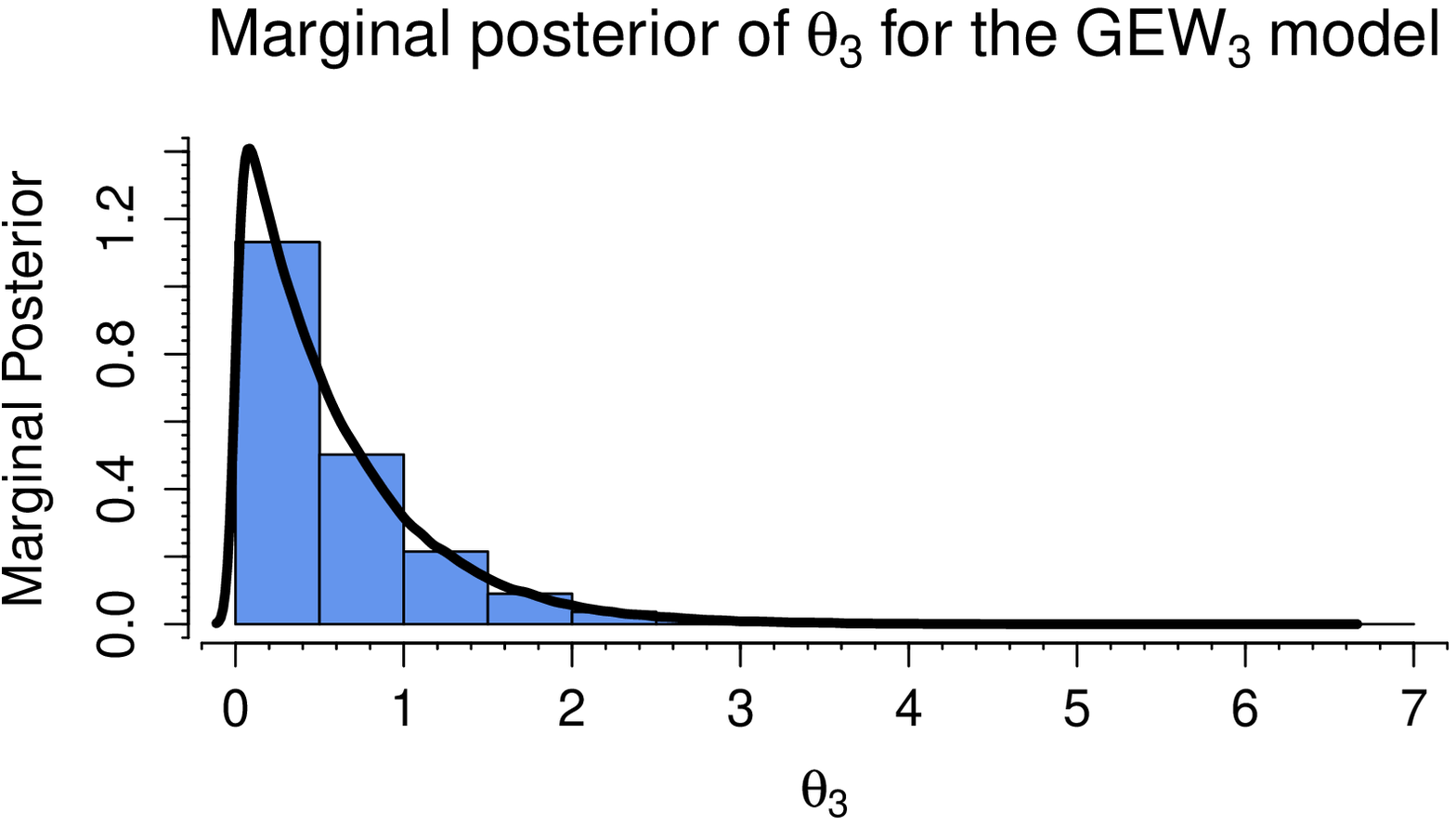}\includegraphics[width=4.8cm,height=2.2cm]{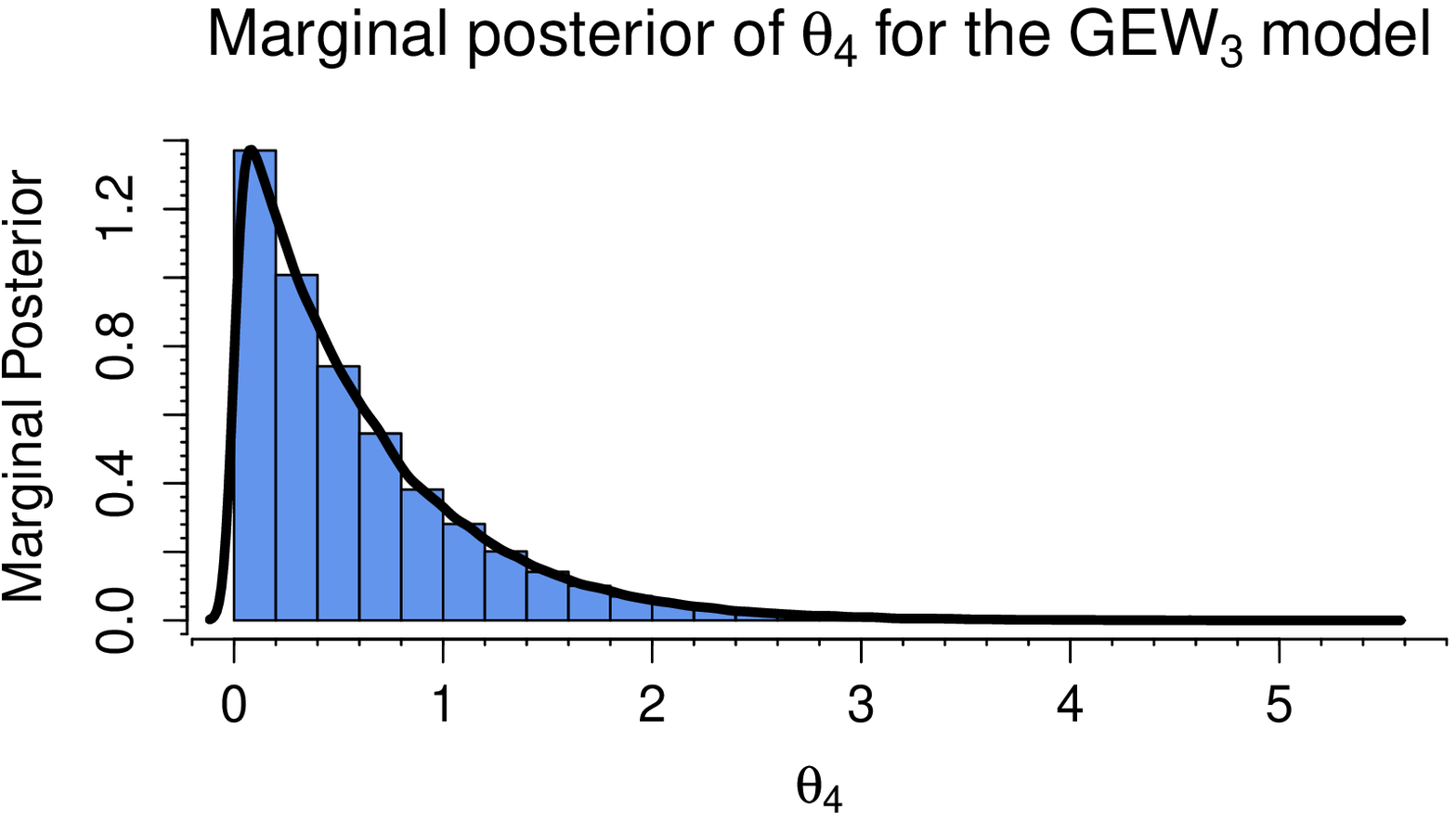}\includegraphics[width=4.8cm,height=2.2cm]{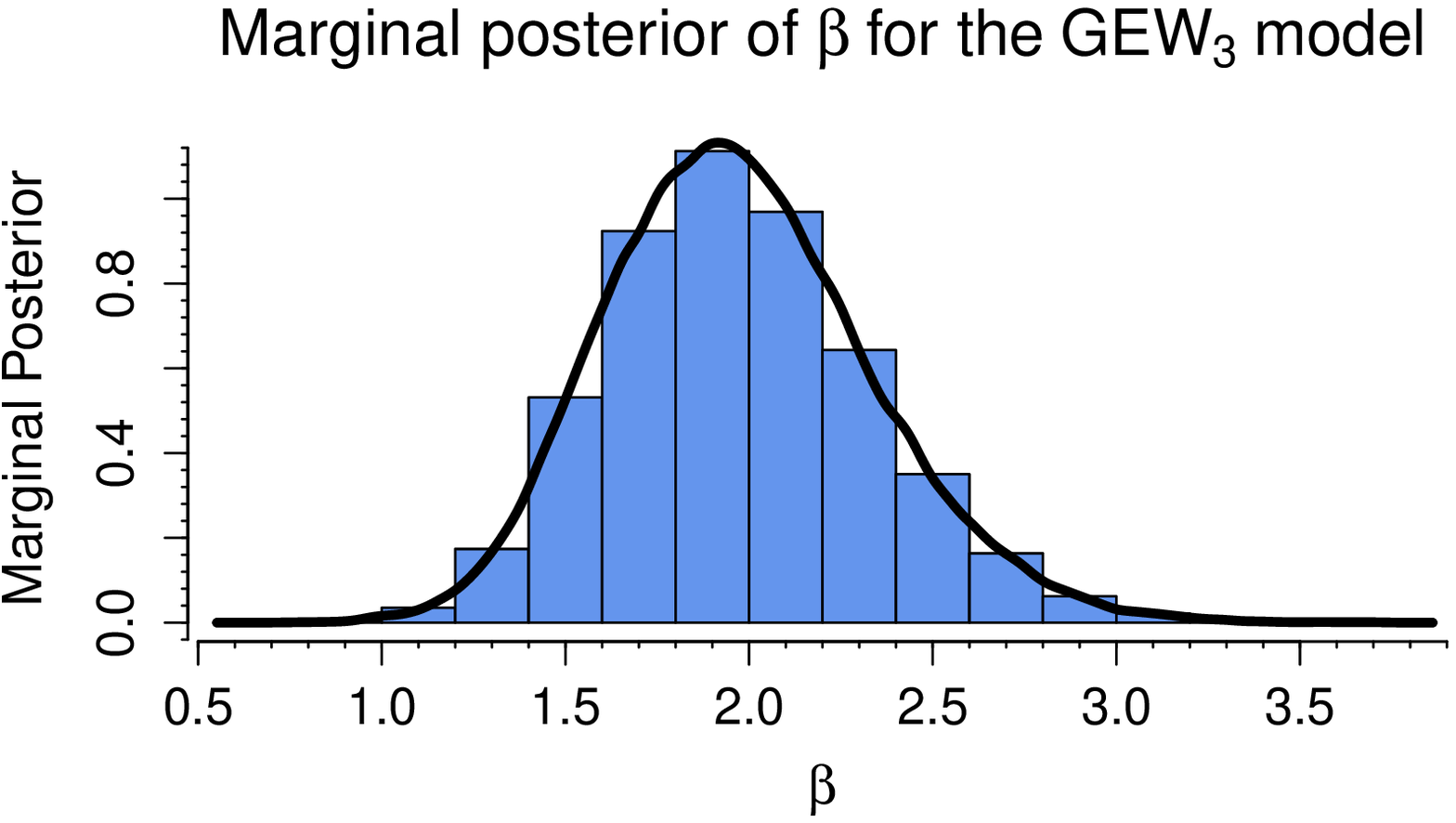}

\includegraphics[width=4.8cm,height=2.2cm]{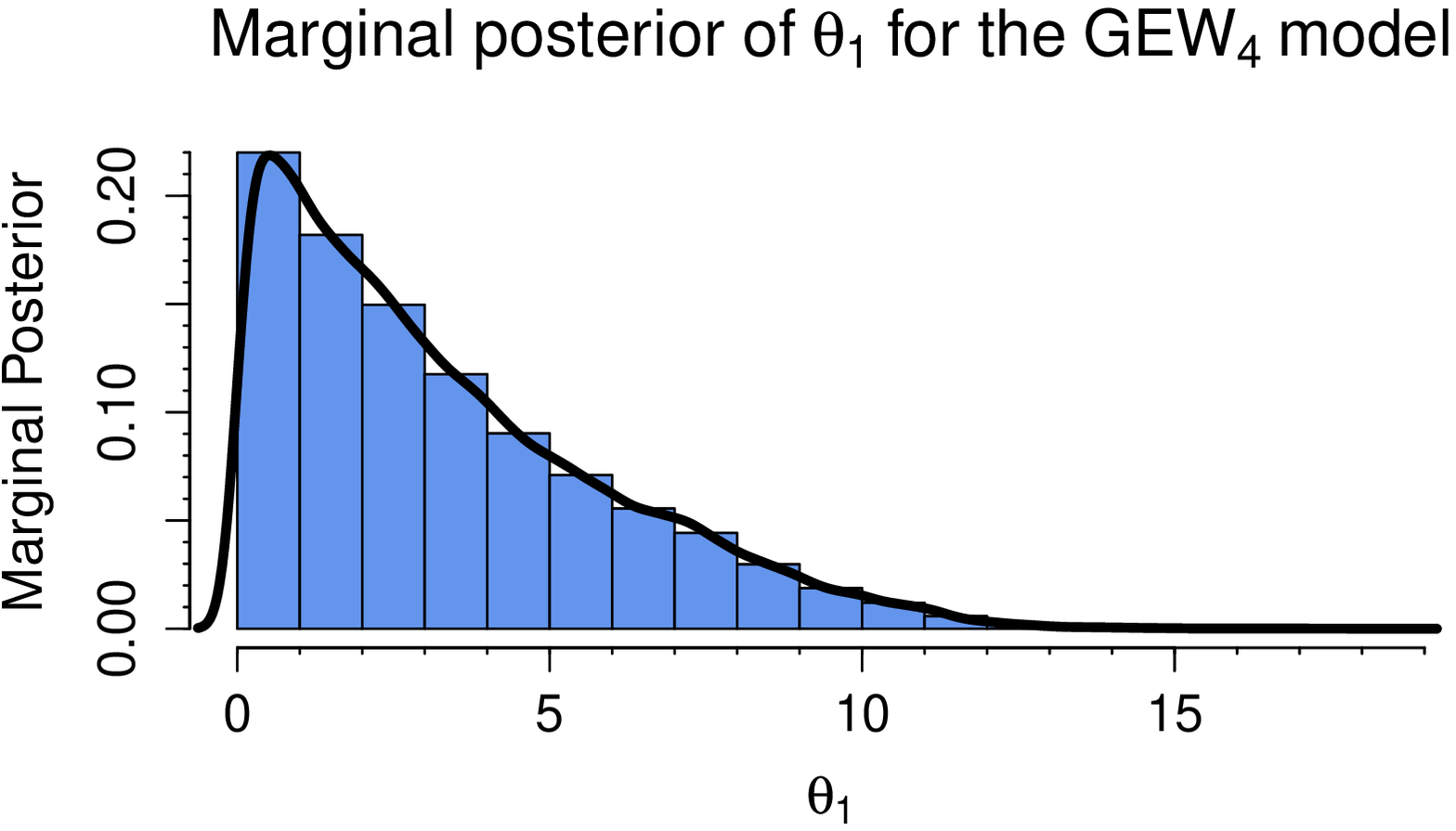}\includegraphics[width=4.8cm,height=2.2cm]{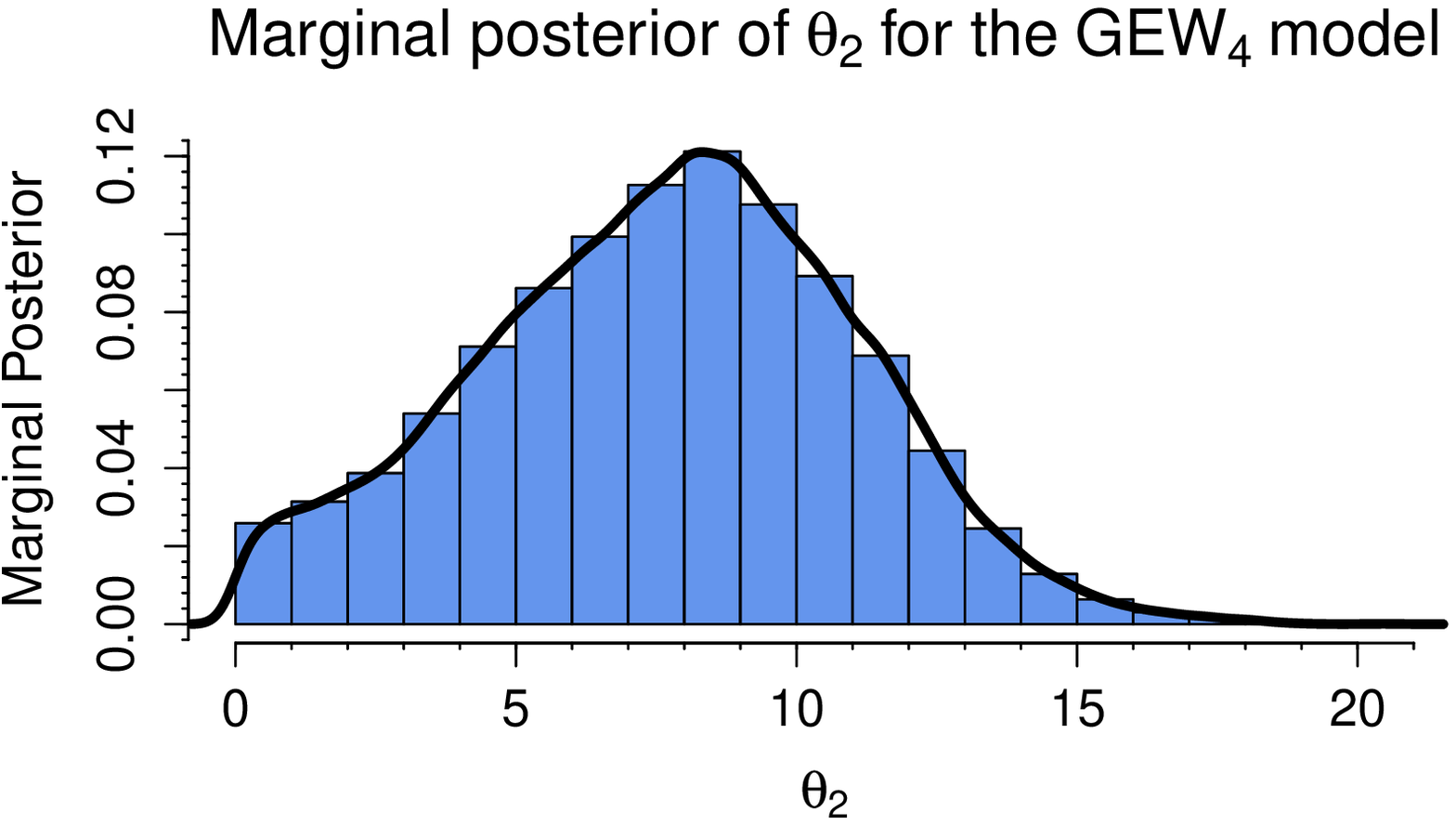}\includegraphics[width=4.8cm,height=2.2cm]{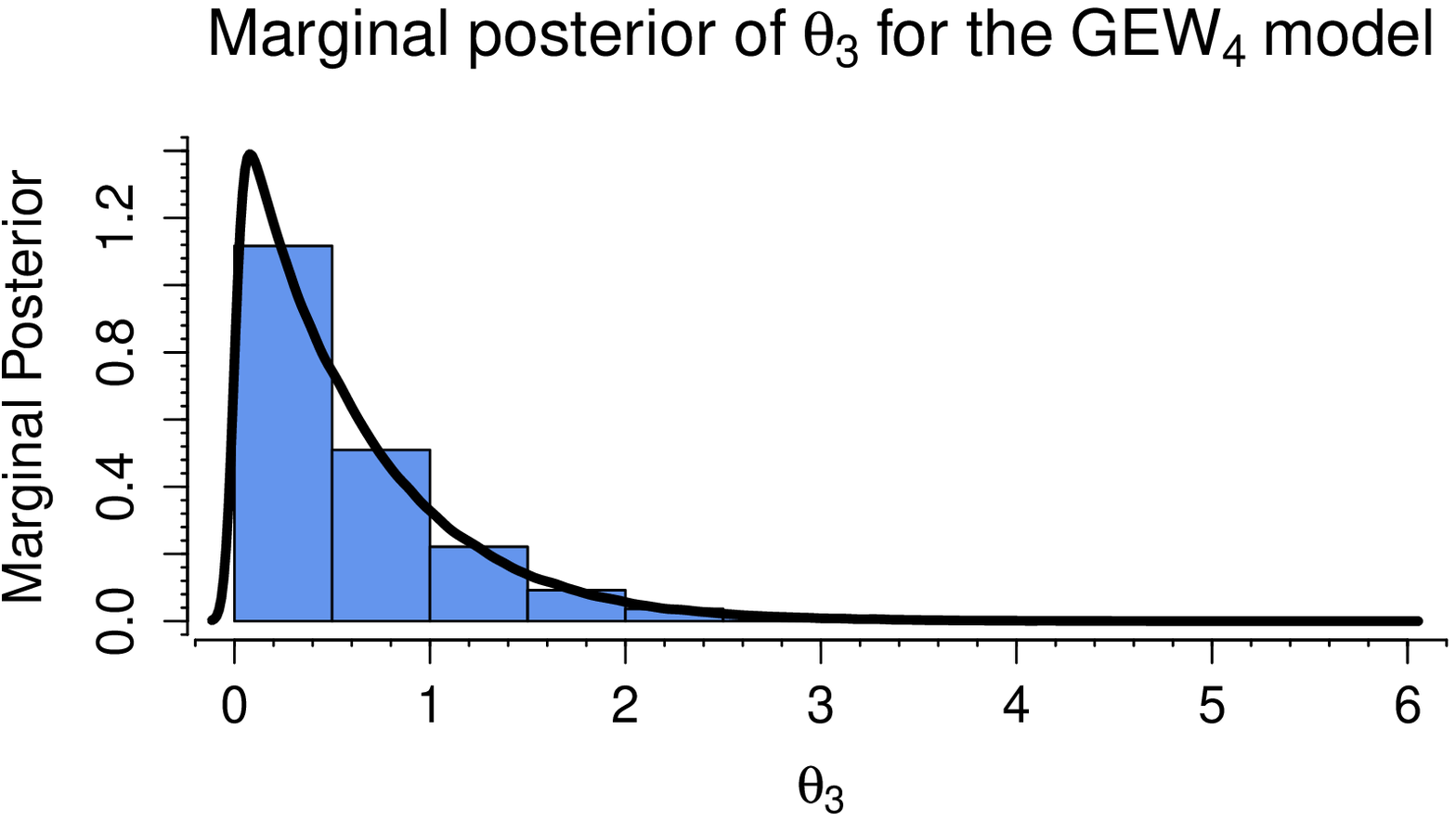}\includegraphics[width=4.8cm,height=2.2cm]{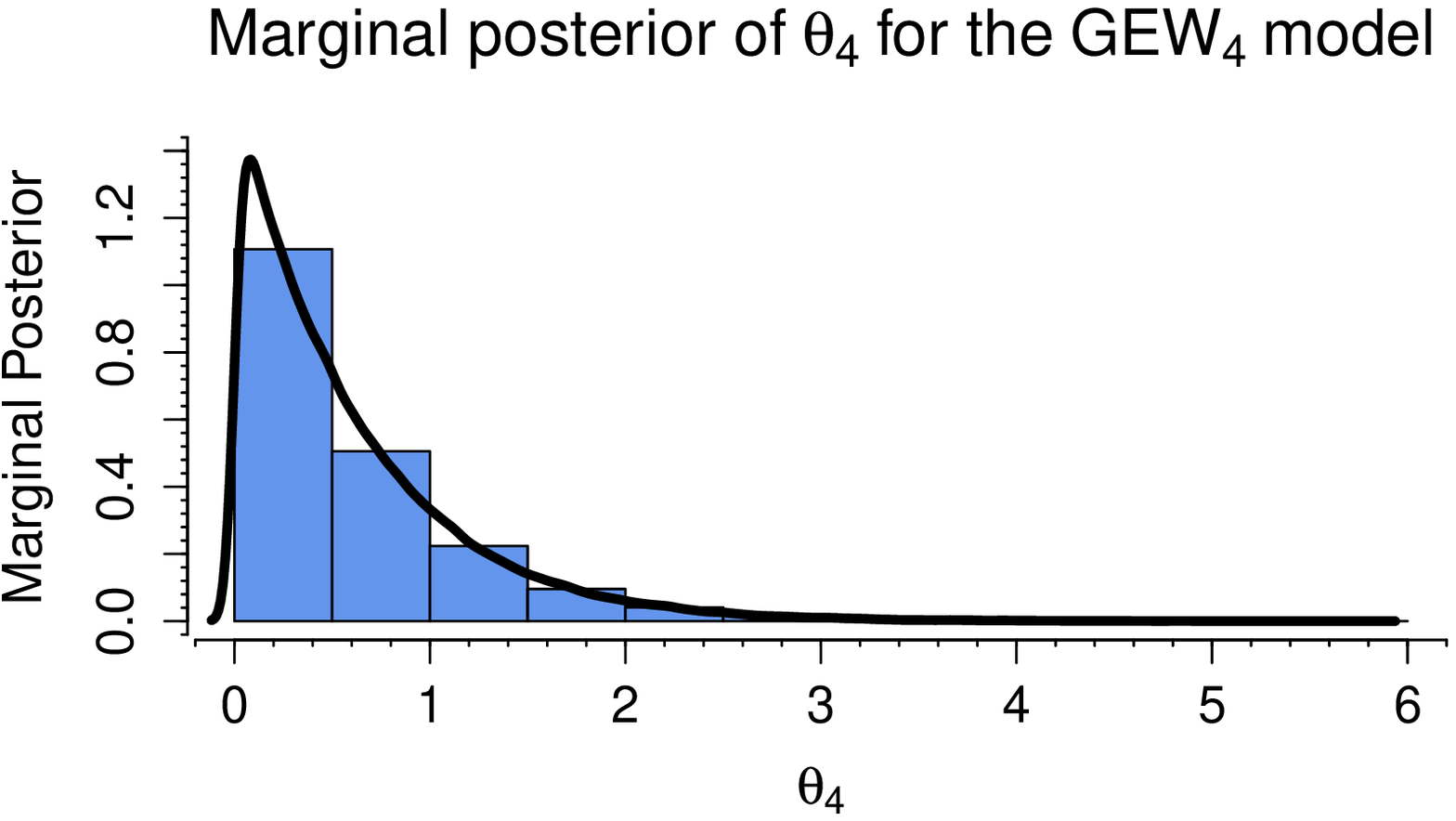}\includegraphics[width=4.8cm,height=2.2cm]{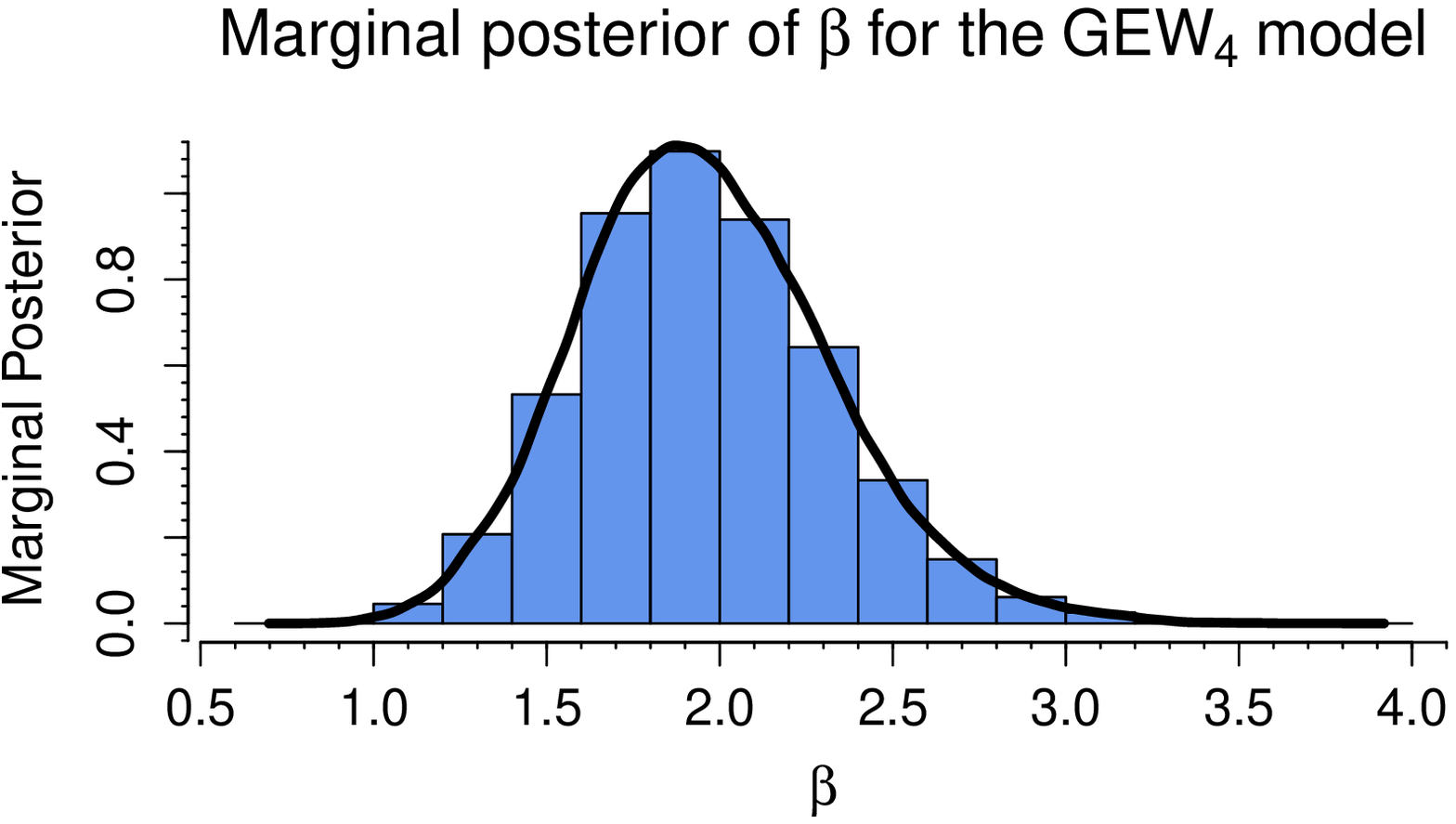}

\caption{Marginal posterior distributions for the GEW models.}
\label{Flo:GEW_MARGINAL}
\end{sidewaysfigure}

\begin{quotation}
Finally, the aim of these Bayesian ALT models is to obtain the predictive
reliability for the item being tested. The predictive reliability
at normal use stress levels is given by
\begin{equation}
R\left(x_{u}\left|\underline{x}\right.\right)=\int\int\int\int\int R\left(x_{u}\left|\theta_{1},\theta_{2},\theta_{3},\theta_{4},\beta\right.\right)\pi\left(\theta_{1},\theta_{2},\theta_{3},\theta_{4},\beta\left|\underline{x}\right.\right)d\theta_{1}d\theta_{2}d\theta_{3}d\theta_{4}d\beta,\label{eq:GEW_PREDREL}
\end{equation}
where $R\left(x_{u}\left|\theta_{1},\theta_{2},\theta_{3},\theta_{4},\beta\right.\right)$
is the Weibull reliability function at use stress levels $T_{u}$
and $S_{u}$. To evaluate $R\left(x_{u}\left|\underline{x}\right.\right)$,
the following is done:\\
1. Sample $\theta_{1},\theta_{2},\theta_{3},\theta_{4}$ and $\beta$
from the posterior $M$ times, where $M$ is a large number.\\
2. Calculate the integral in (\ref{eq:GEW_PREDREL}) by the Monte
Carlo average
\[
R\left(x_{u}\left|\underline{x}\right.\right)\approx\frac{1}{M}\sum_{m=1}^{M}R\left(x_{u}\left|\theta_{1}^{(m)},\theta_{2}^{(m)},\theta_{3}^{(m)},\theta_{4}^{(m)},\beta^{(m)}\right.\right),
\]
which is the expected reliability at time $x_{u}$, using the posterior
sample $\left\{ \theta_{1}^{(m)},\theta_{2}^{(m)},\theta_{3}^{(m)},\theta_{4}^{(m)},\beta^{(m)}\right\} ,$
$m=1,\ldots,M$. Table \ref{Flo:GEW_PREDRELT} and Figure \ref{Flo:GEW_PREDRELG}
show the predictive reliability of the models. Consistent with the
previous findings, the predictive reliability results for the $GEW_{1}$,
$GEW_{2,1}$, $GEW_{2,2}$, $GEW_{3}$ and $GEW_{4}$ models are comparable.
The GEW model is not sensitive to different choices of flat priors.
The models that make use of subjective priors produce significantly
higher predictive reliability results in this case, compared to the
models where flat priors are used. The use of subjective priors can
result in either an underestimation or overestimation of the predictive
reliability, compared to data driven results obtained from employing
flat priors, depending on the choice of hyperparameters for the gamma
priors in the $GEW_{2}$ model.
\begin{table}[H]
{\footnotesize{}\caption{Predictive reliability at use stress $T_{u}=350,S_{u}=0.3$.\label{Flo:GEW_PREDRELT}}
}{\footnotesize\par}
\centering{}{\footnotesize{}\medskip{}
}%
\begin{tabular}{|>{\centering}p{1cm}|>{\centering}m{1.7cm}|>{\centering}m{1.7cm}|>{\centering}m{1.7cm}|>{\centering}m{1.7cm}|>{\centering}m{1.7cm}|>{\centering}m{1.7cm}|>{\centering}m{1.7cm}|>{\centering}m{1.7cm}|}
\hline 
{\small{}Time} & {\small{}$GEW_{1}$} & {\small{}$GEW_{2,1}$} & {\small{}$GEW_{2,2}$} & {\small{}$GEW_{2,3}$} & {\small{}$GEW_{2,4}$} & {\small{}$GEW_{2,5}$} & {\small{}$GEW_{3}$} & {\small{}$GEW_{4}$}\tabularnewline
\hline 
{\small{}1} & {\small{}0.999987} & {\small{}0.999991} & {\small{}0.999990} & {\small{}0.999990} & {\small{}0.999998} & {\small{}1.000000} & {\small{}0.999990} & {\small{}0.999989}\tabularnewline
{\small{}2} & {\small{}0.999968} & {\small{}0.999976} & {\small{}0.999974} & {\small{}0.999974} & {\small{}0.999994} & {\small{}1.000000} & {\small{}0.999974} & {\small{}0.999973}\tabularnewline
{\small{}3} & {\small{}0.999943} & {\small{}0.999957} & {\small{}0.999954} & {\small{}0.999956} & {\small{}0.999989} & {\small{}1.000000} & {\small{}0.999954} & {\small{}0.999952}\tabularnewline
{\small{}$\vdots$} & {\small{}$\vdots$} & {\small{}$\vdots$} & {\small{}$\vdots$} & {\small{}$\vdots$} & {\small{}$\vdots$} &  & {\small{}$\vdots$} & {\small{}$\vdots$}\tabularnewline
{\small{}500} & {\small{}0.826214} & {\small{}0.845024} & {\small{}0.836377} & {\small{}0.894309} & {\small{}0.952155} & {\small{}0.997263} & {\small{}0.840580} & {\small{}0.837681}\tabularnewline
{\small{}501} & {\small{}0.825665} & {\small{}0.844519} & {\small{}0.835848} & {\small{}0.893973} & {\small{}0.951989} & {\small{}0.997251} & {\small{}0.840064} & {\small{}0.837159}\tabularnewline
{\small{}502} & {\small{}0.825115} & {\small{}0.844013} & {\small{}0.835319} & {\small{}0.893638} & {\small{}0.951823} & {\small{}0.997239} & {\small{}0.839547} & {\small{}0.836637}\tabularnewline
{\small{}$\vdots$} & {\small{}$\vdots$} & {\small{}$\vdots$} & {\small{}$\vdots$} & {\small{}$\vdots$} & {\small{}$\vdots$} &  & {\small{}$\vdots$} & {\small{}$\vdots$}\tabularnewline
{\small{}2000} & {\small{}0.179376} & {\small{}0.193196} & {\small{}0.183626} & {\small{}0.369133} & {\small{}0.586761} & {\small{}0.946555} & {\small{}0.186418} & {\small{}0.187901}\tabularnewline
{\small{}2001} & {\small{}0.179176} & {\small{}0.192981} & {\small{}0.183418} & {\small{}0.368884} & {\small{}0.586517} & {\small{}0.946499} & {\small{}0.186206} & {\small{}0.187692}\tabularnewline
{\small{}2002} & {\small{}0.178977} & {\small{}0.192766} & {\small{}0.183210} & {\small{}0.368636} & {\small{}0.586273} & {\small{}0.946442} & {\small{}0.185995} & {\small{}0.187482}\tabularnewline
\hline 
\end{tabular}{\footnotesize\par}
\end{table}
\begin{figure}[H]
\begin{centering}
\includegraphics[scale=1.1]{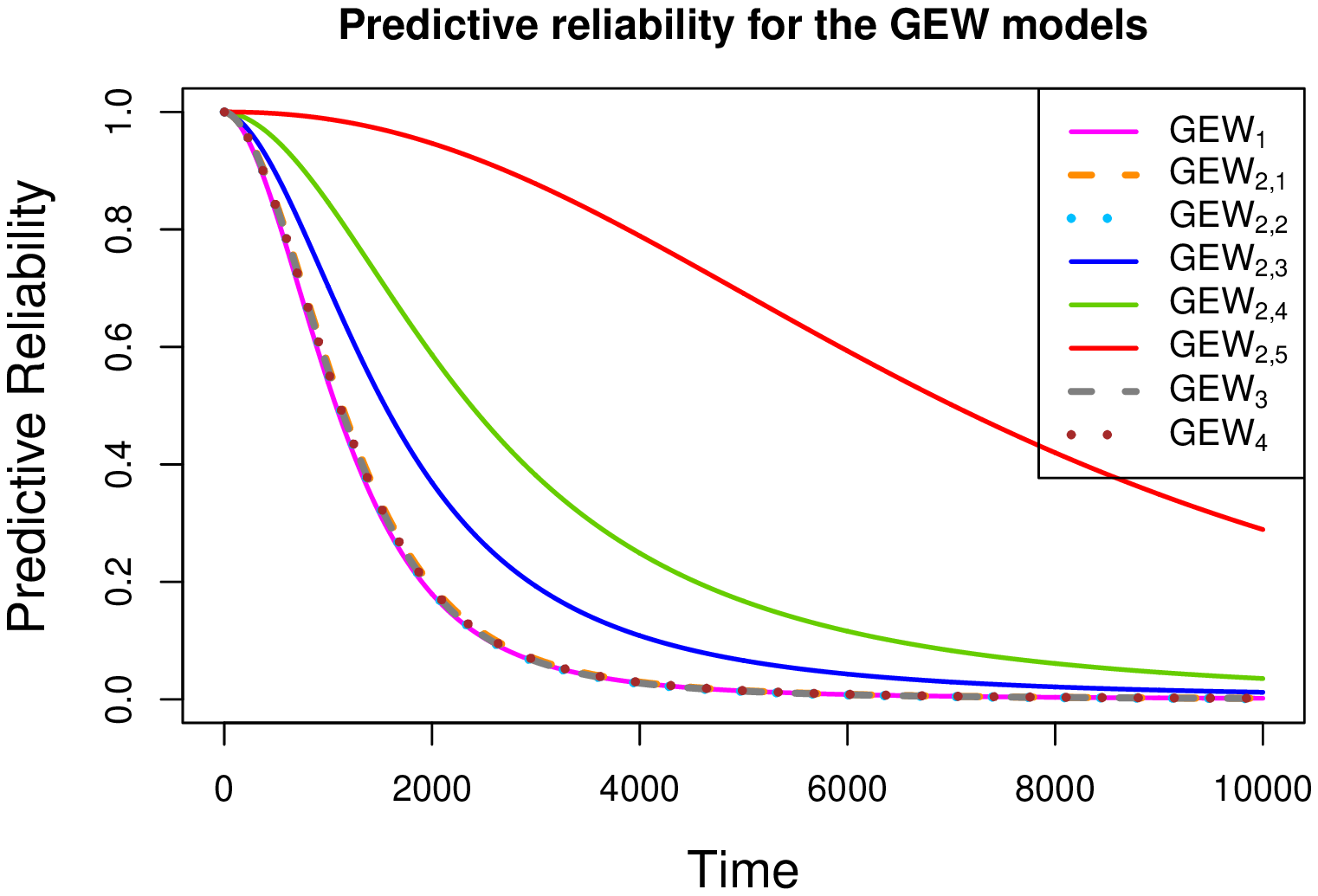}
\par\end{centering}
\caption{Predictive reliability at use stress $T_{u}=350,S_{u}=0.3$.}
\label{Flo:GEW_PREDRELG}
\end{figure}
\end{quotation}

\section{Conclusions\label{sec:GEW_CONCLUSIONS}}

In this paper, a Bayesian ALT model is presented, where the generalised
Eyring model is used as the TTF and lifetimes follow a Weibull distribution.
This is a versatile ALT model which allows for more than one accelerated
stressor, whereas many other TTFs only permit the use of a single
accelerated stressor. A general likelihood formulation is given, which
allows for complete samples, type-I censoring and type-II censoring.
Different GEW models are defined by means of the choice of prior distributions
for the model parameters. The full conditional posteriors for each
model are given and the log-concavity for each is also assessed. Under
no, to very lenient, conditions for the various models, the ARS method
within the Gibbs sampler can be used to obtain posterior samples from
the full conditional distributions. Alternatively, since these are
complex models, slice sampling can also be used to obtain posterior
samples. The models are applied to a data set concerning an electronics
epoxy packaging ALT, where temperature and relative humidity are used
as the accelerated stressors, and the results are compared.

Various choices for the hyper-parameters of the $GEW_{2}$ model are
considered. OpenBUGS is used to generate posterior samples for the
models, and the modified Gelman-Rubin statistic is calculated to assess
the convergence of the Markov chains. The fit of the models are compared
via the DIC. The $GEW_{1}$, $GEW_{2,1}$, $GEW_{2,2}$, $GEW_{3}$
and $GEW_{4}$ models show very similar DIC. The five models where
flat priors are imposed on the parameters produce alike summary statistics,
marginal posteriors and predictive reliability results. The use of
subjective priors in the $GEW_{2,3}$, $GEW_{2,4}$ and $GEW_{2,5}$
models lead to much different results, as noted in the summary statistics,
marginal posteriors and significantly higher predictive reliability.
Subjective priors can thus be utilised to adjust reliability estimates
if the researcher is of the opinion that the use of flat priors either
overestimates or underestimates the predictive reliability. This can
be achieved by means of the choice of hyper-parameters for the gamma
priors in the $GEW_{2}$ model. It may also be possible to adjust
a single parameter that is related to a specific stress, by means
of a subjective prior, if the researcher is of the opinion that the
effect of the accelerated stress is not correctly expressed by a model
where only flat priors are used. The use of priors with very small
variance in the $GEW_{2}$ model is not recommended and great caution
should be exercised by the reliability engineer if doing so.

\section*{References}
\begin{list}{}{}
\item Achcar, J. A. 1993. Approximate Bayesian methods: Some applications
in life testing problems. Brazilian Journal of Probability and Statistics,
7, 135--158. 
\item Achcar, J. A., \& Louzada-Neto, F. 1991. Accelerated life tests with
one stress variable: A Bayesian analysis of the Eyring model. Brazilian
Journal of Probability and Statistics, 5, 169--179. 
\item Ahmad, N. 1990. Bayesian and accelerated reliability analysis. Ph.D.
thesis, Aligarh Muslim University. 
\item Bagnoli, M., \& Bergstrom, T. 2005. Log-concave probability and its
applications. Economic Theory, 26, 445--469. 
\item Banerjee, A., \& Kundu, D. 2008. Inference based on type-II hybrid
censored data from a Weibull distribution. IEEE Transactions on Reliability,
57(2), 369--378. 
\item Barriga, G. D. C., Ho, L. L., \& Cancho, V. G. 2008. Planning accelerated
life tests under exponentiated- Weibull-Arrhenius model. International
Journal of Quality \& Reliability Management, 25(6), 636--653. 
\item Brooks, S. P., \& Gelman, A. 1998. General methods for monitoring
convergence of iterative simulations. Journal of Computational and
Graphical Statistics, 7(4), 434--455. 
\item Chaloner, K., \& Larntz, K. 1992. Bayesian design for accelerated
life testing. Journal of Statistical Planning and Inference, 33, 245--259. 
\item Dietrich, D. L., \& Mazzuchi, T. A. 1996. An alternative method of
analyzing multi-stress multi-level life and accelerated life tests.
Proceedings of the 1996 Annual Reliability and Maintainability Symposium,
90--96. 
\item Erkanli, A., \& Soyer, R. 2000. Simulation-based designs for accelerated
life tests. Journal of Statistical Planning and Inference, 90, 335--348. 
\item Escobar, L. A., \& Meeker, W. Q. 2006. A review of accelerated test
models. Statistical Science, 21(4), 552--577. 
\item Gilks, W. R., \& Wild, P. 1992. Adaptive rejection sampling for Gibbs
sampling. Applied Statistics, 41(2), 337--348. 
\item Kececioglu, D. B. 2002. Reliability and Life Testing Handbook. Lancaster:
DEStech Publications.
\item Kundu, D. 2008. Bayesian inference and life testing plan for the Weibull
distribution in presence of progressive censoring. Technometrics,
50(20), 144--154. 
\item Le�n, R. V., Ramachandran, R., Ashby, A. J., \& Thyagarajan, J. 2007.
Bayesian modeling of accelerated life tests with random effects. Journal
of Quality Technology, 39(1), 3--16. 
\item Mazzuchi, T. A., \& Soyer, R. 1992. A dynamic general linear model
for inference from accelerated life tests. Naval Research Logistics,
39, 757--773. 
\item Mazzuchi, T. A., Soyer, R., \& Vopatek, A. L. 1997. Linear Bayesian
inference for accelerated Weibull model. Lifetime Data Analysis, 3,
63--75. 
\item Mukhopadhyay, C., \& Roy, S. 2016. Bayesian accelerated life testing
under competing log-location-scale family of causes of failure. Computational
Statistics, 31, 89--119. 
\item Neal, R. M. 2003. Slice sampling. The Annals of Statistics, 31(3),
705--767. 
\item Nelson, W. B. 1990. Accelerated Testing: Statistical Models, Test
Plans, and Data Analysis. New York: Wiley. 
\item Pan, R. 2009. A Bayes approach to reliability prediction utilizing
data from accelerated life tests and field failure observations. Quality
and Reliability Engineering International, 25, 229--240. 
\item Papadopoulos, A. S., \& Tsokos, C. P. 1976. Bayesian analysis of the
Weibull failure model with unknown scale and shape parameters. Statistica,
36(4), 547--560. 
\item Perdona, G. C., \& Louzada-Neto, F. 2005. An approximate Bayesian
analysis for accelerated tests with log-non-linear stress-response
relationship. Journal of Mathematics and Statistics, 23(1), 59--65. 
\item ReliaSoft Corporation. 2015. Accelerated life testing reference. \url{http://www:synthesisplatform:net/references/Accelerated_Life_Testing_Reference:pdf.} 
\item Sha, N. 2018. Statistical inference for progressive stress accelerated
life testing with Birnbaum-Saunders distribution. Stats, 1, 189--203. 
\item Singpurwalla, N. D. 1971a. Inference from accelerated life tests when
observations are obtained from censored samples. Technometrics, 13(1),
161--170. 
\item Singpurwalla, N. D. 1971b. A problem in accelerated life testing.
Journal of the American Statistical Association, 66(336), 841--845. 
\item Singpurwalla, N. D. 1973. Inference from accelerated life tests using
Arrhenius type re-parameterizations. Technometrics, 15(2), 289--299. 
\item Singpurwalla, N. D., Castellino, V. C., \& Goldschen, D. Y. 1975.
Inference from accelerated life tests using Eyring type re-parameterizations.
Naval Research Logistics Quarterly, 22(2), 289--296. 
\item Soland, R. M. 1969. Bayesian analysis of the Weibull process with
unknown scale and shape parameters. IEEE Transactions on Reliability,
18(4), 181--184. 
\item Soyer, R. 2008. Accelerated life tests: Bayesian design. Encyclopedia
of Statistics in Quality and Reliability, 1, 12--20.
\item Soyer, R., Erkanli, A., \& Merrick, J. R. 2008. Accelerated life tests:
Bayesian models. Encyclopedia of Statistics in Quality and Reliability,
1, 20--30. 
\item Spiegelhalter, D. J., Best, N. G., Carlin, B. P., \& Van der Linde,
A. 2002. Bayesian measures of model complexity and fit. Journal of
the Royal Statistical Society, Series B, 64(4), 583--639. 
\item Thiraviam, A. R. 2010. Accelerated life testing of subsea equipment
under hydrostatic pressure. Ph.D. thesis, University of Central Florida. 
\item Tsokos, C. P. 1972. A Bayesian approach to reliability: Theory and
simulation. Proceedings of the 1972 Annual Reliability and Maintainability
Symposium, 78--87. 
\item Upadhyay, S. K., \& Mukherjee, B. 2010. Bayes analysis and comparison
of accelerated Weibull and accelerated Birnbaum-Saunders models. Communications
in Statistics - Theory and Methods, 39, 195--213. 
\item Van Dorp, J. R., \& Mazzuchi, T. A. 2004. A general Bayes exponential
inference model for accelerated life testing. Journal of Statistical
Planning and Inference, 119, 55--74. 
\item Van Dorp, J. R., \& Mazzuchi, T. A. 2005. A general Bayes Weibull
inference model for accelerated life testing. Reliability Engineering
and System Safety, 90(90), 140--147. 
\item Van Dorp, J. R., Mazzuchi, T. A., Fornell, G. E., \& Pollock, L. R.
1996. A Bayes approach to step-stress accelerated life testing. IEEE
Transactions on Reliability, 45(3), 491--498. 
\item Van Dorp, J. R., Mazzuchi, T. A., \& Garcidue�as, J. E. 2006. A comparison
of accelerated life testing designs within a single Bayesian inferential
framework. RAMS\textquoteright 06. Annual Reliability and Maintainability
Symposium, 208--214. 
\item Yuan, T., Liu, X., Ramadan, S. Z., \& Kuo, Y. 2014. Bayesian analysis
for accelerated life tests using a Dirichlet process Weibull mixture
model. IEEE Transactions on Reliability, 63(1), 58--67.
\end{list}

\section*{Appendix}

In this appendix, the log-concavity of the full conditional posteriors
for the GEW models in Section \ref{sec:PRIORS_POSTERIORS} will be
discussed. The log-concavity of the $GEW_{1}$ and $GEW_{2}$ models
are evaluated. The results can then easily be extended to the $GEW_{3}$
and $GEW_{4}$ models.

A twice-differentiable function $f(x)$ is said to be log-concave
if the second derivative of its natural log is non-positive on its
domain (see, for example, Bagnoli \& Bergstrom, 2005), thus if
\[
\frac{\partial^{2}\ln\left[f(x)\right]}{\partial x^{2}}\leq0\quad\forall x.
\]

\begin{thm}
The full conditional posterior distributions of the $GEW_{1}$ model
are all log-concave on their domains.
\end{thm}

\begin{proof}
For the $GEW_{1}$ model, the second derivatives of the natural logs
of the full conditional posteriors are determined as follows
\begin{eqnarray*}
\ell_{1,\theta_{1}} & = & \ln\left[\pi_{1}\left(\theta_{1}\left|\underline{x},\theta_{2},\theta_{3},\theta_{4},\beta\right.\right)\right]\\
 & = & -\theta_{1}\sum_{i=1}^{k}r_{i}-\sum_{i=1}^{k}\left(n_{i}-r_{i}\right)T_{i}\exp\left(-\theta_{1}-\frac{\theta_{2}}{T_{i}}-\theta_{3}V_{i}-\frac{\theta_{4}V_{i}}{T_{i}}\right)\tau_{i}^{\beta}\\
 &  & -\sum_{i=1}^{k}\sum_{j=1}^{r_{i}}T_{i}\exp\left(-\theta_{1}-\frac{\theta_{2}}{T_{i}}-\theta_{3}V_{i}-\frac{\theta_{4}V_{i}}{T_{i}}\right)x_{ij}^{\beta}\\
\frac{\partial\ell_{1,\theta_{1}}}{\partial\theta_{1}} & = & -\sum_{i=1}^{k}r_{i}+\sum_{i=1}^{k}\left(n_{i}-r_{i}\right)T_{i}\exp\left(-\theta_{1}-\frac{\theta_{2}}{T_{i}}-\theta_{3}V_{i}-\frac{\theta_{4}V_{i}}{T_{i}}\right)\tau_{i}^{\beta}\\
 &  & +\sum_{i=1}^{k}\sum_{j=1}^{r_{i}}T_{i}\exp\left(-\theta_{1}-\frac{\theta_{2}}{T_{i}}-\theta_{3}V_{i}-\frac{\theta_{4}V_{i}}{T_{i}}\right)x_{ij}^{\beta}\\
\frac{\partial^{2}\ell_{1,\theta_{1}}}{\partial\theta_{1}^{2}} & = & -\sum_{i=1}^{k}\left(n_{i}-r_{i}\right)T_{i}\exp\left(-\theta_{1}-\frac{\theta_{2}}{T_{i}}-\theta_{3}V_{i}-\frac{\theta_{4}V_{i}}{T_{i}}\right)\tau_{i}^{\beta}\\
 &  & -\sum_{i=1}^{k}\sum_{j=1}^{r_{i}}T_{i}\exp\left(-\theta_{1}-\frac{\theta_{2}}{T_{i}}-\theta_{3}V_{i}-\frac{\theta_{4}V_{i}}{T_{i}}\right)x_{ij}^{\beta}
\end{eqnarray*}
\begin{eqnarray*}
\ell_{1,\theta_{2}} & = & \ln\left[\pi_{1}\left(\theta_{2}\left|\underline{x},\theta_{1},\theta_{3},\theta_{4},\beta\right.\right)\right]\\
 & = & -\theta_{2}\sum_{i=1}^{k}\frac{r_{i}}{T_{i}}-\sum_{i=1}^{k}\left(n_{i}-r_{i}\right)T_{i}\exp\left(-\theta_{1}-\frac{\theta_{2}}{T_{i}}-\theta_{3}V_{i}-\frac{\theta_{4}V_{i}}{T_{i}}\right)\tau_{i}^{\beta}\\
 &  & -\sum_{i=1}^{k}\sum_{j=1}^{r_{i}}T_{i}\exp\left(-\theta_{1}-\frac{\theta_{2}}{T_{i}}-\theta_{3}V_{i}-\frac{\theta_{4}V_{i}}{T_{i}}\right)x_{ij}^{\beta}\\
\frac{\partial\ell_{1,\theta_{2}}}{\partial\theta_{2}} & = & -\sum_{i=1}^{k}\frac{r_{i}}{T_{i}}+\sum_{i=1}^{k}\left(n_{i}-r_{i}\right)\exp\left(-\theta_{1}-\frac{\theta_{2}}{T_{i}}-\theta_{3}V_{i}-\frac{\theta_{4}V_{i}}{T_{i}}\right)\tau_{i}^{\beta}\\
 &  & +\sum_{i=1}^{k}\sum_{j=1}^{r_{i}}\exp\left(-\theta_{1}-\frac{\theta_{2}}{T_{i}}-\theta_{3}V_{i}-\frac{\theta_{4}V_{i}}{T_{i}}\right)x_{ij}^{\beta}\\
\frac{\partial^{2}\ell_{1,\theta_{2}}}{\partial\theta_{2}^{2}} & = & -\sum_{i=1}^{k}\left(n_{i}-r_{i}\right)\frac{1}{T_{i}}\exp\left(-\theta_{1}-\frac{\theta_{2}}{T_{i}}-\theta_{3}V_{i}-\frac{\theta_{4}V_{i}}{T_{i}}\right)\tau_{i}^{\beta}\\
 &  & -\sum_{i=1}^{k}\sum_{j=1}^{r_{i}}\frac{1}{T_{i}}\exp\left(-\theta_{1}-\frac{\theta_{2}}{T_{i}}-\theta_{3}V_{i}-\frac{\theta_{4}V_{i}}{T_{i}}\right)x_{ij}^{\beta}
\end{eqnarray*}
\begin{eqnarray*}
\ell_{1,\theta_{3}} & = & \ln\left[\pi_{1}\left(\theta_{3}\left|\underline{x},\theta_{1},\theta_{2},\theta_{4},\beta\right.\right)\right]\\
 & = & -\theta_{3}\sum_{i=1}^{k}r_{i}V_{i}-\sum_{i=1}^{k}\left(n_{i}-r_{i}\right)T_{i}\exp\left(-\theta_{1}-\frac{\theta_{2}}{T_{i}}-\theta_{3}V_{i}-\frac{\theta_{4}V_{i}}{T_{i}}\right)\tau_{i}^{\beta}\\
 &  & -\sum_{i=1}^{k}\sum_{j=1}^{r_{i}}T_{i}\exp\left(-\theta_{1}-\frac{\theta_{2}}{T_{i}}-\theta_{3}V_{i}-\frac{\theta_{4}V_{i}}{T_{i}}\right)x_{ij}^{\beta}\\
\frac{\partial\ell_{1,\theta_{3}}}{\partial\theta_{3}} & = & -\sum_{i=1}^{k}r_{i}V_{i}+\sum_{i=1}^{k}\left(n_{i}-r_{i}\right)T_{i}V_{i}\exp\left(-\theta_{1}-\frac{\theta_{2}}{T_{i}}-\theta_{3}V_{i}-\frac{\theta_{4}V_{i}}{T_{i}}\right)\tau_{i}^{\beta}\\
 &  & +\sum_{i=1}^{k}\sum_{j=1}^{r_{i}}T_{i}V_{i}\exp\left(-\theta_{1}-\frac{\theta_{2}}{T_{i}}-\theta_{3}V_{i}-\frac{\theta_{4}V_{i}}{T_{i}}\right)x_{ij}^{\beta}\\
\frac{\partial^{2}\ell_{1,\theta_{3}}}{\partial\theta_{3}^{2}} & = & -\sum_{i=1}^{k}\left(n_{i}-r_{i}\right)T_{i}V_{i}^{2}\exp\left(-\theta_{1}-\frac{\theta_{2}}{T_{i}}-\theta_{3}V_{i}-\frac{\theta_{4}V_{i}}{T_{i}}\right)\tau_{i}^{\beta}\\
 &  & -\sum_{i=1}^{k}\sum_{j=1}^{r_{i}}T_{i}V_{i}^{2}\exp\left(-\theta_{1}-\frac{\theta_{2}}{T_{i}}-\theta_{3}V_{i}-\frac{\theta_{4}V_{i}}{T_{i}}\right)x_{ij}^{\beta}
\end{eqnarray*}
\begin{eqnarray*}
\ell_{1,\theta_{4}} & = & \ln\left[\pi_{1}\left(\theta_{4}\left|\underline{x},\theta_{1},\theta_{2},\theta_{3},\beta\right.\right)\right]\\
 & = & -\theta_{4}\sum_{i=1}^{k}\frac{r_{i}V_{i}}{T_{i}}-\sum_{i=1}^{k}\left(n_{i}-r_{i}\right)T_{i}\exp\left(-\theta_{1}-\frac{\theta_{2}}{T_{i}}-\theta_{3}V_{i}-\frac{\theta_{4}V_{i}}{T_{i}}\right)\tau_{i}^{\beta}\\
 &  & -\sum_{i=1}^{k}\sum_{j=1}^{r_{i}}T_{i}\exp\left(-\theta_{1}-\frac{\theta_{2}}{T_{i}}-\theta_{3}V_{i}-\frac{\theta_{4}V_{i}}{T_{i}}\right)x_{ij}^{\beta}\\
\frac{\partial\ell_{1,\theta_{4}}}{\partial\theta_{4}} & = & -\sum_{i=1}^{k}\frac{r_{i}V_{i}}{T_{i}}+\sum_{i=1}^{k}\left(n_{i}-r_{i}\right)V_{i}\exp\left(-\theta_{1}-\frac{\theta_{2}}{T_{i}}-\theta_{3}V_{i}-\frac{\theta_{4}V_{i}}{T_{i}}\right)\tau_{i}^{\beta}\\
 &  & +\sum_{i=1}^{k}\sum_{j=1}^{r_{i}}V_{i}\exp\left(-\theta_{1}-\frac{\theta_{2}}{T_{i}}-\theta_{3}V_{i}-\frac{\theta_{4}V_{i}}{T_{i}}\right)x_{ij}^{\beta}\\
\frac{\partial^{2}\ell_{1,\theta_{4}}}{\partial\theta_{4}^{2}} & = & -\sum_{i=1}^{k}\left(n_{i}-r_{i}\right)\frac{V_{i}^{2}}{T_{i}}\exp\left(-\theta_{1}-\frac{\theta_{2}}{T_{i}}-\theta_{3}V_{i}-\frac{\theta_{4}V_{i}}{T_{i}}\right)\tau_{i}^{\beta}\\
 &  & -\sum_{i=1}^{k}\sum_{j=1}^{r_{i}}\frac{V_{i}^{2}}{T_{i}}\exp\left(-\theta_{1}-\frac{\theta_{2}}{T_{i}}-\theta_{3}V_{i}-\frac{\theta_{4}V_{i}}{T_{i}}\right)x_{ij}^{\beta}
\end{eqnarray*}
\begin{eqnarray*}
\ell_{1,\beta} & = & \ln\left[\pi_{1}\left(\beta\left|\underline{x},\theta_{1},\theta_{2},\theta_{3},\theta_{4}\right.\right)\right]\\
 & = & \ln\left(\beta\right)\sum_{i=1}^{k}r_{i}-\sum_{i=1}^{k}\left(n_{i}-r_{i}\right)T_{i}\exp\left(-\theta_{1}-\frac{\theta_{2}}{T_{i}}-\theta_{3}V_{i}-\frac{\theta_{4}V_{i}}{T_{i}}\right)\tau_{i}^{\beta}\\
 &  & -\sum_{i=1}^{k}\sum_{j=1}^{r_{i}}T_{i}\exp\left(-\theta_{1}-\frac{\theta_{2}}{T_{i}}-\theta_{3}V_{i}-\frac{\theta_{4}V_{i}}{T_{i}}\right)x_{ij}^{\beta}+\left(\beta-1\right)\sum_{i=1}^{k}\sum_{j=1}^{r_{i}}\ln\left(x_{ij}\right)\\
\frac{\partial\ell_{1,\beta}}{\partial\beta} & = & \frac{1}{\beta}\sum_{i=1}^{k}r_{i}-\sum_{i=1}^{k}\left(n_{i}-r_{i}\right)T_{i}\exp\left(-\theta_{1}-\frac{\theta_{2}}{T_{i}}-\theta_{3}V_{i}-\frac{\theta_{4}V_{i}}{T_{i}}\right)\tau_{i}^{\beta}\ln\left(\tau_{i}\right)\\
 &  & -\sum_{i=1}^{k}\sum_{j=1}^{r_{i}}T_{i}\exp\left(-\theta_{1}-\frac{\theta_{2}}{T_{i}}-\theta_{3}V_{i}-\frac{\theta_{4}V_{i}}{T_{i}}\right)x_{ij}^{\beta}\ln\left(x_{ij}\right)+\sum_{i=1}^{k}\sum_{j=1}^{r_{i}}\ln\left(x_{ij}\right)\\
\frac{\partial^{2}\ell_{1,\beta}}{\partial\beta^{2}} & = & -\frac{1}{\beta^{2}}\sum_{i=1}^{k}r_{i}-\sum_{i=1}^{k}\left(n_{i}-r_{i}\right)T_{i}\exp\left(-\theta_{1}-\frac{\theta_{2}}{T_{i}}-\theta_{3}V_{i}-\frac{\theta_{4}V_{i}}{T_{i}}\right)\tau_{i}^{\beta}\ln^{2}\left(\tau_{i}\right)\\
 &  & -\sum_{i=1}^{k}\sum_{j=1}^{r_{i}}T_{i}\exp\left(-\theta_{1}-\frac{\theta_{2}}{T_{i}}-\theta_{3}V_{i}-\frac{\theta_{4}V_{i}}{T_{i}}\right)x_{ij}^{\beta}\ln^{2}\left(x_{ij}\right).
\end{eqnarray*}
Since $T_{i}\geq0$ (temperature measured in kelvin), $r_{i}\leq n_{i}$
(can not be more failures than items tested), $\tau_{i}\geq0$ (censoring
time), $x_{ij}\geq0$ (failure time), $\beta>0$ (shape parameter
of the Weibull distribution), $\exp\left(\cdot\right)\geq0$, $V_{i}^{2}\geq0$
and $\ln^{2}\left(\cdot\right)\geq0$, the full conditional posteriors
for the $GEW_{1}$ model are confirmed to be log-concave on their
domains.
\end{proof}
\begin{thm}
The full conditional posterior distributions of the $GEW_{2}$ model
are all log-concave on their domains, subject to $c_{10},c_{12},c_{14},c_{16},\sum_{i=1}^{k}r_{i}\geq1$.
\end{thm}

\begin{proof}
The second derivatives of the natural logs of the full conditional
posteriors for the $GEW_{2}$ model are given by
\begin{eqnarray*}
\ell_{2,\theta_{1}} & = & \ln\left[\pi_{2}\left(\theta_{1}\left|\underline{x},\theta_{2},\theta_{3},\theta_{4},\beta\right.\right)\right]\\
 & = & \left(c_{10}-1\right)\ln\left(\theta_{1}\right)-c_{11}\theta_{1}-\theta_{1}\sum_{i=1}^{k}r_{i}-\sum_{i=1}^{k}\left(n_{i}-r_{i}\right)T_{i}\exp\left(-\theta_{1}-\frac{\theta_{2}}{T_{i}}-\theta_{3}V_{i}-\frac{\theta_{4}V_{i}}{T_{i}}\right)\tau_{i}^{\beta}\\
 &  & -\sum_{i=1}^{k}\sum_{j=1}^{r_{i}}T_{i}\exp\left(-\theta_{1}-\frac{\theta_{2}}{T_{i}}-\theta_{3}V_{i}-\frac{\theta_{4}V_{i}}{T_{i}}\right)x_{ij}^{\beta}\\
\frac{\partial\ell_{2,\theta_{1}}}{\partial\theta_{1}} & = & \frac{c_{10}-1}{\theta_{1}}-c_{11}-\sum_{i=1}^{k}r_{i}+\sum_{i=1}^{k}\left(n_{i}-r_{i}\right)T_{i}\exp\left(-\theta_{1}-\frac{\theta_{2}}{T_{i}}-\theta_{3}V_{i}-\frac{\theta_{4}V_{i}}{T_{i}}\right)\tau_{i}^{\beta}\\
 &  & +\sum_{i=1}^{k}\sum_{j=1}^{r_{i}}T_{i}\exp\left(-\theta_{1}-\frac{\theta_{2}}{T_{i}}-\theta_{3}V_{i}-\frac{\theta_{4}V_{i}}{T_{i}}\right)x_{ij}^{\beta}\\
\frac{\partial^{2}\ell_{2,\theta_{1}}}{\partial\theta_{1}^{2}} & = & \frac{1-c_{10}}{\theta_{1}^{2}}-\sum_{i=1}^{k}\left(n_{i}-r_{i}\right)T_{i}\exp\left(-\theta_{1}-\frac{\theta_{2}}{T_{i}}-\theta_{3}V_{i}-\frac{\theta_{4}V_{i}}{T_{i}}\right)\tau_{i}^{\beta}\\
 &  & -\sum_{i=1}^{k}\sum_{j=1}^{r_{i}}T_{i}\exp\left(-\theta_{1}-\frac{\theta_{2}}{T_{i}}-\theta_{3}V_{i}-\frac{\theta_{4}V_{i}}{T_{i}}\right)x_{ij}^{\beta}
\end{eqnarray*}
\begin{eqnarray*}
\ell_{2,\theta_{2}} & = & \ln\left[\pi_{2}\left(\theta_{2}\left|\underline{x},\theta_{1},\theta_{3},\theta_{4},\beta\right.\right)\right]\\
 & = & \left(c_{12}-1\right)\ln\left(\theta_{2}\right)-c_{13}\theta_{2}-\theta_{2}\sum_{i=1}^{k}\frac{r_{i}}{T_{i}}-\sum_{i=1}^{k}\left(n_{i}-r_{i}\right)T_{i}\exp\left(-\theta_{1}-\frac{\theta_{2}}{T_{i}}-\theta_{3}V_{i}-\frac{\theta_{4}V_{i}}{T_{i}}\right)\tau_{i}^{\beta}\\
 &  & -\sum_{i=1}^{k}\sum_{j=1}^{r_{i}}T_{i}\exp\left(-\theta_{1}-\frac{\theta_{2}}{T_{i}}-\theta_{3}V_{i}-\frac{\theta_{4}V_{i}}{T_{i}}\right)x_{ij}^{\beta}\\
\frac{\partial\ell_{2,\theta_{2}}}{\partial\theta_{2}} & = & \frac{c_{12}-1}{\theta_{2}}-c_{13}-\sum_{i=1}^{k}\frac{r_{i}}{T_{i}}+\sum_{i=1}^{k}\left(n_{i}-r_{i}\right)\exp\left(-\theta_{1}-\frac{\theta_{2}}{T_{i}}-\theta_{3}V_{i}-\frac{\theta_{4}V_{i}}{T_{i}}\right)\tau_{i}^{\beta}\\
 &  & +\sum_{i=1}^{k}\sum_{j=1}^{r_{i}}\exp\left(-\theta_{1}-\frac{\theta_{2}}{T_{i}}-\theta_{3}V_{i}-\frac{\theta_{4}V_{i}}{T_{i}}\right)x_{ij}^{\beta}\\
\frac{\partial^{2}\ell_{2,\theta_{2}}}{\partial\theta_{2}^{2}} & = & \frac{1-c_{12}}{\theta_{2}^{2}}-\sum_{i=1}^{k}\left(n_{i}-r_{i}\right)\frac{1}{T_{i}}\exp\left(-\theta_{1}-\frac{\theta_{2}}{T_{i}}-\theta_{3}V_{i}-\frac{\theta_{4}V_{i}}{T_{i}}\right)\tau_{i}^{\beta}\\
 &  & -\sum_{i=1}^{k}\sum_{j=1}^{r_{i}}\frac{1}{T_{i}}\exp\left(-\theta_{1}-\frac{\theta_{2}}{T_{i}}-\theta_{3}V_{i}-\frac{\theta_{4}V_{i}}{T_{i}}\right)x_{ij}^{\beta}
\end{eqnarray*}
\begin{eqnarray*}
\ell_{2,\theta_{3}} & = & \ln\left[\pi_{2}\left(\theta_{3}\left|\underline{x},\theta_{1},\theta_{2},\theta_{4},\beta\right.\right)\right]\\
 & = & \left(c_{14}-1\right)\ln\left(\theta_{3}\right)-c_{15}\theta_{3}-\theta_{3}\sum_{i=1}^{k}r_{i}V_{i}-\sum_{i=1}^{k}\left(n_{i}-r_{i}\right)T_{i}\exp\left(-\theta_{1}-\frac{\theta_{2}}{T_{i}}-\theta_{3}V_{i}-\frac{\theta_{4}V_{i}}{T_{i}}\right)\tau_{i}^{\beta}\\
 &  & -\sum_{i=1}^{k}\sum_{j=1}^{r_{i}}T_{i}\exp\left(-\theta_{1}-\frac{\theta_{2}}{T_{i}}-\theta_{3}V_{i}-\frac{\theta_{4}V_{i}}{T_{i}}\right)x_{ij}^{\beta}\\
\frac{\partial\ell_{2,\theta_{3}}}{\partial\theta_{3}} & = & \frac{c_{14}-1}{\theta_{3}}-c_{15}-\sum_{i=1}^{k}r_{i}V_{i}+\sum_{i=1}^{k}\left(n_{i}-r_{i}\right)T_{i}V_{i}\exp\left(-\theta_{1}-\frac{\theta_{2}}{T_{i}}-\theta_{3}V_{i}-\frac{\theta_{4}V_{i}}{T_{i}}\right)\tau_{i}^{\beta}\\
 &  & +\sum_{i=1}^{k}\sum_{j=1}^{r_{i}}T_{i}V_{i}\exp\left(-\theta_{1}-\frac{\theta_{2}}{T_{i}}-\theta_{3}V_{i}-\frac{\theta_{4}V_{i}}{T_{i}}\right)x_{ij}^{\beta}\\
\frac{\partial^{2}\ell_{2,\theta_{3}}}{\partial\theta_{3}^{2}} & = & \frac{1-c_{14}}{\theta_{3}^{2}}-\sum_{i=1}^{k}\left(n_{i}-r_{i}\right)T_{i}V_{i}^{2}\exp\left(-\theta_{1}-\frac{\theta_{2}}{T_{i}}-\theta_{3}V_{i}-\frac{\theta_{4}V_{i}}{T_{i}}\right)\tau_{i}^{\beta}\\
 &  & -\sum_{i=1}^{k}\sum_{j=1}^{r_{i}}T_{i}V_{i}^{2}\exp\left(-\theta_{1}-\frac{\theta_{2}}{T_{i}}-\theta_{3}V_{i}-\frac{\theta_{4}V_{i}}{T_{i}}\right)x_{ij}^{\beta}
\end{eqnarray*}
\begin{eqnarray*}
\ell_{2,\theta_{4}} & = & \ln\left[\pi_{2}\left(\theta_{4}\left|\underline{x},\theta_{1},\theta_{2},\theta_{3},\beta\right.\right)\right]\\
 & = & \left(c_{16}-1\right)\ln\left(\theta_{4}\right)-c_{17}\theta_{4}-\theta_{4}\sum_{i=1}^{k}\frac{r_{i}V_{i}}{T_{i}}-\sum_{i=1}^{k}\left(n_{i}-r_{i}\right)T_{i}\exp\left(-\theta_{1}-\frac{\theta_{2}}{T_{i}}-\theta_{3}V_{i}-\frac{\theta_{4}V_{i}}{T_{i}}\right)\tau_{i}^{\beta}\\
 &  & -\sum_{i=1}^{k}\sum_{j=1}^{r_{i}}T_{i}\exp\left(-\theta_{1}-\frac{\theta_{2}}{T_{i}}-\theta_{3}V_{i}-\frac{\theta_{4}V_{i}}{T_{i}}\right)x_{ij}^{\beta}\\
\frac{\partial\ell_{2,\theta_{4}}}{\partial\theta_{4}} & = & \frac{c_{16}-1}{\theta_{4}}-c_{17}-\sum_{i=1}^{k}\frac{r_{i}V_{i}}{T_{i}}+\sum_{i=1}^{k}\left(n_{i}-r_{i}\right)V_{i}\exp\left(-\theta_{1}-\frac{\theta_{2}}{T_{i}}-\theta_{3}V_{i}-\frac{\theta_{4}V_{i}}{T_{i}}\right)\tau_{i}^{\beta}\\
 &  & +\sum_{i=1}^{k}\sum_{j=1}^{r_{i}}V_{i}\exp\left(-\theta_{1}-\frac{\theta_{2}}{T_{i}}-\theta_{3}V_{i}-\frac{\theta_{4}V_{i}}{T_{i}}\right)x_{ij}^{\beta}\\
\frac{\partial^{2}\ell_{2,\theta_{4}}}{\partial\theta_{4}^{2}} & = & \frac{1-c_{16}}{\theta_{4}^{2}}-\sum_{i=1}^{k}\left(n_{i}-r_{i}\right)\frac{V_{i}^{2}}{T_{i}}\exp\left(-\theta_{1}-\frac{\theta_{2}}{T_{i}}-\theta_{3}V_{i}-\frac{\theta_{4}V_{i}}{T_{i}}\right)\tau_{i}^{\beta}\\
 &  & -\sum_{i=1}^{k}\sum_{j=1}^{r_{i}}\frac{V_{i}^{2}}{T_{i}}\exp\left(-\theta_{1}-\frac{\theta_{2}}{T_{i}}-\theta_{3}V_{i}-\frac{\theta_{4}V_{i}}{T_{i}}\right)x_{ij}^{\beta}
\end{eqnarray*}
\begin{eqnarray*}
\ell_{2,\beta} & = & \ln\left[\pi_{2}\left(\beta\left|\underline{x},\theta_{1},\theta_{2},\theta_{3},\theta_{4}\right.\right)\right]\\
 & = & \left(c_{18}-1\right)\ln\left(\beta\right)-c_{19}\beta+\ln\left(\beta\right)\sum_{i=1}^{k}r_{i}-\sum_{i=1}^{k}\left(n_{i}-r_{i}\right)T_{i}\exp\left(-\theta_{1}-\frac{\theta_{2}}{T_{i}}-\theta_{3}V_{i}-\frac{\theta_{4}V_{i}}{T_{i}}\right)\tau_{i}^{\beta}\\
 &  & -\sum_{i=1}^{k}\sum_{j=1}^{r_{i}}T_{i}\exp\left(-\theta_{1}-\frac{\theta_{2}}{T_{i}}-\theta_{3}V_{i}-\frac{\theta_{4}V_{i}}{T_{i}}\right)x_{ij}^{\beta}+\left(\beta-1\right)\sum_{i=1}^{k}\sum_{j=1}^{r_{i}}\ln\left(x_{ij}\right)\\
\frac{\partial\ell_{2,\beta}}{\partial\beta} & = & \frac{c_{18}-1}{\beta}-c_{19}+\frac{1}{\beta}\sum_{i=1}^{k}r_{i}-\sum_{i=1}^{k}\left(n_{i}-r_{i}\right)T_{i}\exp\left(-\theta_{1}-\frac{\theta_{2}}{T_{i}}-\theta_{3}V_{i}-\frac{\theta_{4}V_{i}}{T_{i}}\right)\tau_{i}^{\beta}\ln\left(\tau_{i}\right)\\
 &  & -\sum_{i=1}^{k}\sum_{j=1}^{r_{i}}T_{i}\exp\left(-\theta_{1}-\frac{\theta_{2}}{T_{i}}-\theta_{3}V_{i}-\frac{\theta_{4}V_{i}}{T_{i}}\right)x_{ij}^{\beta}\ln\left(x_{ij}\right)+\sum_{i=1}^{k}\sum_{j=1}^{r_{i}}\ln\left(x_{ij}\right)\\
\frac{\partial^{2}\ell_{2,\beta}}{\partial\beta^{2}} & = & \frac{1-c_{18}}{\beta^{2}}-\frac{1}{\beta^{2}}\sum_{i=1}^{k}r_{i}-\sum_{i=1}^{k}\left(n_{i}-r_{i}\right)T_{i}\exp\left(-\theta_{1}-\frac{\theta_{2}}{T_{i}}-\theta_{3}V_{i}-\frac{\theta_{4}V_{i}}{T_{i}}\right)\tau_{i}^{\beta}\ln^{2}\left(\tau_{i}\right)\\
 &  & -\sum_{i=1}^{k}\sum_{j=1}^{r_{i}}T_{i}\exp\left(-\theta_{1}-\frac{\theta_{2}}{T_{i}}-\theta_{3}V_{i}-\frac{\theta_{4}V_{i}}{T_{i}}\right)x_{ij}^{\beta}\ln^{2}\left(x_{ij}\right).
\end{eqnarray*}
Since $T_{i}\geq0$ (temperature measured in kelvin), $r_{i}\leq n_{i}$
(can not be more failures than items tested), $\tau_{i}\geq0$ (censoring
time), $x_{ij}\geq0$ (failure time), $\beta>0$ (shape parameter
of the Weibull distribution), $\exp\left(\cdot\right)\geq0$, $V_{i}^{2}\geq0$
and $\ln^{2}\left(\cdot\right)\geq0$, the full conditional posteriors
for the $GEW_{2}$ model are log-concave on their domains, subject
to $c_{10},c_{12},c_{14},c_{16},\sum_{i=1}^{k}r_{i}\geq1$.
\end{proof}
\begin{thm}
The full conditional posterior distributions of the $GEW_{3}$ model
are all log-concave on their domains, when $\sum_{i=1}^{k}r_{i}\geq1$.
The same holds for the full conditional posterior distributions of
the $GEW_{4}$ model, when $c_{30},c_{32},c_{34},c_{36}\geq1$.
\end{thm}

\begin{proof}
Following the same reasoning as in Theorem 1 and Theorem 2, it is
easy to show that the full conditional posteriors for the $GEW_{3}$
and $GEW_{4}$ models are log-concave on their domains subject to
the conditions $\sum_{i=1}^{k}r_{i}\geq1$ and $c_{30},c_{32},c_{34},c_{36}\geq1$,
respectively.
\end{proof}

\end{document}